\documentclass[11pt,a4paper,aps]{article}
\usepackage{etex}
\usepackage{inputenc}
\usepackage[T1]{fontenc}
\usepackage{lmodern}					
\usepackage{textcomp}				
\usepackage{verbatim}			
\usepackage[english]{babel}
\usepackage[toc,page]{appendix} 		
\usepackage{cite}
\usepackage{amsmath,					
			amsthm,
			amsfonts,
			amssymb, 
			mathtools}
\usepackage{enumitem} 
\usepackage{todonotes}

\usepackage[all]{xy}
\usepackage{tikz-cd} 
\usetikzlibrary{tikzmark}
\usetikzlibrary{arrows}

\usepackage{graphicx}				
\usepackage{tabularx}				

\usepackage{pict2e}

\makeatletter
\DeclareRobustCommand{\loplus}{\mathbin{\mathpalette\dog@lsemi{+}}}
\DeclareRobustCommand{\lotimes}{\mathbin{\mathpalette\dog@lsemi{\times}}}
\DeclareRobustCommand{\roplus}{\mathbin{\mathpalette\dog@rsemi{+}}}
\DeclareRobustCommand{\rotimes}{\mathbin{\mathpalette\dog@rsemi{\times}}}

\newcommand{\dog@rsemi}[2]{\dog@semi{#1}{#2}{-90,90}}
\newcommand{\dog@lsemi}[2]{\dog@semi{#1}{#2}{270,90}}
\newcommand{\dog@semi}[3]{%
  \begingroup
  \sbox\z@{$\m@th#1#2$}%
  \setlength{\unitlength}{\dimexpr\ht\z@+\dp\z@\relax}%
  \makebox[\wd\z@]{\raisebox{-\dp\z@}{%
    \begin{picture}(1,1)
    \linethickness{\variable@rule{#1}}
    \roundcap
    \put(0.4,0.425){\makebox(0,0){\raisebox{\dp\z@}{$\m@th#1#2$}}}
    \put(0.4,0.4){\arc[#3]{0.5}}
    \end{picture}%
  }}%
  \endgroup
}
\newcommand{\variable@rule}[1]{%
  \fontdimen8  
  \ifx#1\displaystyle\textfont3\else
    \ifx#1\textstyle\textfont3\else
      \ifx#1\scriptstyle\scriptfont3\else
        \scriptscriptfont3\relax
  \fi\fi\fi
}
\makeatother

\usepackage{authblk}   

\definecolor{pur}{RGB}{255,102,102}

\usepackage[hmargin=2.2cm,vmargin=2.5cm]{geometry}		
\numberwithin{equation}{section}

\usepackage[colorlinks]{hyperref}			

\newcommand{\bigslant}[2]{{\left.\raisebox{.2em}{$#1$}\middle/\raisebox{-.2em}{$#2$}\right.}}

\newcommand{\dd}{\mathrm{d}}

\newtheorem{definition}{Definition}[section]
\newtheorem*{definition*}{Definition}
\newtheorem{theoreme}[definition]{Theorem}
\newtheorem*{theoreme*}{Theorem}
\newtheorem{proposition}[definition]{Proposition}

\newtheorem{lemme}[definition]{Lemma}
\newtheorem*{lemme*}{Lemma}

\newtheorem{corollaire}[definition]{Corollary}
\newtheorem*{corollaire*}{Corollary}

\theoremstyle{remark}
\newtheorem*{remarque}{Remark}
\newtheorem*{remarques}{Remarks}
\newtheorem{example}{Example}
\newtheorem*{prooftheo*}{Proof of the Theorem}

\newtheorem*{convention}{Convention}

\title{Tensor hierarchies and Leibniz algebras}

\author[]{Sylvain Lavau\thanks{lavau@math.univ-lyon1.fr}}
\affil[]{\small \emph{IMJ-PRG, Universit\'e Paris Diderot, Paris, France.}}
\date{}

\begin{document}

\maketitle

\abstract{
Tensor hierarchies are algebraic objects that emerge in gauging procedures in supergravity models, and that present a very deep and intricate relationship with Leibniz (or Loday) algebras. In this paper, we show that one can canonically associate a tensor hierarchy to any Loday algebra. By formalizing the construction that is performed in supergravity, we build this tensor hierarchy explicitly. We show that this tensor hierarchy can be canonically equipped with a differential graded Lie algebra structure that coincides with the one that is found in supergravity theories. 


}

\tableofcontents

\section{Introduction}

Tensor hierarchies form a class of objects found in supergravity theories, which emerge as compactifications of superstring theories \cite{Henning2005, Trigiante, Henning08ter}. These theoretical models have the particularity of being \emph{ungauged}, i.e. the 1-form fields are not even minimally coupled to any other fields. In the late nineties, a vast amount of new compactifications techniques was discovered, and this led to \emph{gauged} supergravities. One passes from ungauged to gauged models by promoting a suitable Lie subalgebra $\mathfrak{h}$ of the Lie algebra $\mathfrak{g}$ of global symmetries to being a \emph{gauge algebra}. 
The choice of such a Lie subalgebra is made through a $\mathfrak{h}$-covariant linear map $\Theta$, called the \emph{embedding tensor}, that relates the space $V$ in 1-form fields take values, to the Lie algebra of symmetries $\mathfrak{g}$. The existence of such a map $\Theta:V\to\mathfrak{g}$ is conditioned by a linear and a quadratic constraint. If they are satisfied, this map uniquely defines the Lie subalgebra $\mathfrak{h}$.

The main difference with classical gauge field theories is that the the 1-form fields do not take values in this Lie subalgebra, but in some $\mathfrak{g}$-module $V$. This implies that the transformations of the 2-form field strengths are not covariant in the usual sense. To solve this problem, physicists add a set of 2-form fields coupled to the 2-form field strengths to compensate for the lack of covariance. However, this lack of covariance is now transferred to the 3-form field strengths associated to the 2-form fields. This requires to add a set of 3-form fields to compensate this lack of covariance, etc. Thus, one can  build a hierarchy of $p$-form fields, for $p\geq2$, coupled to the $p$-form field strengths associated to the $p-1$-forms, to eventually ensure covariance of the Lagrangian. This \emph{tensor hierarchy} is a priori infinite but in supergravity theories, it is bounded by the dimension of space-time. The original presentation of the construction of the tensor hierarchy has been given in \cite{Henning08, Henning08bis}.

Recent developments towards the direction of giving a mathematical framework for this construction has been attempted \cite{palmkvist1, palmkvist, monpapier, kotovstrobleibniz, Palmer, Saemann}. In particular, the clearest construction of the tensor hierarchy up-to-date was performed in \cite{palmkvist1} by Jakob Palmkvist: his original `top-down' approach involves Borcherds algebras, which are generalizations of Kac-Moody algebras. Thus, he can equip the tensor hierarchy with a differential graded Lie algebra structure, that he calls a \emph{tensor hierarchy algebra}. He then applies his abstract construction to the general framework developed in supergravity models in \cite{palmkvist}.  
 Unfortunately, the relationship between this powerful construction and the step-by-step construction made by the physicist is not apparent.
 We propose in the present paper to provide a  `bottom-up' construction of the tensor hierarchy algebra, that sticks to the construction performed by physicists \cite{Henning08, Henning08bis}, and that moreover coincides with the structure defined by Jakob Palmkvist in \cite{palmkvist}. This proves that both the formal `top-down' and the computational `bottom-up' approaches to the tensor hierarchy algebra match what the construction which is done by physicists in gauging procedures in supergravity.

In the present paper, the construction of the tensor hierarchy algebra is based on the observation that gauging procedures in supergravity theories involve Leibniz algebras instead of mere Lie algebras.  A \emph{Leibniz} (or \emph{Loday}) \emph{algebra} is a generalization of a Lie algebra, where the product  is not necessarily skew-symmetric anymore. The Jacobi identity is modified in consequence: the corresponding identity -- the \emph{Leibniz identity} -- epitomizes the derivation property of the product on itself. This new notion was originally  defined by Jean-Louis Loday in the early nineties, see for example \cite{Loday}.  
The inner product of any Leibniz algebra can be split into its symmetric and its skew-symmetric part. The skew-symmetric bracket usually does not satisfy the Jacobi identity, emphasizing that Leibniz algebras are non-trivial generalizations of Lie algebras. 

The original motivation of this paper was actually to provide a systematic construction of the tensor hierarchies from a Leibniz algebras perspective. The link between embedding tensors and Leibniz algebras was known for a few years and has been investigated in \cite{kotovstrobleibniz}.  
The construction of a tensor hierarchy form the Leibniz algebra perspective crucially relies on the observation that, given a Lie algebra $\mathfrak{g}$ and a $\mathfrak{g}$-module $V$, an embedding tensor $\Theta:V\to\mathfrak{g}$ -- as defined in supergravity models -- induces a Leibniz algebra structure on the $V$. Reciprocally, any Leibniz algebra $V$ gives rise to an embedding tensor taking values in the quotient of $V$ by its center. More precisely, we show that tensor hierarchies are actually built from the data of a Lie algebra $\mathfrak{g}$, a $\mathfrak{g}$-module $V$ that carries a Leibniz algebra structure, and an embedding tensor $\Theta:V\to\mathfrak{g}$ that is compatible with both structures.  We call such triples of objects \emph{Lie-Leibniz triples}.

Most  of the paper is dedicated to show that any Lie-Leibniz triple induces a unique -- up to equivalence -- tensor hierarchy algebra.
The process of constructing the tensor hierarchy algebra canonically associated to a Leibniz algebra $V$ goes as follows: given a Lie-Leibniz triple $\mathcal{V}=(\mathfrak{g},V,\Theta)$, we first define what is called a \emph{stem associated to $V$}, which can be seen as the skeleton of the tensor hierarchy induced by $\mathcal{V}$. We then define the notion of `robust' stem, and we show 1. that a Lie-Leibniz triple induces a robust stem, which is unique up to morphisms of stems, and 2. that the robust stems associated to $\mathcal{V}$ are in one-to-one correspondence with the tensor hierarchy algebras associated to the same triple. Then, we deduce that the tensor hierarchy algebra induced by $\mathcal{V}$ are equivalent up to dgLa morphisms. This can be summarized in the following diagram:

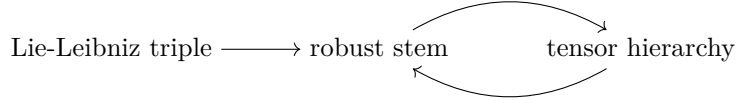
\begin{center}
\begin{tikzpicture}
\matrix(a)[matrix of math nodes, 
row sep=5em, column sep=3em, 
text height=1.5ex, text depth=0.25ex] 
{\text{\parbox{3cm}{\raggedleft Lie-Leibniz triple}}&\text{robust stem}& \text{tensor hierarchy}\\}; 
\path[->](a-1-1) edge node{} (a-1-2); 
\path[->](a-1-2) edge [bend left] node{} (a-1-3);
\path[->](a-1-3) edge [bend left] node{} (a-1-2);
\end{tikzpicture}
\end{center}

In the last few years, a renewal of interest in $L_\infty$-algebras has soared in the supergravity community, in relation with the formalism of tensor hierarchies \cite{Henning2018, Hohm, Cagnacci, Cederwall, Henning2019, monpapier}. These $L_\infty$-algebras would in some sense encode the field strength of the model. Until now, the only way one could determine the content of some given field strength would be to compute the Bianchi identity of lowest from-degree in which it appears. Hence, if there is a more convenient way of determining the content of the field strengths, without computing the Bianchi identities would be of some help for physicists. The tensor hierarchy algebra defined in Palmkvist's papers \cite{palmkvist, palmkvist1} and in the present paper are differential graded Lie algebras. Hence, by some mathematical subtleties, there may be some ways of defining the desired $L_\infty$-algebra from the tensor hierarchy algebra. This might be possible by using a result by Fiorenza and Manetti \cite{Fiorenza} (and found again later by Getzler \cite{getzler}) but this is still under investigation. That is why in this paper we will focus on the construction of the tensor hierarchy algebra, and we will not discuss how to induce a $L_\infty$-algebra from it.



The first part of this paper presents the mathematical tools that are used in the second part. Section \ref{tensorsection} presents the embedding tensor as defined by the physicists, whereas Section \ref{leibnizalgebras} provides the basic notions on Leibniz algebras, and their relationship with the embedding tensor, which leads to the crucial notion of Lie-Leibniz triple in Definition \ref{def:lieleibniz}.
 Then we discuss some important properties of Lie-Leibniz triples in Section \ref{someproperties}, before concluding the first part of the paper by elementary notions on graded geometry in Section \ref{graded}.

 Then, these mathematical tools are used through the entire second part, through some technical degree-juggling sessions. This part starts with the definition of tensor hierarchy algebras (see Definition \ref{def:tensoralgebra}), and a short explanation of the proof that any Lie-Leibniz triple induces a tensor hierarchy algebra.
 Section \ref{tower} introduces the notion of $i$-stems associated to Lie-Leibniz triples, and contains Theorem \ref{lemmestrand}, which is an existence statement for $i$-stems:
 
 \begin{theoreme*}
Let $i\in\mathbb{N}$ and let $\mathcal{V}=(\mathfrak{g},V,\Theta)$ be a Lie-Leibniz triple admitting a  $i$-stem $\mathcal{U}=(U,\delta,\pi,\mu)$. 
Then there exists a $(i+1)$-stem whose $i$-truncation is $\mathcal{U}$. 
\end{theoreme*}

  In Section \ref{london}, we discuss the notion of morphisms and equivalences of stems; in particular we give the definition of robust stems and we eventually obtain an important unicity result in Corollary \ref{inftystem}. Then, Section  
 \ref{sec:hierarchy} is devoted to the explanation of building a tensor hierarchy algebra from the data contained in the stem associated to a Lie-Leibniz triple. In particular it contains Theorem \ref{prop:tensorhierarchy} that shows that there is a one-to-one correspondence between tensor hierarchy algebras and robust stems associated to the same Lie-Leibniz triple. We conclude this section with the main result of this paper, Corollary \ref{ultimate}, that proves that a Lie-Leibniz triple induces a unique -- up to equivalence -- tensor hierarchy algebra:
 
 \begin{corollaire*}
A Lie-Leibniz triple induces -- up to equivalence -- a unique tensor hierarchy algebra.
\end{corollaire*}
 
 Eventually, in Section \ref{sectionexamples} we give precise and explicit constructions of the tensor hierarchy algebras associated to some Lie-Leibniz triples. 
  Appendices \ref{appendicite} and \ref{appendicite2} gather technical computations that appear in the proofs of Theorems \ref{lemmestrand} and \ref{prop:tensorhierarchy}, respectively. 


\subsection*{Acknowledgments}

I would like to thank Thomas Strobl and Henning Samtleben for invaluable discussions at the beginning of this work. I would also like to thank Jakob Palmkvist for its patience and precious ideas, Camille Laurent-Gengoux for its insights, Jim Stasheff for his help during the redaction of this manuscript, and Titouan Chary for fructuous talks. I also thank Peter Gothen and Andr\'e Oliveira for their kindness and hosting me for one year at CMUP, Porto.
This research was partially supported by CMUP (UID/MAT/00144/2013) funded
by FCT (Portugal) with national funds. 

\vspace{1cm}

\section{Mathematical background}

\subsection{The embedding tensor in supergravity theories}
\label{tensorsection}

We propose here a sketch of gauging procedures in supergravity, since it is a very intricate subject in itself, that has induced a flourishing literature on the subject \cite{Henning2005, Henning08, Henning08bis, Trigiante, Henning08ter}. 
Maximal supergravity theories in $3\leq D\leq7$ dimensions admit as Lie algebra of (global) symmetries the finite-dimensional real simple Lie algebras $\mathfrak{e}_{11-D\, (11-D)}$, which are the split real forms of the corresponding complex Lie algebra $\mathfrak{e}_{11-D}$ \cite{Trigiante}.  Gauging procedures in supergravity theories rely on promoting a Lie subalgebra of these real Lie algebras to the status of gauge algebra. Contrary to the usual gauge procedure as in classical field theories, in supergravity models the gauge fields do not take values in the Lie algebra of global symmetries, but rather in its standard representation, i.e. its smallest irreducible faithful representation \cite{Trigiante}.

\begin{convention}
In this paper, Lie algebras will always be real and finite-dimensional, but not necessarily semi-simple. We will use the gothic letters $\mathfrak{g}$ to denote what is considered as the Lie algebra of global symmetries, and $\mathfrak{h}$ for the Lie subalgebra of $\mathfrak{g}$ that is promoted to be a gauge algebra. 
We use the letter $V$ to denote the representation of $\mathfrak{g}$ in which the 1-form fields would take values. 
\end{convention}


The starting point for gauging in supergravity is to define a Lie subalgebra $\mathfrak{h}$ of $\mathfrak{g}$ that will be promoted to the status of \emph{gauge algebra}. Physicists define $\mathfrak{h}$ as the image of some linear mapping $\Theta:V\to\mathfrak{g}$ that satisfies some consistency conditions. 
Physicists call this map the \emph{embedding tensor}. 
The first condition is a contraint required by supersymmetry considerations and defines to which sub-module $T_\Theta$ of $V^*\otimes \mathfrak{g}$ the embedding tensor belongs, and the other condition is a closure constraint that ensures that $\mathfrak{h}\equiv\mathrm{Im}(\Theta)$ is indeed a Lie algebra.

 The fact that $\mathfrak{h}$ would be a gauge algebra means that it should induce a covariant derivative on the space of fields. Since the 1-form fields do not take values in $\mathfrak{h}$, 
 Physicists expect that the field strengths associated to the 1-form fields might not transform covariantly. Thus, for consistency, the theory might involve some 2-form fields that would be coupled to the field strengths. These 2-form fields would take values in some $\mathfrak{g}$-sub-module of $S^2(V)$. Since $\mathfrak{g}$ is semi-simple in supegravity theories, physicists can  decompose $S^2(V)$ into a sum a irreducible representations and check which one is allowed by supersymmetry. It turns out that in general there is a $\mathfrak{g}$-sub-module $\widetilde{W}\subset S^2(V)$ in which the 2-form fields cannot take values, because of these supersymmetry considerations. 
From this, Physicists require that the action of $\Theta$ on $\widetilde{W}$ gives zero. By decomposing $V^*\otimes\mathfrak{g}$ into a sum of irreducible $\mathfrak{g}$-modules, they can find which representation in particular has such a property.
 We call it $T_\Theta$. Hence physicists know that every element of $T_\Theta$ defines a subspace of $\mathfrak{g}$ that cancels $\widetilde{W}$.  This condition is called the \emph{linear}, or \emph{representation constraint}, and it defines $T_\Theta$ uniquely.


\begin{remarque}
In (half)-maximal supergravities, the representation $V$ in which the 1-form fields take values is faithful and the representation $T_\Theta$ to which belongs the embedding tensor is in general irreducible. In this paper, for more generality, we drop both conditions.
\end{remarque}



Now assume that such a  representation $T_\Theta\subset V^*\otimes \mathfrak{g}$ has been settled. We denote by $\rho^\Theta:\mathfrak{g}\to\mathrm{End}(T_\Theta)$ the linear mapping that encodes the representation of $\mathfrak{g}$ on $T_\Theta$.
\begin{align*}
\rho^\Theta:\hspace{0.2cm}\mathfrak{g}&\xrightarrow{\hspace*{1.7cm}} \hspace{0.2cm}\mathrm{End}(T_\Theta)\\
	a&\xmapsto{\hspace*{1.7cm}}\rho^\Theta_a:\Xi\mapsto\rho^\Theta_a(\Xi)
\end{align*}
However, the choice for the embedding tensor $\Theta$ has not been done yet. The condition that $\mathfrak{h}\equiv \mathrm{Im}(\Theta)$ is a Lie subalgebra of $\mathfrak{g}$ is equivalent to saying that the image of $\Theta$ in $\mathfrak{g}$ is closed under the Lie bracket. This will give the second constraint on the existence of the embedding tensor, that specify 
all possible elements of $T_\Theta$ that can be possibles candidates to being an embedding tensor. There might be many, and they can define different gauge algebras. 

As an immediate consequence of the definition of the action of $\mathfrak{g}$ on the tensor product $V^*\otimes\mathfrak{g}$, we have:
\begin{equation}\label{quadratic0}
\rho^\Theta_a(\Xi)(x)=\big[a,\Xi(x)\big]-\Xi\big(\rho_a(x)\big)
\end{equation}
for any $a\in\mathfrak{g}$,  $x\in V$ and $\Xi\in T_\Theta$. One can see that if $\rho^\Theta_a(\Xi)=0$ then $[a,\Xi(x)]\in\mathrm{Im}(\Xi)$. Thus, a sufficient condition for the subspace $\mathrm{Im}(\Theta)\subset\mathfrak{g}$ to be stable under the Lie bracket is: \begin{equation}\label{quadratic}
\rho^\Theta_a(\Theta)=0\hspace{1cm}\text{for any $a\in\mathrm{Im}(\Theta)$}\end{equation} Indeed, Equation \eqref{quadratic} applied to Equation \eqref{quadratic0} gives:
\begin{equation}\label{equivariance}
\big[\Theta(x),\Theta(y)\big]=\Theta\big(\rho_{\Theta(x)}(y)\big)
\end{equation}
for any $x,y\in V$. This implies that $\mathrm{Im}(\Theta)$ is a Lie subalgebra of $\mathfrak{g}$, that we denote by $\mathfrak{h}$ and that physicists call the \emph{gauge algebra}. The name is justified by the  analogy with the classical case, where gauge fields take values in the gauge algebra adjoint representation. Here, the 1-form fields are taking values in $V$ but they are associated to elements of  $\mathfrak{h}$ through the embedding tensor. 

Given that $V$ inherits a $\mathfrak{h}$-module structure induced by its $\mathfrak{g}$-module structure, Equation \eqref{equivariance} shows that the embedding tensor $\Theta$ is $\mathfrak{h}$-equivariant, with respect to the induced action of $\mathfrak{h}$ on $V$ and to the adjoint action of $\mathfrak{h}$ on itself. Indeed, writing $a$ instead of $\Theta(x)$ in Equation \eqref{equivariance}, we obtain:
\begin{equation}
\mathrm{ad}_{a}\big(\Theta(y)\big)=\Theta\big(\rho_{a}(y)\big)
\end{equation}
In supergravity theories, the gauge invariance condition \eqref{quadratic} is often written under the form of the `equivariance' condition \eqref{equivariance}, 
 and is called the \emph{quadratic}, or \emph{closure constraint}. This constraint is another formulation of the fact that $\mathrm{Im}(\Theta)$ is a Lie subalgebra of $\mathfrak{g}$. 

\begin{remarque}
Usually physicists do not specify which embedding tensor they want to pick up till the very end of their calculations, where they make a definite choice.  They say that it is a \emph{spurionic object}. They perform the calculations under the assumption that the generic embedding tensor $\Theta$ satisfies the linear and the quadratic constraint. This makes the computations easier, and more importantly, it prevents also a manifest symmetry breaking in the Lagrangian. When they fix a choice of embedding tensor, and thus of gauge algebra, at the very end of the computations, the symmetry is broken and they obtain the desired model. 
\end{remarque}

The specificity of supergravity theories is that even if the action of the gauge algebra on the 1-form gauge fields $A^a$ is a Lie algebra representation, the corresponding field strengths $F^a$ do not transform covariantly.
To get rid of this issue, as said before, a set of 2-form fields $B^I$ taking value in the complementary sub-space of $\widetilde{W}$ in $S^2(V)$ are added to the theory. They are coupled to the field strength $F^a$ through a Stuckelberg-like coupling, and their gauge transformation is defined to compensate the lack of covariance of $F^a$. However, the addition of these new fields $B^I$ necessarily implies to add their  corresponding field strengths, i.e. some 3-forms $H^I$. However, it turns out that they are not covariant either. One then adds to the model a set of 3-form fields living in a very specific $\mathfrak{g}$-module  so that the field strengths $H^I$ become covariant. The procedure continues and $p$-form fields are added to the theory until the dimension of space-time is reached. The set of all these fields form what is known as a \emph{tensor hierarchy}. If not for the dimension of space-time, nothing prevents this tower of fields to be infinite in full generality \cite{Henning08, Henning08bis}.

 To summarize, having a Lie algebra $\mathfrak{g}$ and a $\mathfrak{g}$-module $V$, the gauging procedure in supergravity theories consists of the following steps:

\begin{enumerate}
\item Defining a specific $\mathfrak{g}$-module $T_\Theta\subset V^*\otimes\mathfrak{g}$ to which all possible candidates as an embedding tensor would belong. This is done using the linear constraint, which picks up every elements of $V^*\otimes \mathfrak{g}$ that have a trivial action on some particular $\mathfrak{g}$-sub-module $\widetilde{W}$ of $S^2(V)$, that has been selected by supersymmetry considerations;
\item Setting a specific element $\Theta\in T_\Theta$ by the quadratic constraint \eqref{equivariance}, which ensures that $\mathfrak{h}\equiv\mathrm{Im}(\Theta)$ is a Lie subalgebra of $\mathfrak{g}$;
\item The action of $\mathfrak{h}$ on the 1-form fields $A^a$ induces an action on their corresponding field strengths $F^a$, but they do not transform covariantly. Then, physicists add a set of 2-form fields $B^I$ that take value in some specific sub-module so that the field strengths $F^a$ become covariant when they get coupled to the $B^I$'s;
\item If the field strengths associated to the $B^I$'s are not covariant, one should add 3-form fields, etc.
\end{enumerate}

Following the same kind of considerations for the 2-form fields, 3-form fields and so on, physicists manage to build a whole sequence of $\mathfrak{g}$-modules in which those higher fields take values. The construction  of this (possibly infinite) tower of spaces is automatic as soon as one has chosen the embedding tensor. The goal of this paper is to provide a detailed `bottom-up' approach to this construction, whereas the `top-down' approach was given in \cite{palmkvist1}, using Borcherds algebras. Both approaches seem to give the same result, as guessed in \cite{Cederwall2}. 


\subsection{Embedding tensors and Leibniz algebras} \label{leibnizalgebras}

At first, \emph{Leibniz} (or \emph{Loday}) \emph{algebras} have been introduced by Jean-Louis Loday in \cite{Loday} as a non commutative generalization of Lie algebras. In a Lie algebra, the Jacobi identity is equivalent to saying that the adjoint action is a derivation of the bracket. In a Leibniz algebra, we preserve this derivation property but we do not require the bracket to be skew symmetric anymore. More precisely:

\begin{definition}\label{defleibniz}
A 
 \emph{Leibniz algebra} is a finite dimensional real vector space $V$ equipped with a bilinear operation $\bullet $ satisfying the derivation property, or \emph{Leibniz identity}:
\begin{equation}\label{leibnizidentity}
x\bullet (y\bullet  z)=(x\bullet  y)\bullet  z+y\bullet (x\bullet  z)
\end{equation}
for all $x,y,z\in V$. A \emph{Leibniz algebra morphism} between  $(V,\bullet)$ and $(V',\bullet')$ is a linear mapping $\chi:V\to V'$ that is compatible with the respective products, that is:
\begin{equation}
\chi(x)\bullet'\chi(y)=\chi(x\bullet y)
\end{equation}
for every $x,y\in V$.
\end{definition}

\begin{convention}
In general, and for clarity of the exposition, we will often omit to write the couple $(V,\bullet)$ to designate a Leibniz algebra. In that case, we will assume that the Leibniz product $\bullet$ is implicitly attached to $V$.
\end{convention}

\begin{example}
A Lie algebra $(\mathfrak{g},[\,.\,,.\,])$ is a Leibniz algebra, with product $\bullet =[\,.\,,.\,]$. The Leibniz identity is nothing more than the Jacobi identity on $\mathfrak{g}$. Conversely a Leibniz algebra $(V,\bullet )$ is a Lie algebra when the product does not carry a symmetric part, that is: $x\bullet  x=0$ for every $x\in V$. Hence, the Leibniz identity \eqref{leibnizidentity} is a possible generalization of the Jacobi identity to non skew-symmetric brackets.
\end{example}


\begin{example}
Let $(A,\cdot)$ be an associative algebra equipped with an endomorphism $P:A\to A$ satisfying $P\big(P(x)\cdot y\big)=P(x)\cdot P(y)=P\big(x\cdot P(y)\big)$ for every $x,y\in A$ (for example when $P$ is an algebra morphism satisfying $P^2=P$). Then the product that is defined by:
\begin{equation}
x\bullet y\equiv P(x)\cdot y-y\cdot P(x)
\end{equation}
induces a Leibniz algebra structure on $A$. It is a Lie algebra precisely when $P=\mathrm{id}$.

\end{example}

\begin{example}\label{embeddingleibniz}
Let $\mathfrak{g}$ be a finite dimensional real Lie algebra and let $V$ be a $\mathfrak{g}$-module. 
 Let $\Theta:V\to\mathfrak{g}$ be an embedding tensor as in Section \ref{tensorsection}, i.e. a linear map from $V$ to $\mathfrak{g}$ satisfying the quadratic constraint \eqref{equivariance}. This implies that $\mathfrak{h}\equiv\mathrm{Im}(\Theta)$ is a Lie subalgebra of $\mathfrak{g}$. The action $\rho$ of $\mathfrak{g}$ on $V$ descends to an action of $\mathfrak{h}$ on $V$, that induces an action $\bullet$ of $V$ on $V$ itself by the following formula:
\begin{equation}\label{product}
x\bullet  y\equiv\rho_{\Theta(x)}(y)
\end{equation}
This action may not be symmetric nor skew-symmetric. By the equivariance condition \eqref{equivariance}, we deduce that $\Theta$ intertwines the product on $V$ and the Lie bracket on $\mathfrak{h}$:
\begin{equation}\label{integrability}
\Theta(x\bullet  y)=\big[\Theta(x),\Theta(y)\big]
\end{equation}
This is the most compact form of the quadratic constraint found in supergravity theories.  From Equations \eqref{product} and \eqref{integrability}, and from the fact that $V$ is a representation of the Lie algebra $\mathfrak{h}$, we deduce the following identity:
\begin{equation}\label{eq:leibnizidentity}
x\bullet (y\bullet  z)=(x\bullet  y)\bullet  z+y\bullet (x\bullet  z)
\end{equation}
In other words, the product $\bullet$ is a derivation of itself. This turns $V$ into a Leibniz algebra. Hence, Leibniz algebras emerge naturally through the gauging procedure in supergravity theories. 
\end{example}

We can  split the product $\bullet$ of a Leibniz algebra $V$ into its symmetric part $\{.\,,.\}$ and its skew-symmetric part $[\,.\,,.\,]$:
\begin{equation}\label{splitting}
x\bullet  y=[x,y]+\{x,y\}
\end{equation}
where
\begin{equation*}
[x,y]=\frac{1}{2}\big(x\bullet  y-y\bullet  x\big)
\hspace{0.7cm}\text{and}\hspace{0.7cm}
\{x,y\}=\frac{1}{2}\big(x\bullet  y+y\bullet  x\big)
\end{equation*}
for any $x,y\in V$.
As a consequence of Equation \eqref{leibnizidentity}, the Leibniz product is a derivation of both brackets. An important remark here is that even if the bracket $[\,.\,,.\,]$ is skew-symmetric, it does not satisfy the Jacobi identity since, using Equation \eqref{leibnizidentity}, we have:
\begin{equation}\label{jacobiator0}
\big[x,[y,z]\big]+\big[y,[z,x]\big]+\big[z,[x,y]\big]=\mathrm{Jac}(x,y,z)
\end{equation}
where the Jacobiator is defined by:
\begin{equation}\label{jacobiator}
\mathrm{Jac}(x,y,z)=-\frac{1}{3}\Big(\big\{x,[y,z]\big\}+\big\{y,[z,x]\big\}+\big\{z,[x,y]\big\}\Big)
\end{equation}
for every $x,y,z\in V$.  Hence the skew-symmetric bracket $[\,.\,,.\,]$ is not a Lie bracket. Since the Jacobi identity  does not close, one is tempted to lean on the notion of $L_\infty$-algebras to extend the bracket $[\,.\,,.\,]$. These are algebraic structures that generalize the notion of (differential graded) Lie algebras, by allowing the Jacobi identity to be satisfied only up to homotopy. Finding a $L_\infty$-algebra that extends the skew-symmetric bracket of $V$ is a topic that is currently under heavy investigation in the physics community, and that is actually interesting on its own.  Indeed, having a recipe to build a $L_\infty$-algebra lifting the skew-symmetric part of the product of any Leibniz algebra would be very important. This is currently under investigation.

Given a Leibniz algebra $V$, the subspace $\mathcal{I}\subset V$ generated by the set of elements of the form $\{x,x\}$ contains all symmetric elements of the form $\{x,y\}$, since they can always be written as a sum of squares.
Using Equation  \eqref{leibnizidentity}, one can check that $\mathcal{I}$ is an ideal of $V$ for the Leibniz product, i.e. $V\bullet \mathcal{I}\subset \mathcal{I}$, and that the action of $\mathcal{I}$ on $V$ is null. 

\begin{definition}\label{definitiontoutou}
The sub-space $\mathcal{I}$ of $V$ generated by elements of the form $\{x,x\}$ is an ideal called the \emph{ideal of squares of V}. An ideal of $V$ whose action is trivial is said \emph{central}. The union of all central ideals of $V$ is called the \emph{center of $V$}:
\begin{equation}
\mathcal{Z}=\Big\{x\in V\ \big|\ x\bullet  y=0\ \ \text{for all}\ \ y\in V\Big\}
\end{equation}
\end{definition}

We have seen  in Example \ref{embeddingleibniz} that the embedding tensor $\Theta$ defines a Leibniz product $\bullet$ on the $\mathfrak{g}$-module $V$. 
In the gauging procedure of maximal supergravity theories, the module $V$ is faithful, which implies that the map $\rho:\mathfrak{g}\to\mathrm{End}(V)$ is injective. Then by Equation \eqref{product}, we deduce that in that case the center of $V$ satisfies:
$$\mathcal{Z}=\mathrm{Ker}(\Theta)$$

Inspired by this result, we intend to define an embedding tensor from the data contained in a Leibniz algebra structure. Given a Leibniz algebra $(V,\bullet)$, we can define a particular vector space, noted $\mathfrak{h}_V$, by quotienting $V$ by the center $\mathcal{Z}$:
\begin{equation*}
\mathfrak{h}_V\equiv\bigslant{V}{\mathcal{Z}}
\end{equation*}
and we define $\Theta_V:V\to \mathfrak{h}_V$ to be the corresponding quotient map. 
The projection of the Leibniz product via $\Theta_V$ defines a bilinear product on $\mathfrak{h}_V$:
\begin{equation}\label{defbracket}
[a,b]_{\mathfrak{h}_V}\equiv \Theta_V\big(\widetilde{a}\bullet\widetilde{b}\big)
\end{equation}
for every $a,b\in\mathfrak{h}_V$, and where $\widetilde{a},\widetilde{b}$ are any pre-image of $a,b$ in $V$. Because $\mathcal{Z}$ is a central ideal, this bilinear product does not depend on the choice of pre-images.
This discussion can be summarized in the following diagram:
\begin{center}
\begin{tikzpicture}
\matrix(a)[matrix of math nodes, 
row sep=5em, column sep=7em, 
text height=1.5ex, text depth=0.25ex] 
{V\otimes V&V\\ 
\mathfrak{h}_V\otimes\mathfrak{h}_V&\mathfrak{h}_V\\}; 
\path[->](a-1-1) edge node[above]{$\bullet $}  (a-1-2); 
\path[->](a-1-1) edge node[left]{$\Theta_V\otimes\Theta_V$} (a-2-1);
\path[->](a-1-2) edge node[right]{$\Theta_V$} (a-2-2);
\path[->](a-2-1) edge node[above]{$[\,.\,,.\,]_{\mathfrak{h}_V}$}  (a-2-2); 
\end{tikzpicture}
\end{center}
From Equation \eqref{defbracket}, we have:
\begin{equation}\label{integrabilityleib}
\Theta_V(x\bullet y)=\Big[\Theta_V(x),\Theta_V(y)\Big]_{\mathfrak{h}_V}
\end{equation}
for every $x,y\in V$.
The Leibniz identity on V -- see Equation \eqref{leibnizidentity} -- and the fact that $\Theta_V$ is onto, implies that the Jacobi identity for $[\,.\,,.\,]_{\mathfrak{h}_V}$ is satisfied.
This turns $\big(\mathfrak{h}_V,[\,.\,,.\,]_{\mathfrak{h}_V}\big)$ into a Lie algebra, that we call the \emph{gauge algebra of $V$}. 
This analogy with the vocabulary from gauging procedures in supergravity  is not a coincidence. There is indeed a close relationship between tensors hierarchies and Leibniz algebras. As a first clue, one can notice the analogy between Equation \eqref{integrabilityleib} and Equation \eqref{integrability}.

Moreover, one can define an action $\rho$ of $\mathfrak{h}_V$ on $V$ by:
\begin{equation}\label{action}
\rho_{a}(x)\equiv\widetilde{a}\bullet  x
\end{equation}
for any $a\in\mathfrak{h}_V, x\in V$ and where $\widetilde{a}$ is any pre-image of $a$ in $V$. The action does not depend on the pre-image of $a$ since the component of $\widetilde{a}$ which is in $\mathcal{Z}$ acts trivially on $x$. This implies that for any $x,y\in V$,
we have:
\begin{equation}\label{productbis}
x\bullet y=\rho_{\Theta_V(x)}(y)
\end{equation}
Then, since $\Theta_V$ is onto, one can check that the Leibniz identity for the product $\bullet$ is equivalent to the fact that $\rho$ is a representation of $\mathfrak{h}_V$. Moreover, one can further notice the analogy between Equation \eqref{productbis} and Equation \eqref{product}. Hence we have shown that given a Leibniz algebra $V$, one can define a Lie algebra $\mathfrak{h}_V$ and a surjective map $\Theta_V:V\to\mathfrak{h}_V$ satisfying the linear constraint \eqref{product} and the quadratic constraint \eqref{integrability}. 




This strong relationship between the embedding tensor formalism and Leibniz algebras can be captured by the following object:

\begin{definition}\label{def:lieleibniz}
A \emph{Lie-Leibniz triple} is a triple $(\mathfrak{g},V,\Theta)$ where:
\begin{enumerate} 
\item $\mathfrak{g}$ is a real, finite dimensional, Lie algebra,
\item $V$ is a $\mathfrak{g}$-module equipped with a Leibniz algebra structure $\bullet$, and
\item $\Theta: V\to \mathfrak{g}$ is a linear mapping called the \emph{embedding tensor}, that satisfies two compatibility conditions.  The first one is the \emph{linear constraint}:
\begin{equation}\label{eq:compat}
x\bullet y=\rho_{\Theta(x)}(y)
\end{equation}
where $\rho:\mathfrak{g}\to \mathrm{End}(V)$ denotes the action of $\mathfrak{g}$ on $V$.
The second one is called the \emph{quadratic constraint}:
\begin{equation}\label{eq:equiv}
\Theta(x\bullet  y)=\big[\Theta(x),\Theta(y)\big]
\end{equation}
where $[\,.\,,.\,]$ is the Lie bracket on $\mathfrak{g}$.
\end{enumerate}
\end{definition}

The two conditions that $\Theta$ has to satisfy guarantee the compatibility between the Leibniz algebra structure on $V$, its $\mathfrak{g}$-module structure and the Lie bracket of $\mathfrak{g}$. 
The names of the constraints are justified because the symmetrization of the first equation gives the relationship between the symmetric bracket and the embedding tensor that is underlying the linear constraint of gauging procedures in supergravity theories. Moreover, using the first equation into the second one implies that $\Theta$ satisfies the quadratic constraint \eqref{equivariance}. Given these data, we deduce that $\mathfrak{h}\equiv\mathrm{Im}(\Theta)$ is a Lie subalgebra of $\mathfrak{g}$. In other words, we have mathematically encoded what is the embedding tensor. Moreover, Equation \eqref{eq:compat} implies that:
\begin{equation*}
\mathrm{Ker}(\Theta)\subset \mathcal{Z}
\end{equation*}
We have the equality when the representation of $\mathfrak{h}$ on $V$ is faithful.

\begin{example}
 Given a Leibniz algebra $V$, setting $\mathfrak{g}\equiv\mathfrak{h}_V$  and $\Theta\equiv\Theta_V$, 
we observe that the data $(\mathfrak{h}_V,V,\Theta_V)$ canonically define a Lie-Leibniz triple associated to $V$. This justifies that we call $\Theta_V$ the \emph{embedding tensor of $V$} and, as said before, we call $\mathfrak{h}_V$ the \emph{gauge algebra of $V$}. The Lie-Leibniz triple $(\mathfrak{h}_V,V,\Theta_V)$ satisfies every argument of Section \ref{tensorsection}, and in particular since $\mathrm{Ker}(\Theta_V)=\mathcal{Z}$, the action of $\mathfrak{h}_V$ on $V$ is faithful.
\end{example}

\begin{example} If $(\mathfrak{g},V,\Theta)$ is a Lie-Leibniz triple where the Leibniz algebra structure on $V$ is a mere Lie algebra structure, and where the embedding tensor $\Theta: V\to \mathfrak{g}$ is surjective, then the Lie-Leibniz triple $(\mathfrak{g},V,\Theta)$ is what we call a differential crossed module. 
Thus, the obstruction for a Lie-Leibniz triple $(\mathfrak{g},V, \Theta)$ where $V$ is a mere Lie algebra to be a differential crossed module comes from the fact that $\Theta$ might not be $\mathfrak{g}$-invariant, inducing a supplementary term in the usual condition:
\begin{equation}
\Theta\big(\rho_a(x)\big)=\big[a,\Theta(x)\big]-\rho^\Theta_a(\Theta)(x)
\end{equation}
for some $a\in\mathfrak{g}$ (and not in $\mathfrak{h}$) and $x\in V$.
\end{example}

\begin{example} In \cite{lodayseul}, Loday defines a \emph{pre-crossed module}: it is a triple $(\mathfrak{g}, V,\Theta)$ consisting of a Lie algebra $\mathfrak{g}$, a $\mathfrak{g}$-module $V$ and a $\mathfrak{g}$-equivariant linear map $\Theta:V\to \mathfrak{g}$. Then equips $V$ with a Leibniz algebra structure that is given by Equation \eqref{eq:compat}. The triples $(\mathfrak{g}, V,\Theta)$ defining pre-crossed modules are Lie-Leibniz triples.
\end{example}


More generally, we have the following result:
\begin{lemme}\label{propositionreve}
Let $(\mathfrak{g},V,\Theta)$ be a Lie-Leibniz triple. 
Then, there is a canonical surjective Lie algebra morphism $\varphi:\mathfrak{h}\to\mathfrak{h}_V$ that makes the following diagram commute:
\begin{center}
\begin{tikzpicture}
\matrix(a)[matrix of math nodes, 
row sep=5em, column sep=7em, 
text height=1.5ex, text depth=0.25ex] 
{V&\mathfrak{h}_V\\ 
\mathfrak{h}&\\}; 
\path[->>](a-2-1) edge node[above left]{$\varphi$} (a-1-2); 
\path[->](a-1-1) edge node[above]{$\Theta_V$} (a-1-2);
\path[->](a-1-1) edge node[left]{$\Theta$} (a-2-1);
\end{tikzpicture}
\end{center}
Moreover, $\varphi$ is an isomorphism if and only if $V$ is a faithful $\mathfrak{h}$-module.
\end{lemme}
\begin{proof}
First of all, notice that by Equation \eqref{eq:compat}, we have the inclusion $\mathrm{Ker}(\Theta)\subset\mathrm{Ker}(\Theta_V)=\mathcal{Z}$.
 Now let $a\in\mathfrak{h}$ and let $x$ be some preimage of $a$ in $V$. We then define $\varphi(a)\equiv\Theta_V(x)$. This definition does not depend on the choice of pre-image of $a$, because if we had chosen another one, say $y$, the difference $x-y$ would be in $\mathrm{Ker}(\Theta)$, which is a subspace of the center $\mathcal{Z}=\mathrm{Ker}(\Theta_V)$. Thus, the linear map $\varphi$ is well defined; it is also surjective since for every $u\in\mathfrak{h}_V$, there exists $x\in V$ such that $u=\Theta_V(x)$, so that the element $a=\Theta(x)$ is a preimage of $u$ by $\varphi$.
 Finally, the map $\varphi$ is a Lie algebra morphism because, for any $x,y\in V$:
\begin{align}
\varphi\Big(\big[\Theta(x),\Theta(y)\big]\Big)&=\varphi\big(\Theta(x\bullet y)\big)\\
&=\Theta_V(x\bullet y)\\
&=\big[\Theta_V(x),\Theta_V(y)\big]\\
&=\Big[\varphi\big(\Theta(x)\big),\varphi\big(\Theta(y)\big)\Big]
\end{align}

When $V$ is faithful, from Equation \eqref{eq:compat} we deduce the equality $\mathrm{Ker}(\Theta)=\mathcal{Z}$ which implies that $\mathrm{Ker}(\Theta)=\mathrm{Ker}(\Theta_V)$. Thus, for every $a\in\mathfrak{h}$ such that $\varphi(a)=0$, we deduce that $a=0$, for otherwise, the element would admit a preimage $x\in V$ that is not in the kernel of $\Theta$. But then $\Theta_V(x)\neq0$, which contradicts the fact that $\varphi\circ\Theta(x)=0$. This implies that $\varphi$ is injective, hence  bijective. The converse is immediate.
\end{proof}

To explore further this relationship, we need to define the notion of morphism of Lie-Leibniz triples:
\begin{definition}
Given two Lie Leibniz triples $\mathcal{V}\equiv(\mathfrak{g},V,\Theta)$ and $\overline{\mathcal{V}}\equiv(\overline{\mathfrak{g}},\overline{V},\overline{\Theta})$, a \emph{morphism between $\mathcal{V}$ and $\overline{\mathcal{V}}$} is a double $(\varphi,\chi)$ consisting of a Lie algebra morphism $\varphi:\mathfrak{g}\to \overline{\mathfrak{g}}$, and a Leibniz algebra morphism $\chi: V\to \overline{V}$, satisfying the following consistency conditions:
\begin{align}
\overline{\Theta}\circ\chi&=\varphi\circ\Theta\label{jfk}\\
\overline{\rho}_{\varphi(a)}\circ\chi&=\chi\circ\rho_a
\end{align}
for every $a\in \mathfrak{g}$, and where $\rho$ (resp. $\overline{\rho}$) denotes the action of $\mathfrak{g}$ (resp. $\overline{\mathfrak{g}}$) on $V$ (resp. $\overline{V}$). We say that $(\varphi,\chi)$ is an \emph{isomorphism of Lie-Leibniz triples} when both $\varphi$ and $\chi$ are isomorphisms in their respective categories.
\end{definition}

\begin{remarque}
We notice that Equation \eqref{jfk} implies that $\phi\big(\mathrm{Im}(\Theta)\big)\subset \mathrm{Im}\big(\overline{\Theta}\big)$.
\end{remarque}

%



Given a Leibniz algebra $(V,\bullet)$,
the Leibniz product can be seen as a map from $V$ to $\mathrm{Der}(V)$:
\begin{align*}
\bullet:\hspace{0.2cm}V&\xrightarrow{\hspace*{1.7cm}} \hspace{0.3cm}\mathrm{Der}(V)\\
	x\hspace{0.1cm}&\xmapsto{\hspace*{1.7cm}}x\,\bullet :y\mapsto x\bullet y
\end{align*}
The kernel of this map is precisely the center of $V$. 
Assume that the Leibniz product is such that there exists a Lie-Leibniz triple $(\mathfrak{g},V,\Theta)$ associated to $V$. By Equation \eqref{eq:compat} and Lemma \ref{propositionreve}, one deduces that the following diagram is commutative:
\begin{center}
\begin{tikzpicture}
\matrix(a)[matrix of math nodes, 
row sep=5em, column sep=7em, 
text height=1.5ex, text depth=0.25ex] 
{\mathfrak{h}&\\
V&\mathrm{Der}(V)\\ 
\mathfrak{h}_V&\\}; 
\path[->](a-1-1) edge node[above right]{$\rho$} (a-2-2); 
\path[->](a-2-1) edge node[above]{$\bullet$} (a-2-2);
\path[->](a-2-1) edge node[right]{$\Theta$} (a-1-1);
\path[->](a-1-1) edge [bend right=60] node[left]{$\varphi$} (a-3-1);
\path[->](a-2-1) edge node[right]{$\Theta_V$} (a-3-1);
\path[->](a-3-1) edge node[below right]{$\eta_V$} (a-2-2);
\end{tikzpicture}
\end{center}
where $\rho$ (resp. $\eta_V$) denotes the action of $\mathfrak{h}$ (resp. $\mathfrak{h}_V$) on $V$. In particular, we have the following equality:
\begin{equation}\label{wtfff}
\rho=\eta_V\circ\varphi
\end{equation}
Thus, Lemma \ref{propositionreve}, together with Equation \eqref{wtfff}, imply the following result:
\begin{proposition}\label{propositionreve2}
Let $(\mathfrak{g},V,\Theta)$ be a Lie-Leibniz triple, then there is a canonical morphism of Lie-Leibniz triples:
\begin{equation*}
\big(\varphi,\mathrm{id}_V\big):(\mathfrak{h},V,\Theta)\xrightarrow{\hspace*{1cm}}(\mathfrak{h}_V,V,\Theta_V)
\end{equation*}
where $\varphi$ is the map defined in Lemma \ref{propositionreve}. If $V$ is a faithful $\mathfrak{h}$-module, it is an isomorphism.
\end{proposition}

\subsection{The bud of a Lie-Leibniz triple}\label{someproperties}



Let us dwelve a bit further in the exploration of some properties of Lie-Leibniz triples.
The map $\bullet:V\to\mathrm{Der}(V)$ can be seen as an element of $V^*\otimes \mathrm{End}(V)$ on which $\mathfrak{g}$ acts. When $\mathfrak{g}$ is semi-simple, this is a completely reducible representation of $\mathfrak{g}$,
and we call $T_\bullet$ the representation to which $\bullet$ belongs. Let us write $\rho^\bullet:\mathfrak{g}\to\mathrm{End}(T_\bullet)$ for the map through which $\mathfrak{g}$ acts on $T_\bullet$.
 Now let $a\in\mathfrak{g}$, then we have a map $\rho^\bullet_a(\bullet):V\to\mathrm{End}(V)$ defined by:
\begin{equation}\label{bulleta}
\rho^\bullet_a(\bullet)(x) (y) \equiv \rho_a(x\bullet y)-x\bullet\rho_a(y)-\rho_a(x)\bullet y
\end{equation}
for every $x,y\in V$. Hence, the map $\rho^\bullet_a(\bullet)$ measures the obstruction of $\rho_a$ to be a derivation of the Leibniz product.
By using Equation \eqref{eq:compat}, one deduces that:
\begin{equation}\label{eq:deformationbullet}
\rho^\bullet_a(\bullet)(x) (y)=\rho_{\rho^\Theta_a(\Theta)(x)}(y)
\end{equation}
for every $x,y\in V$.  Moreover, Equation \eqref{eq:deformationbullet} implies that if $a\in \mathfrak{h}$, then the left hand side of Equation \eqref{bulleta} vanishes, which means that $\rho_a$ is a derivation of the Leibniz product at least when $a\in\mathfrak{h}$.  This was expected because in that case, the right hand side of Equation \eqref{bulleta} is the Leibniz identity \eqref{leibnizidentity}, when we write $a=\Theta(z)$.

Until now, we have only exploited the skew-symmetric part of the Leibniz product. It is now time to turn to the symmetric part.
Let $(V,\bullet)$ be a Leibniz algebra, then the symmetric bracket can be seen as a map $\{.\,,.\}:S^2(V)\to V$ whose image is the ideal of squares $\mathcal{I}$:
\begin{equation}
\big\{.\,,.\big\}(x\odot y)=\{x,y\}
\end{equation}
where $\odot$ represents the symmetric product. 
Hence, it can also be seen as an element of $S^2(V^*)\otimes V$, and as such it can be acted upon by $\mathfrak{g}$. We note $T_{\{,\}}$ the representation to which the symmetric bracket belongs, and $\rho^{\{,\}}:\mathfrak{g}\to \mathrm{End}\big(T_{\{,\}}\big)$ the corresponding map.
Using Equations \eqref{bulleta} and \eqref{eq:deformationbullet}, a short calculation shows that the symmetric bracket obeys the following equation:
\begin{equation}
\rho^{\{,\}}_a\big(\{.\,,.\}\big)(x\odot y)=\frac{1}{2}\Big(\rho_{\rho^\Theta_a(\Theta)(x)}(y)+\rho_{\rho^\Theta_a(\Theta)(y)}(x)\Big)
\end{equation}
for every $a\in\mathfrak{g}$. The fact that the embedding tensor is $\mathfrak{h}$-invariant implies that the right hand-side vanishes when $a\in\mathfrak{h}$. It means that the symmetric bracket is $\mathfrak{h}$-equivariant:
\begin{equation}
\rho_a\big(\{x,y\}\big)=\big\{\rho_a(x),y\big\}+\big\{x,\rho_a(y)\big\}
\end{equation}
for every $a\in\mathfrak{h}$.

Hence, the kernel of $\{.\,,.\}:S^2(V)\to V$ is a $\mathfrak{h}$-module, but not necessarily a $\mathfrak{g}$-module. Let us define $\widetilde{W}$ to be the biggest $\mathfrak{g}$-sub-module of $S^2(V)$ contained in $\mathrm{Ker}\big(\{.\,,.\}\big)$. 
 Then the symmetric bracket factors through the quotient $W\equiv\bigslant{S^2(V)}{\widetilde{W}}$:
\begin{center}
\begin{tikzpicture}
\matrix(a)[matrix of math nodes, 
row sep=5em, column sep=6em, 
text height=1.5ex, text depth=0.25ex] 
{&W\\ 
S^2(V)&V\\}; 
\path[->>](a-2-1) edge node[above left]{$\Pi_W$} (a-1-2); 
\path[->](a-1-2) edge node[right]{$\dd$} (a-2-2);
\path[->](a-2-1) edge node[above]{$\{.\,,.\}$} (a-2-2);
\end{tikzpicture}
\end{center}
This is a $\mathfrak{g}$-module, inheriting this structure from the quotient map $\Pi_W:S^{2}(V)\to W$ is the canonical quotient map. The map $\dd$ is the unique linear application $\dd: W\to \mathcal{I}$ such that:
\begin{equation}\label{defdd}
\{.\,,.\}=\dd\circ\Pi_W
\end{equation}
We can now give the following definition:
\begin{definition}\label{bud}
Given a Lie-Leibniz triple $\mathcal{V}=(\mathfrak{g},V,\Theta)$, the quotient $W$ of $S^2(V)$ by the biggest $\mathfrak{g}$-submodule of $\mathrm{Ker}\big(\{.\,,.\}\big)$ is called \emph{the bud of $\mathcal{V}$}. The map $\dd:W\to V$ defined in Equation \eqref{defdd} is called the \emph{collar of $\mathcal{V}$}.
\end{definition}

\begin{convention}
We have chosen botanical vocabulary because we will see in the following that the construction of the tensor hierarchy can be metaphorically seen as a plant which is growing, one step after another.
\end{convention}

\begin{remarques}
\begin{enumerate}
\item When the Lie-Leibniz triple is the canonical triple $(\mathfrak{h}_V,V,\Theta_V)$, the corresponding bud is obviously $W=\bigslant{S^2(V)}{\mathrm{Ker}\big(\{\,.\,,.\,\}\big)}$. For the clarity of exposition, in that case, we will speak of the bud (resp. collar) of $V$.
\item When $\mathfrak{g}$ is semi-simple, one can decompose $S^2(V)$ into irreducible representations of $\mathfrak{g}$, and see $W$ as the supplementary subspace of $\widetilde{W}$: $S^2(V)\simeq W\oplus \widetilde{W}$.
\item Notice that the definition of the bud and of the collar of $\mathcal{V}$ actually do not depend explicitely on the choice of embedding tensor $\Theta$. They only depend on the choice of the Lie algebra $\mathfrak{g}$ and of the Leibniz structure on $V$. Various embedding tensors satisfying Equation \eqref{eq:compat} will not interfere with the definition of the bud. This is consistent with the situation in supergravity where supersymmetry provides a constraint on the content of the fields, that translates into the choice of a $\mathfrak{g}$-sub-module of $S^2(V)$ that could not appear in the theory. Actually, the bud $W$ is precisely the space in which 2-form fields take values in supergravity models. 
\end{enumerate}
\end{remarques}

We deduce this simple but important result:

\begin{proposition}\label{thetainclusion}
Let $\mathcal{V}=(\mathfrak{g},V,\Theta)$ be a Lie-Leibniz triple, and let $\mathrm{d}$ be the collar of $\mathcal{V}$. Then, we have:
\begin{equation}\label{inclusionlol}
\Theta\circ\mathrm{d}=0
\end{equation}
\end{proposition}
\begin{proof}
By Equation \eqref{eq:equiv}, we deduce that $\Theta(\mathcal{I})=0$. Since $\mathrm{Im}(\mathrm{d})=\mathcal{I}$, we have the result.
\end{proof}
\begin{remarque}
Proposition \ref{thetainclusion} implies the following inclusion:
\begin{equation*}
\mathcal{I}\subset\mathrm{Ker}(\Theta)
\end{equation*}
but, usually, the kernel of the embedding tensor does not necessarily coincide with the ideal of squares. 
\end{remarque}

The vector space $W$ inherits the canonical quotient $\mathfrak{g}$-module structure induced by the action of $\mathfrak{g}$ on $S^2(V)$. Hence, it is the smallest quotient of $S^2(V)$ that has the property that $\{.\,,.\}$ factorizes through it and that is also a representation of $\mathfrak{g}$. In particular, the subspace $W$ cannot be smaller than $\bigslant{S^2(V)}{\mathrm{Ker}\big(\{.\,,.\}\big)}$, which happens when $\mathrm{Ker}\big(\{.\,,.\}\big)$ is a $\mathfrak{g}$-module as well. This is the case when $\mathfrak{g}=\mathfrak{h}_V$ for example.
From Proposition \ref{propositionreve2}, we deduce that following result:
\begin{proposition}\label{propositionreve3}
Let $\mathcal{V}=(\mathfrak{g},V,\Theta)$ be a Lie-Leibniz triple, let $W$ be the bud of $\mathcal{V}$, and let $\mathfrak{h}\equiv\mathrm{Im}(\Theta)$. Then there is canonical surjective linear mapping:
\begin{equation*}
\tau:W\xrightarrow{\hspace*{1cm}}\bigslant{S^2(V)}{\mathrm{Ker}\big(\{\,.\,,.\,\}\big)}
\end{equation*}
that makes the following diagram commute:
\begin{center}
\begin{tikzpicture}
\matrix(a)[matrix of math nodes, 
row sep=5em, column sep=6em, 
text height=1.5ex, text depth=0.25ex] 
{W&\\
S^2(V)&V\\ 
\bigslant{S^2(V)}{\mathrm{Ker}\big(\{\,.\,,.\,\}\big)}&\\}; 
\path[->](a-1-1) edge node[above right]{$\dd$} (a-2-2); 
\path[->](a-2-1) edge node[above]{$\{\,.\,,.\,\}$} (a-2-2);
\path[->](a-2-1) edge node[right]{$\Pi_W$} (a-1-1);
\path[->](a-1-1) edge [bend right=60] node[left]{$\tau$} (a-3-1);
\path[->](a-2-1) edge node[right]{} (a-3-1);
\path[->](a-3-1) edge node[below right]{$\{\,.\,,.\,\}$} (a-2-2);
\end{tikzpicture}
\end{center}
and which is compatible with the respective actions of $\mathfrak{h}$ and $\mathfrak{h}_V$, that is:
\begin{equation}\label{lalol}
\rho_{V,\varphi(a)}\big(\tau(\alpha)\big)=\tau\big(\rho_{a}(\alpha)\big)
\end{equation}
for every $a\in\mathfrak{h}$ and $\alpha\in W$, where $\varphi$ is the map defined in Lemma \ref{propositionreve}.
\end{proposition}

\begin{proof}
For clarity, let $\overline{W}\equiv\bigslant{S^2(V)}{\mathrm{Ker}\big(\big\{\,.\,,.\,\})}$. Let $\Pi_W:S^2(V)\to W$ (resp. $\Pi_{\overline{W}}:S^2(V)\to \overline{W}$) be the quotient map associated to $W$ (resp. $\overline{W}$). Then in particular:
\begin{equation*}
\mathrm{Ker}(\Pi_W)\subset \mathrm{Ker}(\Pi_{\overline{W}})=\mathrm{Ker}\big(\{\,.\,,.\,\}\big)
\end{equation*}
Let us define the map $\tau$ by:
\begin{equation}
\tau\big(\Pi_W(x\odot y)\big)=\Pi_{\overline{W}}(x\odot y)
\end{equation}
for every $x,y\in V$. It is well defined, because for any $\alpha\in W$ that admits two pre-images $u$ and $v$ in $S^2(V)$, we have $\Pi_W(u-v)=0$. Thus, $\Pi_{\overline{W}}(u-v)=0$, which implies that $\Pi_{\overline{W}}(u)=\tau(\alpha)=\Pi_{\overline{W}}(v)$. The map $\tau$ is obviously surjective, and by definition, we have $\{\,.\,,.\,\}\circ\tau=\dd$.

Now let $a\in\mathfrak{h}$, then for any $\alpha\in W$ and any pre-image $u\in S^2(V)$, we have:
\begin{align}
\tau\big(\rho_a(\alpha)\big)&=\tau\big(\rho_a(\Pi_W(u))\big)\\
&=\tau\big(\Pi_W(\rho_a(u))\big)\\
&=\Pi_{\overline{W}}\big(\rho_a(u)\big)\\
&=\Pi_{\overline{W}}\big(\rho_{V,\varphi(a)}(u)\big)\\
&=\rho_{V,\varphi(a)}\big(\Pi_{\overline{W}}(u)\big)\\
&=\rho_{V,\varphi(a)}\big(\tau(\alpha)\big)
\end{align}
which concludes the proof.
\end{proof}

Even if the map $\Pi_W$ is $\mathfrak{g}$-equivariant, the map  $\dd$ may not be. Rather, it transforms as $\{.\,,.\}$ in the representation $T_{\{,\}}$. At least, $\mathrm{d}$ is $\mathfrak{h}$-equivariant because the symmetric bracket is. There is even more: in supergravity theories, physicists show that the representation $T_{\{.,.\}}$ is the same as $T_\bullet$, and as $T_\Theta$. It implies that $\mathrm{d}$ and $\Theta$ transform in the same representation. This property has not been shown in the general case yet.

\subsection{Graded geometry}\label{graded}

The construction of the tensor hierarchies will involve many notions from graded algebra. We define a \emph{graded vector space $E$} as a family of vector spaces $E=(E_{k})_{k\in\mathbb{Z}}$. An element $x$ is said \emph{homogeneous of degree $i$} if $x\in E_{i}$. The degree of an homogeneous element $x$ is noted $|x|$. A \emph{commutative graded algebra} is a graded vector space $A=(A_k)_{k\in\mathbb{Z}}$ equipped with a product $\odot:A\otimes A\to A$ such that $$x\odot y=(-1)^{|x||y|}y\odot x$$ for every homogeneous elements $x,y\in A$. 
If the product is associative, successive products of multiple elements make sense whatever the order in which we perform the products. In that case, given $n$ homogeneous elements $x_1,\ldots,x_n\in A$, and a permutation $\sigma$ of $\{1,\ldots,n\}$, we define the \emph{Koszul sign of the permutation} (with respect to these elements) as the sign $\epsilon^{\sigma}_{x_1,\ldots,x_n}=\pm 1$ satisfying:
\begin{equation}
x_1\odot\ldots\odot x_n=\epsilon^\sigma_{x_1,\ldots,x_n}x_{\sigma(1)}\odot\ldots\odot x_{\sigma(2)}
\end{equation}

Given two graded vector spaces $E$ and $F$, a \emph{linear mapping between $E$ and $F$} is a family $\phi=(\phi_k)_{k\in\mathbb{Z}}$ of linear applications $\phi_k:E_k\to F_k$. For any two commutative graded algebras $A$ and $B$, a \emph{homomorphism from $A$ to $B$} is a degree 0 linear mapping $\Phi:A\to B$ that commutes with the respective products of $A$ and $B$:
$$\Phi(x\odot_A y)=\Phi(x)\odot_B\Phi(y)$$
for any $x,y\in A$. 
 A \emph{morphism from $E$ to $F$} is a (degree 0) graded commutative algebra homomorphism $\Phi:S(F^*)\to S(E^*)$. It induces a degree 0 linear mapping $\phi^*:F^*\to E^*$ whose dual map is a linear mapping $\phi:E\to F$ between the graded vector spaces $E$ and $F$.
  A \emph{function on $E$} is an element of the commutative graded algebra $S(E^*)=\bigoplus_{n\geq0}S^n(E^*)$, where $E^*$ is the graded vector space defined by the family of dual spaces $E^*=\big((E_k)^*\big)_{k\in\mathbb{Z}}$. In particular the degree of an element of $(E_k)^*$ is $-k$, i.e the opposite of the degree of elements of $E_k$. A function $f$ is said to be \emph{homogeneous of degree $p$} if $f\in S(E^*)_{p}$.

It is now time to define the central mathematical object related to tensor hierarchies:
\begin{definition}\label{def:Lie}
A \emph{graded Lie algebra} is a graded vector space $L=(L_{k})_{k\in\mathbb{Z}}$ equipped with a graded skew-symmetric bracket $[\,.\,,.\,]:L_{k}\otimes L_l\to L_{k+l}$ that satisfies the \emph{graded Jacobi identity}:
\begin{equation}
\big[x,[y,z]\big]=\big[[x,y],z\big]+(-1)^{|x||y|}\big[y,[x,z]\big]
\end{equation}
for any $x,y,z\in L$.

A \emph{differential graded Lie algebra} is a graded Lie algebra $\big(L,[\,.\,,.\,]\big)$ that admits a differential $\partial=\big(\partial_k:L_{k-1}\to L_k\big)_{k\in \mathbb{Z}}$ which is a derivation of the bracket:
\begin{equation}\label{eq:compatibility}
\partial\big([x,y]\big)=\big[\partial(x),y\big]+(-1)^{|x|}\big[x,\partial(y)\big]
\end{equation}
for any $x,y\in L$.
If $L$ is negatively graded, i.e. if $L=\bigoplus_{k\geq0} L_{-k}$, we call the \emph{depth of $L$} the unique element $i\in\mathbb{N}\cup\{\infty\}$ such that $L=\bigoplus_{0\leq k<i+1}L_{-k}$, and the sequence $(L_{-k})_{0\leq k<i+1}$ does not converge to the zero vector space.
\end{definition}
\begin{remarque}
The depth of a graded Lie algebra is either an integer, and in this case 
$L_{-i}\neq0$, or it is infinite and then, whatever the rank $n$ we chose, there is always some $k>n$ such that $L_{-k}\neq0$.
\end{remarque}
\begin{example}
Let $\mathfrak{g}$ be a Lie algebra. The Chevalley-Eilenberg algebra $\mathrm{CE}(\mathfrak{g})$ is the graded commutative algebra:
\begin{equation}
\wedge^\bullet\mathfrak{g}^*\equiv\mathbb{R}\oplus\mathfrak{g}^*\oplus\wedge^2\mathfrak{g}^*\oplus\ldots
\end{equation}
The Chevalley-Eilenberg differential $\mathrm{d}_{\mathrm{CE}}$ acts naturally on this algebra. There exist also two kinds of derivations acting on $\mathrm{CE}(\mathfrak{g})$: the inner contractions $\iota_x$ and the Lie derivatives $\mathcal{L}_x\equiv[\mathrm{d}_{\mathrm{CE}},\iota_x]$, for every $x\in\mathfrak{g}$. Here the bracket is the bracket of operators in the space of derivations of $\mathrm{CE}(\mathfrak{g})$. We define the differential graded Lie algebra $\mathfrak{inn}(\mathfrak{g})$ of inner derivations of $\mathfrak{g}$ by the following:
\begin{itemize}
\item elements of degree $-1$ are the contractions;
\item elements of degree $0$ are the Lie derivatives;
\item the differential $\partial: \mathfrak{inn}(\mathfrak{g})_{-1}\to \mathfrak{inn}(\mathfrak{g})_0$ satisfies:
\begin{equation}
\partial=[\mathrm{d}_{\mathrm{CE}},.\,]
\end{equation}
\item and the bracket is defined by:
\begin{equation}
[\mathcal{L}_x,\mathcal{L}_y]=\mathcal{L}_{[x,y]}\hspace{1cm}[\mathcal{L}_x,\iota_y]=\iota_{[x,y]}\hspace{1cm}[\iota_x,\iota_y]=0
\end{equation}
for every $x,y\in\mathfrak{g}$. 
\end{itemize}
\end{example}

There is another formulation of (differential) graded Lie algebras using the notion of \emph{differential graded manifolds}. First, a \emph{graded manifold $\mathcal{M}=(E,M)$} is a sheaf $\mathcal{C}^{\infty}_{\mathcal{M}}$ of graded algebras over a smooth manifold $M$ that is called the \emph{base}, such that for every open set $U\subset M$, $\mathcal{C}^{\infty}_{\mathcal{M}}(U)\simeq\mathcal{C}^{\infty}(U)\otimes S(E^*)$, where $E$ is a graded vector space called the \emph{fiber}. A \emph{morphism between the graded manifolds $\mathcal{M}$ and $\mathcal{N}$} is a family $\Phi=(\phi_{U})_{U\subset M}$ of graded algebra homomorphisms $\phi_U:\mathcal{C}^{\infty}_{\mathcal{N}}(U)\to\mathcal{C}^{\infty}_{\mathcal{M}}(U)$. We define \emph{vector fields on $\mathcal{M}$} as sections to the (graded) vector space of derivations of $\mathcal{C}^{\infty}_{\mathcal{M}}$. If the base manifold $M$ is reduced to a point,  we say that the graded manifold $\mathcal{M}$ is \emph{pointed}, i.e. it is reduced to the graded vector space $E$. In that case a vector field $X$ can be identified with an element of $S(E^*)\otimes E$. A vector field $X$ on $\mathcal{M}$ is said to be \emph{of arity $n$} if for any function $f\in S^k(E^*)$, we have $X(f)\in S^{k+n}(E^*)$. Obviously we can decompose a graded vector field by its components of various arities, but they should not be confused with the degree of the vector field.

\begin{definition}
A \emph{differential graded manifold} is a graded manifold $\mathcal{M}$ equipped with a degree $+1$ vector field $Q$ satisfying $[Q,Q]=0$.
\end{definition}

Given a graded vector space $E$, the \emph{suspension of $E$} is the graded vector space $sE=(sE)_{k\in\mathbb{Z}}$ defined as: 
\begin{equation*}
(sE)_k=E_{k-1}
\end{equation*} In other words, the suspension of a graded vector space is the same vector space, but with all degrees shifted by $+1$. Consequently, the degrees of dual elements are shifted by $-1$:
\begin{equation*}
(sE)^*=s^{-1}(E^*)
\end{equation*}
Also, every graded symmetric object becomes graded skew-symmetric (and vis-versa). Hence a function $f\in S^n(E^*)$ of degree $p$ is transformed into a function $sf\in \wedge^n\big((sE)^*\big)$ of degree $p-n$. In particular, given a linear application $F:S^{2}(E)\to E$ of degree $p$, it suspension $sF:\wedge^2(sE)\to sE$ has degree $p-1$ (precise formulas are given in \cite{Fiorenza}). The suspension isomorphism admits a reverse map which is called the \emph{desuspension} and which is noted $s^{-1}$. The desuspension map satisfies the following identity: 
\begin{equation*}(s^{-1}E)_k=E_{k+1}\end{equation*}

We now define the pairing between a graded vector space $E$ and its dual $E^*$. For any two homogeneous elements $u\in E$ and $\alpha\in E^*$, the pairing:
\begin{equation}
\langle\alpha,u\rangle_E\equiv\alpha(u)
\end{equation}
is non vanishing if an only if $|\alpha|=-|u|$ (recall that the `absolute value' denotes the degree and thus can be negative). The vector space that is written at the bottom of the right angle labels the space to which the right element belongs, here $u\in E$. Set $\iota_u$ to be the degree $|u|$  constant vector field on $E$ satisfying:
\begin{equation}\label{identification}
\iota_u(\alpha)\equiv\langle\alpha, u\rangle_E
\end{equation}
It is an interior product. 
The pairing is symmetric:
\begin{equation}\label{pairing}
\langle\alpha,u\rangle_E=\langle u,\alpha\rangle_{E^*}
\end{equation}
where here one considers that $u\in E^{**}\simeq E$. If $\alpha$ is an element of $(E^*)^{\otimes2}$, then we define the composition $\iota_v\iota_u$, for two homogeneous elements $u,v\in E$, by:
\begin{equation}\label{identification2}
\iota_v\iota_u(\alpha)\equiv2\,\big\langle\alpha,u\otimes v\big\rangle_{E^{\otimes2}}
\end{equation}
If, in particular, $\alpha\in S^2(E^*)$, then we can commute $\iota_v$ and $\iota_u$, so that the following identity holds:
\begin{equation}\label{working}
\iota_v\iota_u(\alpha)=2\,\big\langle\alpha,u\odot v\big\rangle_{S^2(E)}=(-1)^{|u||v|}2\,\big\langle\alpha,v\odot u\big\rangle_{S^2(E)}=(-1)^{|u||v|}\iota_u\iota_v(\alpha)
\end{equation}
An example of particular importance is when one has a basis $\{u_a\}$ of $E$, with dual basis $\{u^a\}$. Then we have:
\begin{equation}
\iota_{u_d}\iota_{u_c}\big(u^a\odot u^b\big)=2\,\big\langle u^a\odot u^b,u_c\odot u_d\big\rangle_{S^2(E)}=\delta^a_c\delta^b_d+(-1)^{|u^a||u_c|}\delta^a_d\delta^b_c\label{works2}
\end{equation}
In that case one can see that by formally identifying $\iota_{u_a}$ with the derivative $\frac{\partial}{\partial u^a}$, one obtains the right hand side of Equation \eqref{works2} by applying $\frac{\partial}{\partial u^d}\frac{\partial}{\partial u^c}$ to $u^a\odot u^b$. Notice that Equation \eqref{working} stands when it is rewritten with derivatives as well:
\begin{align}
(-1)^{|u_c||u_d|}\frac{\partial}{\partial u^c}\frac{\partial}{\partial u^d}(u^a\odot u^b)&=(-1)^{|u_c||u_d|}\delta^a_d\delta^b_c+(-1)^{|u^d|(|u_c|+|u^a|)}\delta^a_c\delta^b_d\\
&=(-1)^{-|u_c||u^a|}\delta^a_d\delta^b_c+\delta^a_c\delta^b_d\\
&=\frac{\partial}{\partial u^d}\frac{\partial}{\partial u^c}(u^a\odot u^b)
\end{align}
where the first term on the second line is non zero if and only if $|u^a|=-|u_d|$ because of the Kronecker's delta, and where the second term is non zero if and only if $|u_c|+|u^a|=0$ (for the same reason).

When one applies the suspension operator on both sides of the pairing, nothing changes:
\begin{equation}\label{suspensiondecal}
\big\langle s(\alpha),s^{-1}(u)\big\rangle_{s^{-1}E}\equiv\langle \alpha,u\rangle_E
\end{equation}
A similar equation holds when we swap $s$ with $s^{-1}$. Moreover, we have the following identity that we will use from time to time:
\begin{equation}\label{suspensiondecal2}
\Big\langle \big(s^2\odot s^2\big)(\alpha),u\odot v\Big\rangle_{S^2(E)}= \Big\langle \alpha\,,s^2(u)\odot s^2(v)\Big\rangle_{S^2(s^2E)}
\end{equation}
for every $u,v\in E$ and $\alpha\in S^2\big(s^{-2}(E^*)\big)$. Let $P:E\to F$ be a degree $p$ linear map between two graded vector spaces, then we define its dual $P^*:F^*\to E^*$ by:
\begin{equation}\label{dualitysmooth}
\big\langle P^*(\alpha),u\big\rangle_{E}\equiv (-1)^{p|\alpha|} \big\langle\alpha,P(u)\big\rangle_F
\end{equation}
This equation does not hold when $P$ is a representation, because usually in that case, the contragredient representation induces only a minus sign. For example, for $\mathfrak{g}$ a Lie algebra acting on a $\mathfrak{g}$-module $V$, we have:
\begin{equation}\label{dualitysmooth2}
\big\langle \rho^\vee_a(\alpha),x\big\rangle_{V}\equiv -\big\langle\alpha,\rho_a(x)\big\rangle_E
\end{equation}
for any $a\in\mathfrak{g}$, $x\in V$ and $\alpha\in V^*$. In the following we denote by $\rho^\vee$ the contragredient represention, induced by $\rho:\mathfrak{g}\to \mathrm{End}(V)$.

We can now give the equivalence that is of interest for us:
\begin{theoreme}\label{correspondence} Let $E=(E_i)_{i\in\mathbb{Z}}$ be a graded vector space. Then differential graded Lie algebra structures on $E$ are in one-to-one correspondence with differential graded manifold structures of arity at most one on the pointed graded manifold $s^{-1}E$.
\end{theoreme}
\begin{proof}
The formulas to pass from one structure to another are taken from \cite{Fiorenza} and \cite{voronov2}. Given $x,y\in E$, the relationship between $[x,y]$ in $E$ and the corresponding homological vector field $Q$ is given by:
\begin{equation}\label{voronov2bis}
\iota_{s^{-1}[x,y]}=(-1)^{|x|}\big[[Q,\iota_{s^{-1}(x)}],\iota_{s^{-1}(y)}\big]
\end{equation}
where on the right hand side, we use the bracket of (graded) vector fields on $s^{-1}E$. On the other hand, the differential $\partial$ satisfies:
\begin{equation}\label{bracketbis}
\iota_{s^{-1}(\partial(x))}=-[Q,\iota_{s^{-1}(x)}]\big|_0
\end{equation}
where the sub-script $|_0$ means that the vector field is constant and its value is the one taken at the origin. Formulas \eqref{voronov2bis} and \eqref{bracketbis} provide a one-to-one correspondence between the differential graded Lie algebra structure on $E$ and the differential graded manifold structure on $s^{-1}E$. The Jacobi and Leibniz identities are indeed incapsulated into the homological condition $[Q,Q]=0$. More details are found in \cite{Fiorenza} and \cite{voronov2} .
\end{proof}


\begin{example}
Let $\big(\mathfrak{g}=\mathfrak{g}_0\oplus\mathfrak{g}_{-1},\partial,[\,.\,,.\,]\big)$ be a differential graded Lie algebra. Given a basis $(e_i)_{1\leq i\leq n}$ of $\mathfrak{g}_0$ and $(f_a)_{1\leq a\leq m}$ of $\mathfrak{g}_{-1}$, there exist tensors $C_{ij}^k$, $C_{ia}^b$ and $d_a^i$ such that:
\begin{equation}
\partial(f_a)=d_a^i\,e_i\,,\hspace{1cm}[e_i,e_j]=C_{ij}^k\,e_k\hspace{1cm}\text{and}\hspace{1cm}[e_i,f_a]=C_{ia}^b\,f_b
\end{equation}
Setting $(\widetilde{e}_i)_{1\leq i\leq n}$ be the basis for $s^{-1}\mathfrak{g}_0$ and $(\widetilde{f}_a)_{1\leq a\leq m}$ be the basis for $s^{-1}\mathfrak{g}_{-1}$, the corresponding homological vector field on $s^{-1}\mathfrak{g}$ is:
\begin{equation}
Q=-d_a^i\,\widetilde{f}^{*a}\otimes\iota_{\widetilde{e}_i}-\frac{1}{2}C_{ij}^k\,\widetilde{e}^{*i}\widetilde{e}^{*j}\otimes\iota_{\widetilde{e}_k}-C_{ia}^b\,\widetilde{e}^{*i}\widetilde{f}^{*a}\otimes\iota_{\widetilde{f}_b}
\end{equation}
where the star denotes the dual basis.
\end{example}

Given this one-to-one correspondence, we can define a cohomology on any  graded Lie algebra that mimics the Chevalley-Eilenberg cohomology of Lie algebras. Let $\big(\mathfrak{g}
,[\,.\,,.\,]\big)$ be a graded Lie algebra, and let $(s^{-1}\mathfrak{g}, Q)$ be the associated differential graded manifold structure. The homological vector field $Q$ can be seen as a differential on $S\big((s^{-1}\mathfrak{g})^*\big)$. Since the only non vanishing term in $Q$ is of arity one, it defines a chain complex of graded vector spaces:
\begin{center}
\begin{tikzcd}[column sep=0.7cm,row sep=0.4cm]
0\ar[r]&(s^{-1}\mathfrak{g})^*\ar[r,"Q"]&S^2\big((s^{-1}\mathfrak{g})^*\big)\ar[r,"Q"]&S^3\big((s^{-1}\mathfrak{g})^*\big)\ar[r,"Q"]&\ldots
\end{tikzcd}
\end{center}
This sequence can be augmented on the left to the Chevalley-Eilenberg complex of $\mathfrak{g}$ acting trivially on $\mathbb{R}$, when $Q$ is identified with the Chevalley-Eilenberg differential $\dd_{\mathrm{CE}}$:
\begin{center}
\begin{tikzcd}[column sep=0.7cm,row sep=0.4cm]
0\ar[r]&\mathbb{R}\ar[r,"0"]&\mathrm{Hom}(\mathfrak{g},\mathbb{R})\ar[r, "\dd_{\mathrm{CE}}"]&\mathrm{Hom}(\wedge^2\mathfrak{g},\mathbb{R})\ar[r, "\dd_{\mathrm{CE}}"]&\mathrm{Hom}(\wedge^3\mathfrak{g},\mathbb{R})\ar[r, "\dd_{\mathrm{CE}}"]&\ldots
\end{tikzcd}
\end{center}
The cohomology that is associated to this complex is called \emph{the Chevalley-Eilenberg cohomology of the graded Lie algebra $\mathfrak{g}$} and it is noted $H_{\mathrm{CE}}(\mathfrak{g})=\bigoplus_{k\geq0}H^k_{\mathrm{CE}}(\mathfrak{g})$. The spaces $H^k_{\mathrm{CE}}(\mathfrak{g})$ inherit the grading of $\mathfrak{g}$.
When $\mathfrak{g}$ is restricted to non-positive degrees, i.e. when $\mathfrak{g}=\bigoplus_{k\geq1}\mathfrak{g}_{-k}$, we have $\mathrm{d}_{\mathrm{CE}}\big((\mathfrak{g}_{-1})^*\big)=0$ and $[\mathfrak{g}_{-1},\mathfrak{g}_{-1}]\subset \mathfrak{g}_{-2}$, which implies the following two inclusions: 
\begin{equation*}
(\mathfrak{g}_{-1})^*\subset H^1_{\mathrm{CE}}(\mathfrak{g})\hspace{1cm}\text{and}\hspace{1cm}\bigslant{\wedge^2 (\mathfrak{g}_{-1})^*}{\mathrm{d}_{\mathrm{CE}}\big((\mathfrak{g}_{-2})^*\big)}\subset H^2_{\mathrm{CE}}(\mathfrak{g})
\end{equation*}
When it is an equality, it means that the restriction of the map $\mathrm{d}_{\mathrm{CE}}$ to any $(\mathfrak{g}_{-k})^*$, for $k>1$, is injective, and that $\mathrm{Im}(\mathrm{d}_{\mathrm{CE}}|_{(\mathfrak{g}_{-k})^*})=\mathrm{Ker}(\mathrm{d}_{\mathrm{CE}}|_{\wedge^2(\mathfrak{g}^*)_{k+1}})$. This property will be important in the following so that it deserves a name:
\begin{definition}
We say that a strictly negatively graded Lie algebra $\mathfrak{g}=\bigoplus_{i\geq1}\mathfrak{g}_{-i}$ is \emph{robust} when either the first or the second following conditions holds: 
\begin{enumerate}
\item $\mathfrak{g}$ is of depth 1, i.e. when $\mathfrak{g}=\mathfrak{g}_{-1}$,
\item if its depth is higher than 1, when the following equalities are satisfied\footnote{Notice that in the final, printed version of the present paper, an important typo has escaped scrutiny of the proofreading, as there it is written $H^2_{\mathrm{CE}}(\mathfrak{g})=\mathrm{d}_{\mathrm{CE}}\big((\mathfrak{g}_{-2})^*\big)$.}:
\begin{equation*}
H^1_{\mathrm{CE}}(\mathfrak{g})=(\mathfrak{g}_{-1})^*\hspace{1cm}\text{and}\hspace{1cm}H^2_{\mathrm{CE}}(\mathfrak{g})=\bigslant{\wedge^2 (\mathfrak{g}_{-1})^*}{\mathrm{d}_{\mathrm{CE}}\big((\mathfrak{g}_{-2})^*\big)}
\end{equation*}
\end{enumerate}
\end{definition}

\begin{remarque}
Notice that in the first case, the Lie algebra structure is trivial for degree reasons: the bracket of two elements of degree $-1$ should be of degree $-2$, but there is no space of degree $-2$ in the first item.
\end{remarque}

\section{Building the tensor hierarchy }

This section is devoted to the construction of a tensor hierarchy algebra associated to a Lie-Leibniz triple $\mathcal{V}=(\mathfrak{g},V,\Theta)$. Let $W$ be the bud of $\mathcal{V}$ and let $\mathrm{d}$ be the collar of $\mathcal{V}$. Since $\mathrm{Im}\big(\{.\,,.\}\big)\subset\mathrm{Ker}(\bullet)$, the following diagram is commutative and the composition of arrows is zero:

\begin{center}
\begin{tikzpicture}
\matrix(a)[matrix of math nodes, 
row sep=5em, column sep=6em, 
text height=1.5ex, text depth=0.25ex] 
{&W&\\ 
S^2(V)&V&\mathrm{Der}(V)\\
&\mathfrak{h}&\\}; 
\path[->>](a-2-1) edge node[above left]{$\Pi_W$} (a-1-2); 
\path[->](a-1-2) edge node[right]{$\dd$} (a-2-2);
\path[->>](a-2-2) edge node[left]{$\Theta$} (a-3-2);
\path[->](a-3-2) edge node[above left]{$\rho$} (a-2-3);
\path[->](a-1-2) edge node[right]{$\dd$} (a-2-2);
\path[->](a-2-1) edge node[above]{$\{.\,,.\}$} (a-2-2);
\path[->](a-2-2) edge node[above]{$\bullet$} (a-2-3);
\end{tikzpicture}
\end{center}

 The motivation for the construction of the tensor hierarchy relies on the observation that if one consider elements of $\mathfrak{h}$, $V$ and $W$ as having degree $0$, $-1$ and $-2$, respectively, the maps $\rho$ and $\Pi_W$ induce a  skew-symmetric bracket on the graded vector space $\mathfrak{h}\oplus V\oplus W$.
Unfortunately, for degree reasons, they do not define a graded Lie algebra structure, since the Jacobi identity cannot be satisfied. This justifies to find a vector space $X$ with degree $-3$ and adapted brackets that would enable the closure of the Jacobi identity. The goal of this section is to construct the tower of spaces that defines the tensor hierarchy. We will then show that this graded vector space can be equipped with a differential graded Lie algebra structure that contain all relevant informations required by gauging procedures in supergravity.

In \cite{palmkvist1}, the tensor hierarchy algebra is defined using Borcherds algebras. One quotients out some particular ideal from the free Lie algebra of $V$. This top-down approach gives, up to a sign change in the grading, a differential graded Lie algebra structure on some graded vector space $T=\bigoplus_{k\geq-1}T_{-k}$, with $T_{+1}=T_\Theta$, $T_0=\mathfrak{g}$,  $T_{-1}=V$, and where each $T_{-k}$ for $k\geq2$ is a quotient of $\big[\cdots[[V,V],V]\cdots\big]$ (with $i$ copies of $V$), see also \cite{palmkvist}. This algebraic structure on $T$ is called a tensor hierarchy algebra. In particular, it is suggested in \cite{Cederwall2} that the graded Lie algebra structure induced on $T'\equiv \bigoplus_{k\geq1}T_{-k}$ is robust.
In this section, we  present a bottom-up construction alternative to the one given in \cite{palmkvist, palmkvist1}. We are convinced that it gives a tensor hierarchy algebra structure on $T'\oplus\mathfrak{h}$ that is the mere restriction of the tensor hierarchy algebra structure on $T$ described in \cite{palmkvist, palmkvist1}. 

We believe that the definition given in \cite{palmkvist} is the correct definition of a tensor hierarchy algebra, but we chose the reverse convention on the grading, and we do not consider $T_{+1}$ nor $T_0$ in the same way as in \cite{palmkvist}:  

\begin{definition}\label{def:tensoralgebra}
Let $\mathcal{V}=(\mathfrak{g},V,\Theta)$ be a Lie-Leibniz triple, let $\mathfrak{h}$ denote $\mathrm{Im}(\Theta)$ and let $W$ be the bud of $\mathcal{V}$. A \emph{tensor hierarchy algebra associated to $\mathcal{V}$} is a differential graded Lie algebra $\big( T,\partial,[\,.\,,.\,]\big)$ that consists of a negatively graded $\mathfrak{g}$-module $T=(T_{-k})_{k\geq0}$ (i.e. such that for every $k\geq1$, $T_{-k}$ is a $\mathfrak{g}$-module) that satisfies:
\begin{enumerate}
\item $T_0=\mathfrak{h}$,
\item  $T_{-1}= s^{-1}V$, and
\item  $T_{-2}=s^{-2}W$. 
\end{enumerate}
The graded Lie bracket $[\,.\,,.\,]$ is  such that:
\begin{enumerate}\setcounter{enumi}{3}
\item the graded Lie algebra $\big((T_{-k})_{k\geq1},[\,.\,,.\,]\big)$ is robust and the bracket is $\mathfrak{g}$-equivariant:
\begin{equation}
\big[\eta_{-k,a}(x),y\big]+\big[x,\eta_{-k,a}(y)\big]=\eta_{-k-l,a}\big([x,y]\big)
\end{equation}
for every $x\in T_{k}$,$y\in T_{-l}$ and $a\in\mathfrak{g}$, where $k,l\geq1$;
\item the bracket $[\,.\,,.\,]:T_{-1}\otimes T_{-1}\to T_{-2}$ satisfies, for all $x,y\in T_{-1}$:
\begin{equation}\label{conditioncrochet}
[x,y]\equiv2 \,s^{-2}\circ\Pi_W\big(s(x),s(y)\big)
\end{equation}
where $\Pi_W:S^2(V)\to W$ is the canonical projection on the bud of $\mathcal{V}$;
\item the bracket on $T_0$ is the Lie bracket on $\mathfrak{h}$;
\item for all $k\geq1$, the bracket $[\,.\,,.\,]: T_0\otimes T_{-k}\to T_{-k}$ is defined by the action of $\mathfrak{h}$ on $T_{-k}$:
\begin{equation}\label{forkequalzero}
\forall\ a\in\mathfrak{h},x\in T_{-k}\hspace{1cm}[a,x]\equiv\eta_{-k,a}(x)=-[x,a]
\end{equation}
where $\eta_{-k}:\mathfrak{g}\to \mathrm{End}(T_{-k})$ encodes the $\mathfrak{g}$-module structure on $T_{-k}$.
\end{enumerate}
The differential $\partial=\big(\partial_{-k}:T_{-k-1}\to T_{-k}\big)_{k\geq0}$  satisfies at highest levels:
\begin{enumerate}\setcounter{enumi}{7}
\item $\partial_{0}\equiv-\Theta\circ s$
\item
$\partial_{-1}\equiv -s^{-1}\circ\dd\circ s^2$
\end{enumerate}
where $\dd$ is the collar of $\mathcal{V}$.
\end{definition}

\begin{remarques}
\begin{enumerate}
\item If the Leibniz algebra $V$ is a Lie algebra, then its bud $W$ is the zero vector space, and the depth of the corresponding tensor hierarchy algebra is 1. 
\item The data that $T_{-2}=s^{-2}W$ and that the bracket between two elements of $T_{-1}$ satisfy Equation \eqref{conditioncrochet} are important since they are characterizing the tensor hierarchies in supergravity.
\item If one defines $\eta_0:\mathfrak{h}\to \mathrm{End}(\mathfrak{h})$ to be the adjoint action, then Equation \eqref{forkequalzero} is even consistent for $k=0$.
\item The fact that the algebra degree stops at 0 implies that $\partial(a)=0$ for every $a\in\mathfrak{h}$. Then, by the derivation property of the differential, we deduce that $\partial$ is $\mathfrak{h}$-equivariant:
\begin{equation}
\partial_{-k+1}\big(\eta_{-k,a}(x)\big)=\eta_{-k,a}\big(\partial_{-k} (x)\big)
\end{equation}
for every $x\in T_{-k}$, where $k\geq1$, and every $a\in\mathfrak{h}$. However, it may not be $\mathfrak{g}$-equivariant. In supergravity theories, the differential $\partial$ is actually an element of $T_\Theta$ \cite{Trigiante}.
 

\item This algebra is related to the tensor hierarchy algebra defined in \cite{palmkvist}, by noticing that the differential $\partial$ can be seen as the adjoint action of an element of degree $+1$. Indeed let $T_{+1}\equiv s\big(\langle \Theta\rangle\big)$ be the one dimensional space generated by the embedding tensor $\Theta$. The differential $\partial$ is then related to the embedding tensor by the following equation: \begin{equation}\partial=[\Theta,.\,]\end{equation}
This applies in particular to $\Theta$, giving $[\Theta,\Theta]=0$ that is precisely the quadratic constraint, and that implies the cohomological condition $\partial^2=0$. This is consistent with the fact that the differential $\partial$ is zero on $T_0=\mathfrak{h}$, because, for every $a\in\mathfrak{h}$ it would write $0=[\Theta,a]=-\rho_a(\Theta)$ which corresponds to the $\mathfrak{h}$-invariance of $\Theta$.
\end{enumerate}
\end{remarques}

The notion of morphism between two tensor hierarchy algebras have to be compatible with the underlying Lie-Leibniz triples:

\begin{definition}
Let $\big( T,\partial,[\,.\,,.\,]\big)$ (resp. $\big(\overline{T},\overline{\partial},\overline{[\,.\,,.\,]}\big)$) be a tensor hierarchy algebra associated to some Lie-Leibniz triple $(\mathfrak{g},V,\Theta)$ (resp. $(\overline{\mathfrak{g}},\overline{V},\overline{\Theta})$). A \emph{tensor hierarchy algebra morphism between $T$ and $\overline{T}$} is a couple $(\varphi,\phi)$, where $\varphi:\mathfrak{g}\to\overline{\mathfrak{g}}$ is a Lie algebra morphism, and where $\phi:T\to \overline{T}$ is a differential graded Lie algebra morphism such that $\phi_0=\varphi\big|_\mathfrak{h}$ and:
\begin{equation}\label{insup}
\overline{\eta}_{-k,\varphi(a)}\circ\phi_{-k}=\phi_{-k}\circ\eta_{-k,a}
\end{equation}
for every $k\geq1$ and every $a\in\mathfrak{g}$. 

When $T$ and $\overline{T}$ have the same depth $i\in\mathbb{N}^*\cup\{\infty\}$,
we say that the tensor hierarchy algebra morphism $\phi: T\to\overline{T}$ is an \emph{isomorphism} if:
\begin{enumerate}
\item $\varphi:\mathfrak{g}\to\overline{\mathfrak{g}}$ is a Lie algebra isomorphism, and
\item if $\phi_k:T_{-k}\to\overline{T}_{-k}$ is an isomorphism for every $1\leq k <i+1$.
\end{enumerate}
%
\end{definition}

\begin{remarques} \begin{enumerate}
\item Notice that this automatically implies that the couple $\big(\phi_0\,,s\circ\phi_{-1}\circ s^{-1}\big)$ is a Lie-Leibniz triple morphism between $\mathcal{V}$ and $\overline{\mathcal{V}}$. In particular, the condition that $s\circ\phi_{-1}\circ s^{-1}$ is a Leibniz algebra morphism follows from Equations \eqref{eq:compat} and \eqref{insup}.
\item The notion of tensor hierarchy algebra morphism differs from the usual notion of quasi-isomorphisms in the category of differential graded Lie algebras, for the latter is onjly bijective at the cohomology level.
\end{enumerate}
\end{remarques}

The first step to build a tensor hierarchy associated to $V$ is to define a chain complex:
\begin{center}
\begin{tikzcd}[column sep=0.7cm,row sep=0.4cm]
0&\ar[l]T_{0}&\ar[l,"\partial_{0}"]T_{-1}&\ar[l,"\partial_{-1}"]T_{-2}&\ar[l,"\partial_{-2}"]T_{-3}&\ar[l]\ldots
\end{tikzcd}
\end{center}
in which we expect that $T_0=\mathfrak{h}$, $T_{-1}=s^{-1}(V)$ and $T_{-2}=s^{-2}W$. Our goal is to show that the process of constructing this structure is unique and straightforward. We have been inspired by the construction that is performed in gauging procedures in supergravity \cite{Henning08, Henning08bis}. 
We will proceed in two steps: first, from a Lie-Leibniz triple, construct a chain complex:
\begin{center}
\begin{tikzcd}[column sep=0.7cm,row sep=0.4cm]
0\ar[r]&U_{0}\ar[r,"\delta_{1}"]&U_{1}\ar[r,"\delta_{2}"]&U_{2}\ar[r,"\delta_{3}"]&U_{3}\ar[r]&\ldots
\end{tikzcd}
\end{center}
that has some adequate properties, e.g. $U_0=V^*$, $U_1=s(W^*)$ and $\delta_1=s\circ\mathrm{d}^*$. The complex $S(U)$ has then to be equipped with some maps that have some convenient properties. This is worked out in Sections \ref{tower} and \ref{london} where some unicity results are discussed. Second, define the shifted dual of this chain complex via the following equality:
\begin{equation*}
T_{-k}\equiv s^{-1}(U_{k-1}^*)\hspace{1cm}\text{for any $k\geq1$}
\end{equation*}
Then, using the data attached to the chain complex $U=(U_i)_{i\geq0}$, we show that the following chain complex:
\begin{center}
\begin{tikzcd}[column sep=0.7cm,row sep=0.4cm]
0&\ar[l]T_{-1}&\ar[l,"\partial_{-1}"]T_{-2}&\ar[l,"\partial_{-2}"]T_{-3}&\ar[l]\ldots
\end{tikzcd}
\end{center}
can be equipped with a robust graded Lie algebra structure. This algebraic structure is not totally compatible with the differential $\partial$, unless we add a space $T_0\equiv \mathfrak{h}$ at level 0. Then by a cautious analysis of the brackets and of the differentials, we conclude that $ T=(T_{-k})_{k\geq0}$ can be equipped by a tensor hierarchy algebra structure. 
The discussion on this second point takes place in Section \ref{sec:hierarchy}, where we conclude that every Lie-Leibniz triple induces a unique tensor hierarchy algebra.  
Section \ref{sectionexamples} then provides examples that are presented in details so that the construction that is made in the preceding sections make sense.

Finally, we would like to emphasize that the construction of the tensor hierarchy algebra that is given in this paper is only the first step toward a better understanding of gauging procedures in supergravity. The next step will be to find a convincing way of building the $L_\infty$-algebras involved in supergavity, from the data of these tensor hierarchy algebras. This topic is not present in the present paper, because it is still under investigation.


\subsection{The stem of a Lie-Leibniz triple}
\label{tower}

The aim of this section is to define the `stem' of a tensor hierarchy algebra associated to a Lie-Leibniz triple, that is: the (possibly infinite) tower of space which is underlying the tensor hierarchy algebra. The construction of this tower of spaces is made by induction. Let $\mathcal{V}=(\mathfrak{g},V,\Theta)$ be a Lie-Leibniz triple and let $W\subset S^2(V)$ be the bud of $\mathcal{V}$. Then we set $U_{0}=V^*$ and $U_1=s(W^*)$. In this setup, the shifted dual of the collar $\dd$ becomes a degree $+1$ map that we call $\delta_{1}$:
\begin{equation}
\delta_1=s\circ\dd^*:U_0\to U_1
\end{equation}
where the dual is taken with respect to the pairing between $V,V^*$ and $W,W^*$, as given in Equation \eqref{dualitysmooth}:
\begin{equation}\label{pouuf}
\big\langle\mathrm{d}^*(\alpha),u\big\rangle_{W}=\big\langle \alpha,\mathrm{d}(u)\big\rangle_V
\end{equation}
for every $\alpha\in V^*$ and $u\in W$. Here, $V\oplus W$ is seen as a mere vector space, and $\mathrm{d}$ as a degree 0 endomorphism, so that $\delta_1$ is a degree $+1$ linear mapping. We see that we have defined the two first spaces of a chain complex:
\begin{center}
\begin{tikzcd}[column sep=0.7cm,row sep=0.4cm]
0\ar[r]&U_{0}\ar[r,"\delta_1"]&U_{1}
\end{tikzcd}
\end{center}

The construction of the tensor hierarchy relies precisely on the choice of $U_{1}\simeq W$. Once this space is fixed, the procedure is unique and straightforward. It is now time to define the backbone of the construction:

\begin{definition}\label{wooo}
Let $\mathcal{V}=(\mathfrak{g},V,\Theta)$ be a Lie-Leibniz triple, let $W$ be the bud of $\mathcal{V}$ and let $\mathrm{d}$ be the collar of $\mathcal{V}$. A \emph{$i$-stem associated to $\mathcal{V}$} (for $i\in\mathbb{N}\cup\{\infty\}$) is a 4-tuple $(U,\delta, \pi,\mu)$ where
$U=(U_{k})_{0\leq k<i+1}$ is a family of $\mathfrak{g}$-modules, with respective action $\rho_k:\mathfrak{g}\to \mathrm{End}(U_{k})$, such that, if $i=0$ then $\mathcal{U}\equiv(V^*,0,0,0)$, and if $i\neq0$ we have the following conditions:
\begin{enumerate}
\item $U_{0}= V^*$ and $U_1=s(W^*)$; 
\item  $\rho_0$ (resp. $\rho_1$) is the contragredient action of $\mathfrak{g}$ on $V$ (resp.  $s^{-1}W$):
\begin{equation*}
\rho_0=\eta_V^{\vee}\hspace{1cm}\text{and}\hspace{1cm}\rho_1=s\circ\eta_W^\vee\circ s^{-1}
\end{equation*}
where $\eta_V$ (resp. $\eta_W$) is the representation of $\mathfrak{g}$ on $V$ (resp. $W$),
\end{enumerate} and where $\delta$, $\pi$ and $\mu$ are three families of maps, consisting of:
\begin{itemize}
\item a $\mathrm{Im}(\Theta)$-equivariant differential $\delta=\big(\delta_k:U_{k-1}\to U_{k}\big)_{1\leq k<i+1}$,  
\item a family $\pi=(\pi_{k})_{0\leq k<i}$ of $\mathfrak{g}$-equivariant degree $-1$ linear maps $\pi_k:U_{k+1}\to S^2(U)_{k}$, 
\item a family $\mu=(\mu_{k})_{0\leq k<i}$ of degree $0$ linear maps $\mu_k:U_{k}\to S^2(U)_{k}$,
\end{itemize}
that are extended to all of $S(U)$ as derivations, and such that they satisfy the following conditions:
\begin{enumerate}
\setcounter{enumi}{2}
\item at lowest orders, the maps $\mu_0, \pi_0$ and $\delta_1$ satisfy:
\begin{center}
\begin{tikzpicture}
\matrix(a)[matrix of math nodes, 
row sep=5em, column sep=5em, 
text height=1.5ex, text depth=0.25ex] 
{U_{0}&U_1\\ 
S^2(U_0)&\\}; 
\path[->](a-1-1) edge node[above]{$s\circ \mathrm{d}^*$}  (a-1-2); 
\path[->](a-1-1) edge node[left]{$-\,\{\,.\,,.\,\}^*$} (a-2-1);
\path[->](a-1-2) edge node[below right]{$-\,\Pi_W^*\circ s^{-1}$} (a-2-1);
\end{tikzpicture}
\end{center}

\item for every $0\leq k <i$, the map $\pi_k$ defines an exact sequence:
\begin{center}
\begin{tikzcd}[column sep=0.7cm,row sep=0.4cm]
0\ar[r]&U_{k+1}\ar[r,"\pi_k"]&S^2(U)_{k}\ar[r,"\pi"]&S^3(U)_{k-1}
\end{tikzcd}
\end{center}
\item for every $1\leq k<i$, the map $\mu_k$ satisfies:
\begin{equation}
2\,\big\langle \alpha\odot \beta,\mu_k(u)\big\rangle_{S^2(U)_k}\equiv\big\langle \alpha,\rho_{k,\Theta(\beta)}(u)\big\rangle_{U_k}+\big\langle \beta,\rho_{k,\Theta(\alpha)}(u)\big\rangle_{U_k}
\end{equation}
for any $\alpha\odot\beta\in S^2(U)_k$, where $\rho:\mathfrak{g}\to \mathrm{End}(U)$ is the unique map that restricts to $\rho_k$ on $U_k$, and where $\Theta$ is considered as the zero map if acting on $U_k^*$, for any $k\geq1$. 
\item the map $\mu:U\to S^2(U)$ is a null-homotopic chain map between $U$ and $S^2(U)$:
\begin{center}
\begin{tikzpicture}
\matrix(a)[matrix of math nodes, 
row sep=4em, column sep=3em] 
{U_{0}&U_{1}&U_{2}&U_{3}&\cdots\\
S^2(U)_0&S^2(U)_{1}&S^2(U)_{2}&\cdots\\};
\path[left hook->](a-1-2) edge node[above left]{$\pi$}  (a-2-1);
\path[->](a-1-2) edge node[right]{$\mu$}  (a-2-2);
\path[->](a-1-1) edge node[above]{$\delta$}  (a-1-2);
\path[->](a-2-2) edge node[above]{$\delta$} (a-2-3);
\path[dotted,->](a-2-3) edge  (a-2-4);
\path[->](a-1-2) edge node[above]{$\delta$}  (a-1-3);
\path[left hook->](a-1-3) edge node[above left]{$\pi$} (a-2-2);
\path[->](a-1-3) edge node[above]{$\delta$} (a-1-4);
\path[dotted,left hook->](a-1-5) edge  (a-2-4);
\path[->](a-2-1) edge node[above]{$\delta$}  (a-2-2);
\path[left hook->](a-1-4) edge node[above left]{$\pi$}  (a-2-3);
\path[dotted,->](a-1-4) edge (a-1-5);
\path[dotted,->](a-1-4) edge  (a-2-4);
\path[->](a-1-3) edge node[right]{$\mu$} (a-2-3);
\path[->](a-1-1) edge node[right]{$\mu$}  (a-2-1);
\end{tikzpicture}
\end{center}
\end{enumerate}
The \emph{$j$-truncation} (for $0\leq j< i$)   of the $i$-stem $(U,\delta,\pi,\mu)$ is the $j$-stem of $\mathcal{V}$ defined by the quadruple $\big(U'\equiv\bigoplus_{0\leq k\leq j}U_{k},\delta|_{U'},\pi|_{U'},\mu|_{U'}\big)$.
\end{definition}


\begin{remarques}
Some remarks are necessary:
\begin{enumerate}
\item 
Using Equations \eqref{pairing} and \eqref{suspensiondecal}, the content of item 2. is equivalent to:
\begin{align}
\big\langle x\,,\rho_{0,a}(u)\big\rangle_{U_0}&\equiv-\big\langle \eta_{V,a}(x),u\big\rangle_{U_0}=-\big\langle u\,, \eta_{V,a}(x)\big\rangle_{V}\label{representation1}\\
\big\langle\alpha\,,\rho_{1,a}(\omega)\big\rangle_{U_1}&\equiv-\big\langle s^{-1}\circ\eta_{W,a}\big(s(\alpha)\big),\omega\big\rangle_{U_1}=-\big\langle s^{-1}(\omega),\eta_{W,a}\big(s(\alpha)\big)\big\rangle_{W}\label{representation2}
\end{align}
for every $a\in\mathfrak{g}$, $x\in U_0^*=V$, $u\in U_0=V^*$, $\alpha\in U_1^*=s^{-1}W$ and $\omega\in U_1=s(W^*)$. In other words, $\rho_1^\vee=s^{-1}\circ \eta_W\circ s$. 
\item By applying Equation \eqref{dualitysmooth} to item 3., the dual of the symmetric bracket is defined by using the pairing between $S^2(V)$ and $S^2(V^*)$ on the one hand, and the pairing between $V$ and $V^*$ on the other hand:
\begin{equation}\label{worksz}
\big\langle\{\,.\,,.\,\}^*(\alpha),x\odot y\big\rangle_{S^2(V)}\equiv\big\langle \alpha\,,\{x,y\}\big\rangle_{V}
\end{equation}
for any $\alpha\in V^*$, and $x,y\in V$.
The definition for the map $\Pi_W^*$  is made in a similar way:
\begin{equation}\label{dualpi}
\big\langle\Pi_W^*(u),x\odot y\big\rangle_{S^2(V)}\equiv\big\langle u\,,\Pi_W(x\odot y)\big\rangle_{W}
\end{equation}
where $u\in W^*$ and $x,y\in V$. There is no minus sign on the right hand side because $W$ is supposed to have degree 0, as well as $\Pi_W$.
\item Item 5. implies that $\mathrm{Im}(\mu_k)\subset U_0\odot U_k$, for every $1\leq k<i$. Moreover, calling $\rho^\vee$ the contragredient representation of $\rho$, 
 item 5. translates as:
\begin{equation}\label{equationmu}
2\,\big\langle \alpha\odot \beta\,,\mu_k(u)\big\rangle_{S^2(U)_k}=-\big\langle\rho^\vee_{\Theta(\alpha)}(\beta)+\rho^\vee_{\Theta(\beta)}(\alpha),u\big\rangle_{U_k}
\end{equation}
where, still, $\Theta$ is considered as the zero map if acting on $U_i^*$, for any $i\geq1$.
When $k=0$, and for $\alpha,\beta\in U_0^*=V$ and $u\in U_0=V^*$, we have the identity $\rho_0^\vee=\eta_V$ and Equation \eqref{equationmu} coincides with the definition of $\mu_0\equiv-\{\,.\,,.\,\}^*$ given in item 3. of Definition \ref{wooo}.
\item Dualizing Equation \eqref{inclusionlol} implies the following important identity:
\begin{equation}\label{eq:quad}
\delta_1\circ\Theta^*=0
\end{equation}
where $\Theta^*:\mathfrak{g}^*\to V^*$ is the dual map of $\Theta$ defined by:
\begin{equation}
\big\langle\Theta^*(u),x\big\rangle_{V}=\big\langle u,\Theta(x)\big\rangle_\mathfrak{g}
\end{equation}
for every $u\in\mathfrak{g}^*$ and $x\in V$. In particular, it is injective on $\mathfrak{h}^*$.
From this, we deduce that the chain complex $(U,\delta)$ admits an augmentation by $\mathfrak{g}^*$:
\begin{center}
\begin{tikzcd}[column sep=0.7cm,row sep=0.4cm]
0\ar[r]&\mathfrak{g}^*\ar[r,"\Theta^*"]&U_{0}\ar[r,"\delta_1"]&U_{1}\ar[r,"\delta_{2}"]&U_{2}\ar[r,"\delta_{3}"]&\ldots
\end{tikzcd}
\end{center}

%
\item The condition $\delta^2=0$ may not be necessary in some cases (see \cite{Trigiante}). In most supergravity models, a careful analysis shows that the null-homotopy condition in item 6. and Equation \eqref{eq:quad} imply the homological condition $\delta^2=0$. 
\end{enumerate}
\end{remarques}

\begin{example}
A natural example of a 1-stem of a Leibniz algebra $V$ is the one described in item 3.
\end{example}

 We now show that if $i\geq0$, a $i$-stem associated to a Lie-Leibniz triple can always be extended a step further:


%


\begin{theoreme}\label{lemmestrand}
Let $i\in\mathbb{N}$ and let $\mathcal{V}=(\mathfrak{g},V,\Theta)$ be a Lie-Leibniz triple admitting a  $i$-stem $\mathcal{U}=(U,\delta,\pi,\mu)$. 
Then there exists a $(i+1)$-stem whose $i$-truncation is $\mathcal{U}$. 
\end{theoreme}


\begin{proof} The result is obvious if $i=0$, so we can assume that $i\in\mathbb{N}^*$. The idea of the proof is that the space of degree $i+1$ will be  defined so as to satisfy exactness of the map $\pi_{i}$, as in item 4. of Definition \ref{wooo}. 
Then, the definition of the map $\mu_i$ is made so that item 5. is satisfied. Most difficulties come from the definition of the map $\delta_{i+1}$: in particular it should be defined in a way so that item 6. is satisfied. We will see that its definition relies on Equations \eqref{donkey2} and \eqref{donkey1} whose proof is technical and thus postponed to Appendix \ref{appendicite}.

Let $\mathcal{U}=(U,\delta,\pi,\mu)$ be a $i$-stem associated to the Lie-Leibniz triple $\mathcal{V}$. In particular, $U=\bigoplus_{0\leq k\leq i} U_k$, $\delta=(\delta_k)_{1\leq k\leq i}$, $\pi=(\pi_k)_{0\leq k\leq i-1}$ and $\mu=(\mu_k)_{0\leq k\leq i-1}$, are such that they satisfy Definition \ref{wooo} up to level $i$. We define the vector space $U_{i+1}$ as:
$$U_{i+1}=s\Big(\mathrm{Ker}\big(\pi|_{S^2(U)_i}\big)\Big)$$
There is no certainty that the space $U_{i+1}$ is not zero, but the construction is still valid in that case. We build the degree $-1$ injective map $\pi_i$ by using the inclusion map:
\begin{equation}
\pi_{i}\equiv\iota\circ s^{-1}:U_{i+1}\to S^2(U)_i
\end{equation}
In particular we have the following exact sequence: 
\begin{center}
\begin{tikzcd}[column sep=0.7cm,row sep=0.4cm]
0\ar[r]&U_{i+1}\ar[r,"\pi_{i}"]&S^2(U)_i\ar[r,"\pi"]&S^3(U)_{i-1}
\end{tikzcd}
\end{center}
Hence item 4. is satisfied at level $i+1$.

By extending the respective actions of $\rho_k$ on $U_k$ -- for every $k\geq0$ -- to $S^2(U)$ by derivation, the space $S^2(U)_i$ becomes a $\mathfrak{g}$-module. 
We call $\rho:\mathfrak{g}\to\mathrm{Der}\big(S^2(U)\big)$ the corresponding map.
Since $\pi$ is $\mathfrak{g}$-equivariant, $\mathrm{Ker}\big(\pi|_{S^2(U)_i}\big)$ is a $\mathfrak{g}$-sub-module of $S^2(U)_i$. Hence, the sub-space $\mathrm{Im}(\pi_i)$ is a representation of $\mathfrak{g}$. Since $\pi_i$ is injective, this $\mathfrak{g}$-module structure can be transported back to $U_{i+1}$, turning it into a representation of $\mathfrak{g}$. For every $x\in U_{i+1}$, the action of $a\in\mathfrak{g}$ on $x$ is defined by:
\begin{equation}
\rho_{i+1,a}(x)\equiv(\pi_i)^{-1}\circ\rho_{a}\big(\pi_i(x)\big)
\end{equation}
Then, by construction, the map $\pi_i$ is $\mathfrak{g}$-equivariant at level $i+1$, as required in Definition \ref{wooo}.

It is now time to show that there exist a map $\mu_{i}$ and a map $\delta_{i+1}$ that combine with $\pi_i$ to satisfy all other items of Definition \ref{wooo} (in particular item 6.).
Since $U_{i}$ admits a $\mathfrak{g}$-action $\rho_i:\mathfrak{g}\to \mathrm{End}(U_{i})$,
 this representation defines a map  $\widetilde{\rho_i}: U_i\to \mathfrak{g}^*\otimes U_i$ by:
\begin{align*}
\widetilde{\rho_i}:\hspace{0.3cm}U_{i}\hspace{0.3cm}&\xrightarrow{\hspace*{2cm}} \hspace{0.4cm}\mathfrak{g}^*\otimes U_{i}\\
x\hspace{0.5cm}&\xmapsto{\hspace*{2cm}}\widetilde{\rho_i}(x):a\to\rho_{i,a}(x)\nonumber
\end{align*}
This map can be lifted to a degree 0 map $\mu_i:U_{i}\to U_{0}\otimes U_{i}$ by composition with $\Theta^*$:
\begin{center}
\begin{tikzpicture}
\matrix(a)[matrix of math nodes, 
row sep=5em, column sep=5em, 
text height=1.5ex, text depth=0.25ex] 
{&U_{0}\otimes U_{i}\\ 
U_{i}&\mathfrak{g}^*\otimes U_{i}\\}; 
\path[->](a-2-1) edge node[above left]{$\mu_i$}  (a-1-2); 
\path[->](a-2-2) edge node[right]{$\Theta^*\otimes\mathrm{id}$} (a-1-2);
\path[->](a-2-1) edge node[above]{$\widetilde{\rho_i}$} (a-2-2);
\end{tikzpicture}
\end{center}
Identifying $U_0\otimes U_i$ with $U_0\odot U_i$, the map $\mu_i$ satisfies item 5. of Definition \ref{wooo} at level $i+1$.
Then, let us define a degree 0 map $h_i$ by:
\begin{align*}
h_i:\hspace{0.2cm}U_{i}&\xrightarrow{\hspace*{1.7cm}} \hspace{0.5cm}S^2(U)_i\\
	x\hspace{0.2cm}&\xmapsto{\hspace*{1.7cm}}\mu_i(x)-\delta\circ\pi_{i-1}(x)
\end{align*}
and we extend it to all of $S(U)$ by derivation. The existence of a well-defined map $\delta_{i+1}:U_{i}\to U_{i+1}$ satisfying item 6. of Definition \ref{wooo} as well as the condition
$\delta_{i+1}\circ\delta_i=0$
is conditioned to these two inclusions:
\begin{equation*}\mathrm{Im}(h_i)\subset\mathrm{Ker}\big(\pi|_{S^2(U)_{i}}\big)\hspace{1cm}\text{and}\hspace{1cm}\mathrm{Im}\big(\delta_{i}\big)\subset\mathrm{Ker}(h_i)\end{equation*}
To show these, we need the two following identities:
\begin{align}
\pi\circ\mu_{i}&=\mu\circ\pi_{i-1}\label{donkey2}\\
\delta\circ\mu_{i-1}&=\mu_i\circ\delta_{i}\label{donkey1}
\end{align}
Their proof is technical and is given in Appendix \ref{appendicite}.


Then, the first inclusion is obtained as follows:
\begin{align}
\pi\circ h_i&=\pi\circ\mu_i-\pi\circ\delta\circ\pi_{i-1}\\
&=\pi\circ\mu_i-\mu\circ\pi_{i-1}+\delta\circ(\pi\circ\pi_{i-1})\\
&=0
\end{align}
where passing from the first line to the second line is done by using item 6. of Definition \ref{wooo} at level $i-1$, whereas passing from the second to the last line is done using item 4. of the same definition, together with Equation \eqref{donkey2}. On the other hand, the second inclusion is obtained as follows:
\begin{align}
h_i\circ\delta_i&=\mu_i\circ\delta_i-\delta\circ\pi_{i-1}\circ\delta_i\\
&=\mu_i\circ\delta_i-\delta\circ\mu_{i-1}+\delta\circ\delta\circ\pi_{i-2}\\
&=0
\end{align}
where passing from the first line to the second line is done by using item 6. of Definition \ref{wooo} at level $i-1$, whereas passing from the second to the last line is done by using Equation \eqref{donkey1}, together with the fact that $\delta$ is a differential on $S(U)$. This concludes the proof of the two inclusions.

Now, let us show that $h_i$ factors through $U_{i+1}$, i.e. that there exists a unique map $\delta_{i+1}:U_{i}\to U_{i+1}$ such that the following triangle is commutative:
\begin{center}
\begin{tikzpicture}
\matrix(a)[matrix of math nodes, 
row sep=5em, column sep=5em, 
text height=1.5ex, text depth=0.25ex] 
{U_{i}&U_{i+1}\\ 
S^2(U)_i&\\}; 
\path[left hook->](a-1-2) edge node[above left]{$\pi_{i}$}  (a-2-1); 
\path[->](a-1-1) edge node[above]{$\delta_{i+1}$} (a-1-2);
\path[->](a-1-1) edge node[left]{$h_{i}$} (a-2-1);
\end{tikzpicture}
\end{center}
We first define the map $\delta_{i+1}$. Let $v\in U_{i}$. Since $\mathrm{Im}(h_i)\subset\mathrm{Ker}(\pi|_{S^2(U)_{i}})$ and since $\mathrm{Ker}(\pi|_{S^2(U)_{i}})=\mathrm{Im}(\pi_{i})$,  then $h_i(v)\in\mathrm{Im}(\pi_{i})$. By injectivity of $\pi_{i}$, there exists a unique $u\in U_{i+1}$ such that $\pi_i(u)=h_i(v)$. Then we set:
\begin{equation}\label{eq:defdelta}
\delta_{i+1}(v)\equiv u
\end{equation}
This automatically implies that $\mathrm{Ker}(h_i)\subset\mathrm{Ker}(\delta_{i+1})$. By the inclusion $\mathrm{Im}(\delta_{i})\subset\mathrm{Ker}(h_i)$, we deduce that:
$$\mathrm{Im}\big(\delta_{i}\big)\subset\mathrm{Ker}(\delta_{i+1})$$
This allows to extend the chain complex $(U,\delta)$ one step further.

The $\mathfrak{h}$-equivariance of $\delta_{i+1}$ is guaranteed by the fact that $\mu_i$ and $\pi_{i}$ are both $\mathfrak{h}$-equivariant. Indeed, let $a\in\mathfrak{h}$, let $v\in U_i$, and let $u\in U_{i+1}$ be the (unique) image of $v$ through $\delta_{i+1}$ (as in Equation \eqref{eq:defdelta}).
By definition, there exists a unique $w\in U_{i+1}$ such that $\delta_{i+1}\big(\rho_{i,a}(v)\big)=w$. Let us show that $w=\rho_{i+1,a}(u)$ so that we will have:
\begin{equation}
\rho_{i+1,a}\big(\delta_{i+1}(v)\big)=\delta_{i+1}\big(\rho_{i,a}(v)\big)
\end{equation}
By definition of $w$, $h_i\big(\rho_{i,a}(v)\big)=\pi_i(w)$. But $\mu_{i}$, $\pi_i$ and the differential $\delta$ are $\mathfrak{h}$-equivariant, hence $h_{i}$ is $\mathfrak{h}$-equivariant as well, then we have:
\begin{equation}
\pi_i(w)=\rho_{i,a}\big(h_i(v)\big)=\rho_{i,a}\big(\pi_i(u)\big)=\pi_i\big(\rho_{i+1,a}(u)\big)
\end{equation}
Since the map $\pi_i$ is injective, we deduce that $w=\rho_{i+1,a}(u)$, proving the $\mathfrak{h}$-equivariance of $\delta_{i+1}$.

 By construction, the quadruple $\big((U_{k})_{0\leq k\leq i+1},(\delta_{k})_{1\leq k\leq i+1},(\pi_{k})_{0\leq k\leq i},(\mu_{k})_{0\leq k\leq i}\big)$ satisfies every axioms of Definition \ref{wooo}, hence it defines a $(i+1)$-stem of $\mathcal{V}$, and its $i$-truncation is $(U,\delta,\pi,\mu)$.
\end{proof}




\begin{example}
If $V$ is a Lie algebra, then the kernel of the symmetric bracket is the whole of $S^2(V)$, and $W=0$. Then, by induction, all spaces $U_i$ are zero, for all $i\geq1$. Then the $\infty$-stem associated to a Lie algebra is itself.
\end{example}

\subsection{Morphisms and equivalences of stems}\label{london}

In the former section, we gave the definition of stems associated to Lie-Leibniz triples, and proved that any Lie-Leibniz triple induces a stem. This existence result will be completed in this section by a unicity result on stems associated to the same Lie-Leibniz triple. First, let us define the notion of morphisms between two stems:
\begin{definition}\label{defi0}
Let $\mathcal{U}=(U,\delta,\pi,\mu)$  (resp. $\overline{\mathcal U}=(\overline{U},\overline{\delta},\overline{\pi},\overline{\mu})$) be a stem 
 associated to a Lie-Leibniz triple $(\mathfrak{g},V,\Theta)$ (resp. $(\overline{\mathfrak{g}},\overline{V},\overline{\Theta})$). A \emph{morphism of stems from $\mathcal{U}$ to $\overline{\mathcal{U}}$} is a couple $(\varphi,\Phi)$, where $\varphi:\overline{\mathfrak{g}}\to\mathfrak{g}$ is a Lie algebra morphism, and where $\Phi=(\Phi_k:U_k\to \overline{U}_k)_{k\geq0}$ is a family of degree 0 linear maps, such that:
\begin{enumerate}
\item the couple $(\varphi,\Phi_0^*)$ is a Lie-Leibniz triple morphism from $:(\overline{\mathfrak{g}},\overline{V},\overline{\Theta})$ to $(\mathfrak{g},V,\Theta)$;
\item  $\Phi$ is compatible with the respective actions of $\mathfrak{g}$ and $\overline{\mathfrak{g}}$, i.e. for every $k\geq0$ and $a\in\overline{\mathfrak{g}}$:
\begin{equation}\label{raoul}
\Phi_{k}\circ\rho_{k,\varphi(a)}=\overline{\rho}_{k,a}\circ\Phi_k
\end{equation}
\item when extended to $S(U)$ as a graded commutative algebra morphism, $\Phi$ intertwines $\pi$, $\overline{\pi}$, and $\delta$, $\overline{\delta}$.
\end{enumerate}
When $\mathcal{U}$ and $\overline{\mathcal{U}}$ are both $i$-stems, for some $i\in\mathbb{N}\cup\{\infty\}$,
we say that $(\varphi,\Phi)$ is an \emph{isomorphism of $i$-stems} if $(\varphi, \Phi_0^*)$ is an isomorphism of Lie-Leibniz triples, and if $\Phi_k:U_k\to\overline{U}_k$ is an isomorphism for every $0\leq k <i+1$.
 \end{definition}

Now let us turn to the study of some unicity questions arising from this definition. First, let us define the following notion of equivalence between two $i$-stems:
\begin{definition}\label{defi}
Let $i\in\mathbb{N}\cup\{\infty\}$, and let $\mathcal{U}$ and $\overline{\mathcal U}$ be two $i$-stems associated to the same Lie-Leibniz triple $\mathcal{V}=(\mathfrak{g},V,\Theta)$. Then $\mathcal{U}$ and $\overline{\mathcal U}$ are said \emph{equivalent} if there exists an isomorphism of $i$-stems $(\varphi,\Phi):\mathcal{U}\to\overline{\mathcal{U}}$ such that: 
\begin{enumerate}
\item $\varphi=\mathrm{id}_\mathfrak{g}$,
\item $\Phi_0=\mathrm{id}_{V^*}$, and
\item $\Phi_1=\mathrm{id}_{s(W^*)}$,
where $W$ is the bud of $\mathcal{V}$.
\end{enumerate}
\end{definition}
\noindent The definition is trivial for $i=0$ and $i=1$. For every $2\leq k<i+1$, it means that $U_k$ is isomorphic to $\overline{U}_k$, but there is more: item 1., together with Equation \eqref{raoul}, imply that the maps $\Phi_k:U_k\to \overline{U}_k$ are equivalence of $\mathfrak{g}$-modules, for every $2\leq k < i+1$.
This notion of equivalence is obviously an equivalence relation between $i$-stems. 
It turns out that the axioms of Definition \ref{wooo} are strict enough so that the following proposition holds:
\begin{proposition}\label{isomequiv}
For any $i\in\mathbb{N}\cup\{\infty\}$, two $i$-stems associated to the same Lie-Leibniz triple are equivalent. 
\end{proposition}
\begin{proof}
We construct this equivalence by first setting $\varphi\equiv\mathrm{id}_\mathfrak{g}$, $\Phi_0=\mathrm{id}_{V^*}$ and $\Phi_1=\mathrm{id}_{s(W^*)}$, as in Definition \ref{defi}. Then, we construct the other components of the linear map $\Phi$ by induction, so that the couple $(\varphi,\Phi)$ defines a morphism of stems. 
Under such a choice of maps $\varphi, \Phi_0$ and $\Phi_1$, item 1. of  Definition \ref{defi0} is automatically satisfied, whereas item 2. implies that the map $\Phi$ should be a mere $\mathfrak{g}$-equivalence. Item 3. is not modified. To show that there exists such a map $\Phi$ satisfying items 2. and 3. of Definition \ref{defi0}, we will do it in two steps, with the use of Lemmas \ref{isomequiv1} and \ref{isomequiv2}.
\end{proof}

\begin{lemme}\label{isomequiv1}
Let $i\in\mathbb{N}\cup\{\infty\}$ and let $\mathcal{U}=(U,\delta, \pi,\mu)$ and $\overline{\mathcal{U}}=(\overline{U},\overline{\delta}, \overline{\pi},\overline{\mu})$ be two $i$-stems associated to the same Lie-Leibniz triple $\mathcal{V}=(\mathfrak{g},V,\Theta)$.  Then, there exists a degree 0 linear mapping of graded vector spaces $\Phi:U\to \overline{U}$ such that:
\begin{enumerate}
\item $\Phi_0:V^*\to V^*$ and $\Phi_1:s(W^*)\to s(W^*)$ behave as the identity;
\item for every $k\geq2$, $\Phi_k:U_k\to \overline{U}_k$ is an equivalence of $\mathfrak{g}$-modules;
\item once extended to $S(U)$ as a graded commutative algebra homomorphism, the map $\Phi:S(U)\to S(\overline{U})$ intertwines $\pi$ and $\overline{\pi}$:
\begin{equation}
\overline{\pi}\circ\Phi=\Phi\circ\pi
\end{equation}
\end{enumerate}
\end{lemme}

\begin{proof}
We can assume that $i\geq2$ because the case $i=0$ and $i=1$ are trivial.
By item 1. of Definition \ref{wooo}, we know that $U_0=\overline{U}_0=V^*$ and that $U_1=\overline{U}_1=s(W^*)$. Then, set $\Phi_0:U_0\to \overline{U}_0$ and $\Phi_1:U_1\to \overline{U}_1$ to be the identity map. Item 3. of the same definition ensures that $\pi_0=\overline{\pi}_0=-\Pi_W^*\circ s^{-1}$. Then, identifying $S^2(U)_1=S^2(\overline{U})_1$ with $U_0\otimes U_1$, we have $\pi|_{S^2(U)_1}=\mathrm{id}\otimes\pi_0=\overline{\pi}|_{S^2(\overline{U})_1}$. Since $S^2(U)_1=S^2(\overline{U})|_1$, we have:
\begin{equation*}
\mathrm{Ker}\big(\pi|_{S^2(U)_1}\big)=\mathrm{Ker}\big(\overline{\pi}|_{S^2(\overline{U})_1}\big)
\end{equation*}

By item 4. of Definition \ref{wooo}, we know that $\pi_1:U_2\to\mathrm{Ker}\big(\pi|_{S^2(U)_1}\big)$ and that $\overline{\pi}_1:\overline{U}_2\to\mathrm{Ker}\big(\overline{\pi}|_{S^2(\overline{U})_1}\big)$ are bijective. Hence we conclude that $U_2$ and $\overline{U}_2$ are isomorphic through the linear map:
\begin{equation}\label{defphi1}
\Phi_2\equiv(\overline{\pi}_1)^{-1}\circ\pi_1:U_2\to \overline{U}_2
\end{equation}
Since it is defined from two $\mathfrak{g}$-equivariant maps, $\Phi_2$ is $\mathfrak{g}$-equivariant. Let us set $\Phi^{(2)}:S(U)\to S(\overline{U})$ to be the unique graded algebra homomorphism from $S(U)$ to $S(\overline{U})$ whose restriction on $U$ satisfies $\Phi^{(2)}_k\big|_{U_k}=\Phi_k$, for $0\leq k\leq 2$. 
We deduce from the definition of $\Phi^{(2)}$ that it intertwines $\pi_1$ and $\overline{\pi}_1$:
\begin{equation}\label{intertwinepi}
\overline{\pi}\circ\Phi^{(2)}=\Phi^{(2)}\circ\pi_1
\end{equation}

Now assume that the maps $\Phi_k:U_k\to \overline{U}_k$ have been defined for $0\leq k\leq j$ for some $j<i$, and let us construct $\Phi_{j+1}:U_{j+1}\to \overline{U}_{j+1}$. Following the induction hypothesis, we assume that the maps $\Phi_k$ are bijective and $\mathfrak{g}$-equivariant. We define $\Phi^{(j)}:S(U)\to S(\overline{U})$ to be the unique graded commutative algebra homomorphism whose restriction to $U$ satisfies $\Phi^{(j)}\big|_{U_k}=\Phi_k$, for every $0\leq k\leq j$. We also assume that $\Phi^{(j)}$ intertwines $\pi$ and $\overline{\pi}$ up to level $j$, i.e. that:
\begin{equation}\label{intertwinepi3}
\overline{\pi}\circ\Phi^{(j)}=\Phi^{(j)}\circ\pi
\end{equation}
holds on $S(U)_k$ for every $1\leq k\leq j$.

 We know from item 4. of Definition \ref{wooo} that the map $\pi_j:U_{j+1}\to S^2(U)_j$ (resp. $\overline{\pi}_j:\overline{U}_{j+1}\to S^2(\overline{U})_j$) is injective, and that its image coincides with $\mathrm{Ker}\big(\pi|_{S^2(U)_j}\big)$ (resp. $\mathrm{Ker}\big(\overline{\pi}|_{S^2(\overline{U})_j}\big)$). We only need to show that:
\begin{equation}\label{inclusion}
\Phi^{(j)}\Big(\mathrm{Ker}\big(\pi|_{S^2(U)_j}\big)\Big)=\mathrm{Ker}\big(\overline{\pi}|_{S^2(\overline{U})_j}\big)
\end{equation}
to define the map $\Phi_{j+1}$.
 Let $\lambda\in\mathrm{Ker}\big(\pi|_{S^2(U)_j}\big)$, then, by Equation \eqref{intertwinepi3}:
\begin{equation}
\overline{\pi}\circ\Phi^{(j)}(\lambda)=\Phi^{(j)}\circ\pi(\lambda)=0
\end{equation}
then $\lambda\in \mathrm{Ker}\big(\overline{\pi}|_{S^2(\overline{U})_j}\big)$. We show the reverse inclusion by the same trick, because $\Phi^{(j)}$ is invertible. Hence, we have the desired equality.  In particular, it implies that $U_{j+1}$ and $\overline{U}_{j+1}$ are necessarily isomorphic as vector spaces. This is also true even if both kernels reduce to zero, i.e. when $\pi|_{S^2(U)_j}$ and $\overline{\pi}|_{S^2(\overline{U})_j}$ are injective. In that case, $U_{j+1}=\overline{U}_{j+1}=0$.

 Thus we can define $\Phi_{j+1}$ by:
\begin{equation}\label{defphi2}
\Phi_{j+1}\equiv(\overline{\pi}_j)^{-1}\circ\Phi^{(j)}\circ\pi_j:U_{j+1}\to \overline{U}_{j+1}
\end{equation}
By construction, it is bijective and $\mathfrak{g}$-equivariant. Define $\Phi^{(j+1)}: S(U)\to S(\overline{U})$ to be to be the unique graded commutative algebra homomorphism whose restriction to $U$ satisfies $\Phi^{(j+1)}\big|_{U_k}=\Phi_k$, for every $0\leq k\leq j+1$. Then by construction we have:
\begin{equation}
\overline{\pi}\circ\Phi^{(j+1)}=\Phi^{(j+1)}\circ\pi
\end{equation}
This equation holds even in the case where $U_{j+1}=\overline{U}_{j+1}=0$, because in that case, $\pi_j=\overline{\pi}_j=0$ and $\Phi_{j+1}:U_{j+1}\to \overline{U}_{j+1}$ is the map that sends $0$ to $0$.
We have thus proven the existence of a map $\Phi^{(j+1)}$ that satisfies all the hypothesis of Lemma \ref{isomequiv1} at level $j+1$. Performing the induction up to level $i$ (or to infinity) proves the statement.
\end{proof}

\begin{lemme}\label{isomequiv2}
Let $i\in\mathbb{N}\cup\{\infty\}$ and let $\mathcal{U}=(U,\delta, \pi,\mu)$ and $\overline{\mathcal{U}}=(\overline{U},\overline{\delta}, \overline{\pi},\overline{\mu})$ be two $i$-stems associated to the same Lie-Leibniz triple $\mathcal{V}=(\mathfrak{g},V,\Theta)$. Then, the map $\Phi: U\to \overline{U}$ defined in Lemma \ref{isomequiv1} intertwines $\delta$ and $\overline{\delta}$:
\begin{equation}
\overline{\delta}\circ\Phi=\Phi\circ\delta
\end{equation}
\end{lemme}

\begin{proof} We can assume that $i\geq2$.
We already know from Lemma \ref{isomequiv1} that $\Phi$ intertwines $\pi$ and $\overline{\pi}$. Let us now show that it intertwines $\mu$ and $\overline{\mu}$. Obviously it is the case on $U_0$ and $U_1$ because in that case $\Phi$ is the identity map. Let $2\leq k<i+1$, and let $x\in U_0^*=V, \overline{\alpha}\in (\overline{U}_k)^*$ and $u\in U_k$, then:
\begin{align}
2\,\big\langle x\odot\overline{\alpha},\overline{\mu}_k\big(\Phi_k(u)\big)\big\rangle&=\big\langle \overline{\alpha},\rho_{k,\Theta(x)}\big(\Phi_k(u)\big)\big\rangle\\
\footnotesize\text{by $\mathfrak{g}$-equivariance of $\Phi$}\hspace{0.3cm}\normalsize&=\big\langle \overline{\alpha},\Phi_k\big(\rho_{k,\Theta(x)}(u)\big)\big\rangle\\
\footnotesize\text{by Equation \eqref{dualitysmooth}}\hspace{0.48cm}\normalsize&=\big\langle \Phi_k^*(\overline{\alpha}),\rho_{k,\Theta(x)}(u)\big\rangle\\
\footnotesize\text{by definition of $\mu_k$}\hspace{0.542cm}\normalsize&=2\,\big\langle x\odot \Phi_k^*(\overline{\alpha}),\mu_k(u)\big\rangle\\
\footnotesize\text{by Equation \eqref{dualitysmooth}}\hspace{0.48cm}\normalsize&=2\,\big\langle x\odot \overline{\alpha},\Phi\big(\mu_k(u)\big)\big\rangle
\end{align}
Thus, we can conclude that $\Phi$ intertwines $\mu$ and $\overline{\mu}$, that is:
\begin{equation}\label{intertwinemu}
\overline{\mu}_k\circ\Phi=\Phi\circ\mu_k
\end{equation}
for every $0\leq k<i+1$.

For $k=1$, we naturally have $\overline{\delta}_1\circ\Phi_0=\Phi_1\circ\delta_1$ because $\Phi_0$ and $\Phi_1$ are the identity maps on $U_0$ and $U_1$. For $k=2$, inspired by the proof of Theorem \ref{lemmestrand}, let us define $h_1=\mu_1-\delta\circ\pi_0$. Since $\mu_1=\overline{\mu}_1$, $\pi_0=\overline{\pi}_0$ and $\delta_1=\overline{\delta}_1$ because of item 3. of Definition \ref{wooo}, we can write $\delta_2:U_1\to U_2$ and $\overline{\delta}_2:\overline{U}_1\to \overline{U}_2$ as:
\begin{equation}\label{osef}
\delta_2=(\pi_1)^{-1}\circ h_1\hspace{1cm}\text{and}\hspace{1cm}\overline{\delta}_2=(\overline{\pi}_1)^{-1}\circ h_1
\end{equation}
We know from Equation \eqref{defphi1}, that $\Phi_2=(\overline{\pi}_1)^{-1}\circ\pi_1$. Applying the map to the expression of  $\delta_2$ in Equation \eqref{osef}, we have:
\begin{equation}
\Phi_2\circ\delta_2=(\overline{\pi}_1)^{-1}h_1=\overline{\delta}_2\circ\Phi_1
\end{equation}
This proves that $\Phi$ commutes with $\delta$ at level $k=2$.

Now, let us assume that $\Phi$ commutes with $\delta$ up to some level $1\leq j<i$, i.e. that for every $1\leq k\leq j$, we have:
\begin{equation}\label{onvayarriver}
\Phi_k\circ\delta_k=\overline{\delta}_k\circ\Phi_{k-1}
\end{equation}
This identity extends naturally to $S(U)$.  Set $h_j\equiv\mu_j-\delta\circ\pi_{j-1}$ and $\overline{h}_j\equiv\overline{\mu}_j-\overline{\delta}\circ\overline{\pi}_{j-1}$, and $\Phi^{(j)}:S(U)\to S(\overline{U})$ be the unique graded commutative algebra homomorphism whose restriction to $U$ satisfies $\Phi^{(j)}\big|_{U_k}=\Phi_k$, for every $0\leq k\leq j$. Since $\Phi$ commutes with $\pi$ (by definition), and with $\mu$ (as was just shown), we deduce the following equality:
\begin{equation}\label{bebop}
\Phi^{(j)}\circ h_j=\overline{h}_j\circ\Phi_j
\end{equation}
Moreover, we know by item 7. of Definition \ref{wooo} that we can write $\delta_{j+1}:U_j\to U_{j+1}$ and $\overline{\delta}_{j+1}:\overline{U}_j\to \overline{U}_{j+1}$ as:
\begin{equation}\label{osef2}
\delta_{j+1}=(\pi_j)^{-1}\circ h_j\hspace{1cm}\text{and}\hspace{1cm}\overline{\delta}_{j+1}=(\overline{\pi}_j)^{-1}\circ \overline{h}_j
\end{equation}
We know from Equation \eqref{defphi2}, that $\Phi_{j+1}=(\overline{\pi}_{j})^{-1}\circ\Phi^{(j)}\circ\pi_j$,  Applying this map to the expression of  $\delta_{j+1}$ in Equation \eqref{osef2}, and using Equation \eqref{bebop}, we have:
\begin{equation}
\Phi_{j+1}\circ\delta_{j+1}=(\overline{\pi}_j)^{-1}\circ\Phi^{(j)}\circ h_j=(\overline{\pi}_j)^{-1}\circ \overline{h}_j\circ\Phi_j=\overline{\delta}_{j+1}\circ\Phi_j
\end{equation}
Thus, we have proven that the map $\Phi$ commutes with $\delta$ at level $j+1$. We conclude the proof by induction.
\end{proof}

Proposition \ref{isomequiv} is a very strong result on $i$-stems: 
it defines an equivalence relation between every $i$-stems associated to the same Lie-Leibniz triple. Then, if a Lie-Leibniz triple admits a $i$-stem, it is `unique' in the sense that every other $i$-stem is isomorphic to this one.  
Now that we know that any two $i$-stems associated to the same Lie-Leibniz triple are equivalent, the question remains to find the `biggest' stem associated to a given Lie-Leibniz triple. For a clear statement, we need to define the following notions:

\begin{definition}
Let $\mathcal{V}$ be a Lie-Leibniz triple and let $i\in\mathbb{N}\cup\{\infty\}$. 
\begin{enumerate}
\item  We say that a $i$-stem $\mathcal{U}=(U,\delta,\pi,\mu)$ is \emph{caulescent} if the sequence $(U_k)_{0\leq k<i+1}$ does not converge to 0. In that case we say that \emph{$\mathcal{U}$ is of height $i$}.
\item We say that the caulescent $i$-stem $\mathcal{U}$ is \emph{robust} if there is no higher caulescent stem of which $\mathcal{U}$ is the $i$-truncation. 
\end{enumerate}
\end{definition}

\begin{remarque}
The condition that the sequence $(U_k)_k$ does not converge has a different meaning when $i\in\mathbb{N}$ or when $i=\infty$. In the first case, it means  that $U_i\neq 0$, whereas in the second case it means  that for every $I>0$ there exists some $i>I$ such that $U_i\neq0$. In regard of this, a caulescent $i$-stem $\mathcal{U}$ is robust either when $i=\infty$, or when there is no caulescent $l$-stem, for $l>i$ with $U_l\neq0$, that contains $\mathcal{U}$.
\end{remarque}

\begin{example}
A Lie algebra is a particular case of a Leibniz algebra that does not admit a symmetric bracket. Hence, the bud $W$ is the quotient of $S^2(\mathfrak{g})$ by itself, hence it is zero. From this, by induction we deduce that $S^2(U)_k=0$ for $k\geq0$. Hence the robust stem associated to $\mathfrak{g}$ is the 0-stem $(\mathfrak{g}^*,0,0,0)$.
\end{example}


Caulescence is a characteristics of stems that is obviously preserved by equivalence, but more importantly, robustness is as well:

\begin{proposition}\label{robustness}
Let $\mathcal{U}$ and $\overline{\mathcal{U}}$ be two equivalent $i$-stems (for $i\in\mathbb{N}\cup\{\infty\}$) associated to the same Lie-Leibniz triple $\mathcal{V}$. Then $\mathcal{U}$ is robust if and only if $\overline{\mathcal{U}}$ is robust.
\end{proposition}

\begin{proof}
Assume that $\mathcal{U}$ is robust and of height $i\in\mathbb{N}\cup\{\infty\}$. If $i$ is infinite, the proof is trivial because at each level $k\geq0$ we know that $U_k$ and $\overline{U}_k$ are isomorphic, then we can assume that $i\geq2$. Since $\mathcal{U}$ and $\overline{\mathcal{U}}$ are equivalent,  we know that $\overline{U}_i\simeq U_i$, so $\overline{\mathcal{U}}$ is caulescent. We have to show that it is robust.
 Suppose it is not the case, i.e. that $\overline{\mathcal{U}}$ is the $i$-truncation of some caulescent $j$-stem $\widetilde{\mathcal{U}}$  for some $j>i$. But then by Proposition \ref{isomequiv}, the $j$-stem $\mathcal{U}\oplus\bigoplus_{i+1\leq k\leq j}\{0\}$ would be equivalent to $\widetilde{\mathcal{U}}$. In particular, that would imply that $\widetilde{U_j}=0$, which is a contradiction.
\end{proof}
From Theorem \ref{lemmestrand}, Proposition \ref{isomequiv} and Proposition \ref{robustness}, we deduce the following fundamental result:

\begin{corollaire}\label{inftystem}
A Lie-Leibniz triple induces -- up to equivalence -- a unique robust stem.\end{corollaire}

\begin{proof}Given a Lie-Leibniz triple $\mathcal{V}=(\mathfrak{g},V,\Theta)$, if $V$ is a Lie algebra then its associated 0-stem is robust and unique. If it is not a Lie algebra, it admits at least a 1-stem by items 1., 2. and 3. of Definition \ref{wooo}, if not a $i$-stem for some $i>1$. Thus, let $\mathcal{U}$ be any $i$-stem associated to $\mathcal{V}$, for some $i\geq1$.
The proof then relies on the fact that one can always extend a given $i$-stem to a $(i+1)$-stem using Theorem \ref{lemmestrand}. 
We can apply this theorem again and again, to extend the stem to higher degrees. 
Going up to infinity, we obtain an $\infty$-stem $\mathcal{U}$.
Then, either it is a caulescent $\infty$-stem, or the sequence of $\mathfrak{g}$-modules $U_k$ converges to the zero vector space after some  rank $i_{max}$: $U_{i_{max}}\neq0$, and $U_k=0$ for every $k>i_{max}$. In that case a robust  stem associated to $\mathcal{V}$ is the truncation $\mathcal{U}'$ of $\mathcal{U}$ at level $i_{max}$. There is no caulescent stem associated to $\mathcal{V}$ that has a bigger height than $i_{max}$, for if we had another caulescent stem $\overline{\mathcal{U}}$ of height $j>i_{max}$, then by Proposition \ref{isomequiv} its $i_{max}$-truncation $\overline{\mathcal{U}}'$ would be equivalent to $\mathcal{U}'$, then by Proposition \ref{robustness}, $\overline{\mathcal{U}}'$ would be robust, so that necessarily $\overline{U}_j=0$, which contradicts the assumption that $\overline{\mathcal{U}}$ is caulescent. Thus, every robust stem associated to $\mathcal{V}$ have the same height $i_{max}$. Finally, equivalence is guaranteed by Proposition \ref{isomequiv}.\end{proof}

\subsection{Unveiling the tensor hierarchy algebra}
\label{sec:hierarchy}


We have shown in the last section that any Lie-Leibniz triple induces a $\infty$-stem. This structure will be at the core of the construction of tensor hierarchies. This section is devoted to showing how to build a tensor hierarchy algebra  from the data of any robust stem $\mathcal{U}=(U,\delta,\pi,\mu)$. We will first proof a Lemma that gives a graded Lie bracket on $s^{-1}(U^*)$ needed in the construction of the tensor hierarchy algebra, and then we built a tensor hierarchy algebra that satisfies all the axioms of Definition \ref{def:tensoralgebra} by construction.

Let us fix a Lie-Leibniz triple $\mathcal{V}=(\mathfrak{g},V,\Theta)$, and let $\mathcal{U}=(U,\delta,\pi,\mu)$ be the unique -- up to equivalence -- robust $i$-stem associated to it by Corollary \ref{inftystem}, where $i\in\mathbb{N}\cup\{\infty\}$. We can legitimately assume that $i\geq2$. 
Let $ T'$ be the dual space of the suspension of the graded vector space~$U$:
\begin{equation*}
T'\equiv s^{-1}(U^*)=\big(s(U)\big)^*
\end{equation*}
In other words, $ T'\equiv(T_{-k})_{1\leq k <i+2}$, with $T_{-1}=s^{-1}V$, $T_{-2}=s^{-1}(U_{1}^*)=s^{-2}W$, $T_{-3}=s^{-1}(U_2^*)$, and more generally:
\begin{equation*}
T_{-k}=s^{-1}(U_{k-1}^*)
\end{equation*}
for any $1\leq k<i+2$. Each vector space $T_{-k}$ is a $\mathfrak{g}$-module, since $U_{k-1}$ is a $\mathfrak{g}$-module. Indeed, the dual representation  of $\mathfrak{g}$ on $U_{k-1}$ induces an action of $\mathfrak{g}$ on $T_{-k}$ through a map $\eta_{-k}:\mathfrak{g}\to\mathrm{End}(T_{-k})$ that is defined by:
\begin{equation}\label{representationshifted}
\eta_{-k}\equiv s^{-1}\circ\rho_{k-1}^\vee\circ s
\end{equation}
where $\rho_{k-1}:\mathfrak{g}\to\mathrm{End}(U_{k-1})$ denotes the action of $\mathfrak{g}$ on $U_{k-1}$.

\begin{remarque}
Due to the suspension and desuspension operators, the contragredient representation of $\eta_k$ is defined by:
\begin{equation}\label{representationshifted2}
\eta_{-k}^\vee\equiv s\circ\rho_{k-1}\circ s^{-1}
\end{equation}
\end{remarque}


Let us now prove the following result:

\begin{lemme}\label{gradedlemma}
Let $\mathcal{V}=(\mathfrak{g},V,\Theta)$ be a Lie-Leibniz triple and let $\mathcal{U}=(U,\delta,\pi,\mu)$ be a robust $i$-stem associated to $\mathcal{V}$, where $i\in\mathbb{N}\cup\{\infty\}$. Then $T'\equiv s^{-1}(U^*)$ canonically inherits a robust graded Lie algebra structure of depth $i+1$. Moreover, the induced bracket is $\mathfrak{g}$-equivariant. 
\end{lemme}

\begin{proof}



If $i=0$ or $i=1$, then the proof is trivial, so we can suppose that $i\geq2$.
Let $T'=s^{-1}(U^*)$, i.e. $T_{-k}=s^{-1}(U_{k-1}^*)$ for every $1\leq k <i+2$. In particular $T'$ is of depth $i+1$.
Consider the space $s^{-1} T'=s^{-2}(U^*)$ which is the graded vector space $U^*$ whose elements have their degree shifted by $-2$. More precisely, for every $k\geq2$:
\begin{equation*}
(s^{-1} T')_{-k}\simeq(U_{k-2})^*
\end{equation*}
so that $(s^{-1}T')^*=s^2U$. Since the only modification is that the grading of has been shifted by the even number 2, the map $\pi:U\to S^2(U)$ induces a map $Q_\pi\equiv s^2\pi:s^2U\to S^2\big(s^2U\big)$ defined by:
\begin{equation}\label{bonjouuurq}
Q_\pi\equiv \big(s^2\odot s^2\big)\circ \pi\circ s^{-2}
\end{equation}
 This map can then be seen as a map from $(s^{-1}T')^*$ to $S^2\big((s^{-1}T')^*\big)$ that can be extended to all of $S\big((s^{-1} T')^*\big)$ by derivation. This symmetric algebra is the algebra of functions on $s^{-1} T'$, so that it turns out that $Q_\pi$ can be seen as a vector field on $s^{-1} T'$. For degree reasons, i.e. since the grading of $U$ has been shifted by 2, the degree of $Q_\pi$ is not $-1$ as the one of $\pi$, but it is $+1$. Moreover it is of arity 1 because $\pi$ is a map from $U$ to $S^2(U)$. And finally, the identity $(\pi)^{2}=0$ that holds on all of $S(U)$ implies that $Q_\pi$ is a homological vector field on the pointed graded manifold with fiber $s^{-1} T'$. In other words, $(s^{-1}  T',Q_\pi)$ is a pointed differential graded manifold. Then by Theorem \ref{correspondence}, we can use the correspondence between a homological vector field of degree $+1$ and of arity 1 on $s^{-1}T'$ and a graded Lie algebra structure on $T'$.

For any $u\in s^{-1} T'$, we define $\iota_u$ as the inner derivation of $S\big((s^{-1}T')^*\big)$ which satisfies, as in Equations \eqref{identification} and \eqref{identification2}:
\begin{align}
\iota_u(\alpha)&=\langle\alpha, u\rangle_{s^{-1}T'}\label{party}\\
\iota_v\iota_u(\alpha)&=2\,\langle\alpha, u\odot v\rangle_{S^2(s^{-1}T')}\label{party2}
\end{align} for any $\alpha\in (s^{-1} T')^*$. We have a natural identification $u\leftrightarrow\iota_u$, and thus by Theorem \ref{correspondence}, the graded Lie bracket $[\,.\,,.\,]'$ on $ T'=s^{-1}(U^*)$ is given by:
\begin{equation}\label{voronov2}
\iota_{s^{-1}[x,y]'}=(-1)^{|x|}\big[[Q_\pi,\iota_{s^{-1}(x)}],\iota_{s^{-1}(y)}\big]
\end{equation}
for all $ x,y\in T'$, and where on the right side, the bracket is the (graded) bracket of vector fields on the pointed graded manifold with fiber $s^{-1} T'$. 
The sign $(-1)^{|x|}$ in front of the term on the right hand side is necessary to enforce the graded skew symmetry of the bracket. Indeed, due to this sign, for any $x,y\in  T'$ we have:
\begin{equation}
[x,y]'=-(-1)^{|x||y|}[y,x]'
\end{equation}
This graded Lie bracket is of degree 0 
and the Jacobi identity is satisfied because it is equivalent to the fact that $\pi$ squares to zero. 
Moreover, by item 5. of Definition \ref{wooo}, the fact that the map $\pi_k$ is injective and that $\mathrm{Im}(\pi_k)=\mathrm{Ker}\big(\pi|_{S^2(U)_k}\big)$ for every $k\geq1$ implies that the graded Lie algebra structure on $T'$ is robust.

We now have to prove that the bracket $[\,.\,,.\,]'$ is $\mathfrak{g}$-equivariant. Let $k,l\geq1$ and let $x\in T_{-k}, y\in T_{-l}, u\in s(T_{-k-l}^*)=s^2(U_{k+l-1})$ and $a\in\mathfrak{g}$. We set $v=s^2\circ\rho_{k+l-1,a}\circ s^{-2}(u)$, so that we have on the one hand, by Equation \eqref{party}:
 \begin{align}
 \iota_{s^{-1}[x,y]'}(v)&=\Big\langle s^2\circ\rho_{k+l-1,a}\circ s^{-2}(u), s^{-1}[x,y]'\Big\rangle_{s^{-1}(T_{-k-l})}\label{varphi}\\
\footnotesize\text{by Equation \eqref{suspensiondecal}}\hspace{0.55cm}\normalsize &=\Big\langle \rho_{k+l-1,a}\circ s^{-2}(u), s\big([x,y]'\big)\Big\rangle_{s(T_{-k-l})}\\
\footnotesize\text{by definition of $\rho^\vee_{k+l-1}$}\hspace{0.35cm}\normalsize  &=-\Big\langle s^{-2}(u), \rho^\vee_{k+l-1,a}\circ s\big([x,y]'\big)\Big\rangle_{s(T_{-k-l})}\\
  \footnotesize\text{by Equation \eqref{suspensiondecal}}\hspace{0.55cm}\normalsize  &=-\Big\langle u, s^{-2}\circ\rho^\vee_{k+l-1,a}\circ s\big([x,y]'\big)\Big\rangle_{s^{-1}(T_{-k-l})}\\
 \footnotesize\text{by Equation \eqref{party}}\hspace{0.55cm}\normalsize &= -  \iota_{s^{-1}(\eta_{-k-l,a}([x,y]'))}(u)\label{finish1}
 \end{align}
We were allowed to use $\rho_{k+l-1}^\vee$ because $s(T_{-k-l})=U_{k+l-1}^*$.
On the other hand, from Equation \eqref{voronov2}, we have:
\begin{align}
 \iota_{s^{-1}[x,y]'}(v)&=(-1)^{kl+1}\iota_{s^{-1}(y)}\iota_{s^{-1}(x)}\circ Q_\pi(v)\label{varphi2}\\
\footnotesize\text{by Eq. \eqref{party2}}\hspace{0.2cm}\normalsize &=(-1)^{kl+1}2\,\big\langle Q_\pi(v),s^{-1}(x)\odot s^{-1}(y)\big\rangle_{S^2(s^{-1}T')_{-k-l-2}}\\
\footnotesize\text{by Eq. \eqref{bonjouuurq}}\hspace{0.2cm}\normalsize &= (-1)^{kl+1}2\,\Big\langle\big(s^2\odot s^2\big)\circ \pi\circ\rho_{k+l-1,a}\big(s^{-2}(u)\big) ,s^{-1}(x)\odot s^{-1}(y)\Big\rangle_{S^2(s^{-1}T')_{-k-l-2}}\\
\footnotesize\text{by Eq. \eqref{suspensiondecal2}}\hspace{0.2cm}\normalsize&= (-1)^{kl+1}2\,\Big\langle \pi\circ\rho_{k+l-1,a}\circ s^{-2}(u) ,s(x)\odot s(y)\Big\rangle_{S^2(sT')_{-k-l+2}}\\
\footnotesize\text{by $\mathfrak{g}$-equiv. of $\pi$}\hspace{0.05cm}\normalsize&= (-1)^{kl+1}2\,\Big\langle \rho_{k+l-2,a}\circ\pi\circ s^{-2}(u) ,s(x)\odot s(y)\Big\rangle_{S^2(sT')_{-k-l+2}}\\
\footnotesize\text{by def. of $\rho^\vee$}\hspace{0.3095cm}\normalsize&= (-1)^{kl}2\,\Big\langle \pi\circ s^{-2}(u) ,\rho^\vee_{a}\big(s(x)\odot s(y)\big)\Big\rangle_{S^2(sT')_{-k-l+2}}\label{varphi2bis}\\
\footnotesize\text{by Eq. \eqref{representationshifted}}\hspace{0.2cm}\normalsize&= (-1)^{kl}2\,\Big\langle \pi\circ s^{-2}(u) ,s\big(\eta_{-k,a}(x)\big)\odot s(y)+s(x)\odot s\big(\eta_{-l,a}(y)\big)\big)\Big\rangle_{S^2(sT')_{-k-l+2}}\\
\footnotesize\text{by Eq. \eqref{bonjouuurq}}\hspace{0.2cm}\normalsize&= (-1)^{kl}2\,\Big\langle  Q_\pi(u) ,s^{-1}\big(\eta_{-k,a}(x)\big)\odot s^{-1}(y)+s^{-1}(x)\odot s^{-1}\big(\eta_{-l,a}(y)\big)\big)\Big\rangle_{S^2(s^{-1}T')_{-k-l-2}}\\
\footnotesize\text{by Eq. \eqref{party2}}\hspace{0.2cm}\normalsize&=  (-1)^{kl}\iota_{s^{-1}[\eta_{-k,a}(x),y]'}\iota_{s^{-1}[x,\eta_{-k,a}(y)]'}\circ Q_\pi(u)\\
\footnotesize\text{by Eq. \eqref{voronov2}}\hspace{0.2cm}\normalsize&= -\iota_{s^{-1}([\eta_{-k,a}(x),y]'+[x,\eta_{-k,a}(y)]')}(u)\label{finish2}
\end{align} 
We were allowed to use $\rho_{k+l-2}^\vee$ because $S^2(sT')=S^2(U^*)$. Since the left-hand sides of Lines \eqref{varphi} and \eqref{varphi2} are the same, we deduce that Lines \eqref{finish1} and \eqref{finish2} are equal, which imply that the bracket $[\,.\,,.\,]'$ is $\mathfrak{g}$-equivariant:
\begin{equation}\label{jacobiaaa}
\eta_{-k-l,a}\big([x,y]'\big)=\big[\eta_{-k,a}(x),y\big]'+\big[x,\eta_{-k,a}(y)\big]'
\end{equation} 
This concludes the proof.
\end{proof}


Now we would like to use $ T'$ to define a tensor hierarchy algebra that would be associated to the Lie-Leibniz pair $(\mathfrak{g},V,\Theta)$. For this, we need to find a differential graded Lie algebra structure on $ T\equiv\mathfrak{h}\oplus T'$ satisfying all axioms of Definition \ref{def:tensoralgebra}.  Since Lemma \ref{gradedlemma} gives a robust graded Lie algebra structure on $s^{-1}(U^*)$, we first need to find a differential on $T'$ that is compatible with this bracket, before extending the differential graded Lie algebra structure to $T=\mathfrak{h}\oplus T'$. Obviously, a natural candidate to define the differential is the map $\delta$. More precisely we have:
\begin{theoreme}\label{prop:tensorhierarchy}
Let $\mathcal{V}$ 
 be a Lie-Leibniz triple, then there is a one-to-one correspondence between robust stems associated to $\mathcal{V}$ and tensor hierarchy algebras associated to $\mathcal{V}$.
\end{theoreme}


\begin{proof}
We will first show that any robust $i$-stem associated to $\mathcal{V}=(\mathfrak{g},V,\Theta)$, for $i\in\mathbb{N}\cup\{\infty\}$, canonically induces a tensor hierarchy algebra of depth $i+1$. The converse claim consists of taking the proof in the reverse direction.

First, if $i=0$ and $i=1$ the proof is trivial, so we can suppose that $i\geq2$. Let $\mathcal{U}=(U,\delta,\pi,\mu)$ be any  $i$-robust stem associated to $\mathcal{V}$. We will show that the graded vector space $ T\equiv\mathfrak{h}\oplus s^{-1}(U^*)$ canonically inherits a tensor hierarchy algebra structure. As in Lemma \ref{gradedlemma}, we set $T'= s^{-1}(U^*)$ and $T_0\equiv\mathfrak{h}$, so that $T$ is a negatively graded vector space of depth $i+1$.
  We have to find a bracket and a differential on $T$ that are compatible in the sense that they induce a differential graded Lie algebra structure on $T$, and such that they satisfy all items of Definition \ref{def:tensoralgebra}.

By Lemma \ref{gradedlemma}, we know that $ T'= s^{-1}(U^*)$ can be equipped with a robust graded Lie algebra structure, whose bracket $[\,.\,,.\,]'$ descends from the map $\pi$, and thus it is $\mathfrak{g}$-equivariant. We take this bracket as the restriction of $[\,.\,,.\,]$ to $T'\wedge T'$, so that item 4. of Definition \ref{def:tensoralgebra} is satisfied. After we have checked that this bracket satisfies item 5. of the same definition,  we will extend it to a bracket $[\,.\,,.\,]$ on all of $T$ that satisfies items 6. and 7. Then, we will define a differential on $T$ satisfying items 8. and 9., and finally, we will check its compatibility with the bracket.

Now, let us compute the restriction of $[\,.\,,.\,]'$ to $T_{-1}\wedge T_{-1}$ to check that it indeed satisfies item 5. of Definition \ref{def:tensoralgebra}. For any $x,y\in T_{-1}=s^{-1}V$ and any $u\in s^2(U_1)=s^3(W^*)$, Equation \eqref{voronov2} implies:
\begin{align}
\iota_{s^{-1}[x,y]'}(u)&=-\big[[Q_\pi,\iota_{s^{-1}(x)}],\iota_{s^{-1}(y)}\big](u)\\
&=-\iota_{s^{-1}(y)}\iota_{s^{-1}(x)}Q_\pi(u)\\
\footnotesize\text{by Equation \eqref{party2}}\hspace{0.35cm}\normalsize&=-2\,\big\langle Q_\pi(u),s^{-1}(x)\odot s^{-1}(y)\big\rangle_{S^2(s^{-2}V)}\\
\footnotesize\text{by Equation \eqref{bonjouuurq}}\hspace{0.35cm}\normalsize&=-2\,\Big\langle \big(s^2\odot s^2\big)\circ\big(-\Pi^*_W\big) \circ s^{-3}(u),s^{-1}(x)\odot s^{-1}(y)\Big\rangle_{S^2(s^{-2}V)}\\
\footnotesize\text{by Equation \eqref{suspensiondecal}}\hspace{0.35cm}\normalsize&=2\,\Big\langle \Pi^*_W\circ s^{-3}(u),s(x)\odot s(y)\big\rangle_{S^2(V)}\\
\footnotesize\text{by Equation \eqref{dualitysmooth}}\hspace{0.35cm}\normalsize&=2\,\Big\langle  s^{-3}(u), \Pi_W\big(s(x),s(y)\big)\big\rangle_{W}\\
\footnotesize\text{by Equation \eqref{suspensiondecal}}\hspace{0.35cm}\normalsize&=2\,\Big\langle  u, s^{-3}\circ\Pi_W\big(s(x),s(y)\big)\big\rangle_{s^{-3}W}\\
\footnotesize\text{by Equation \eqref{party}}\hspace{0.35cm}\normalsize&=\iota_{s^{-1}(2\,s^{-2}\circ\Pi_W(s(x),s(y)))}(u)
\end{align}
Hence, we deduce that at lowest order:
\begin{equation}\label{bracketdebase}
[x,y]'=2\,s^{-2}\circ\Pi_W\big(s(x),s(y)\big)
\end{equation}
as required by item 5. of Definition \ref{def:tensoralgebra}. Recall that this bracket is symmetric because $x$ and $y$ have degree $-1$.



Now, we will define a graded Lie bracket $[\,.\,,.\,]$ on $T=\mathfrak{h}\oplus  T'$  that restricts to $[\,.\,,.\,]'$ on $T'$, and that satisfies items 6. and 7. of Definition \ref{def:tensoralgebra}.
The Lie algebra $\mathfrak{h}$ comes equipped with its own Lie bracket, which is the restriction of the Lie bracket of $\mathfrak{g}$ to $\mathfrak{h}$. Thus, we define the  bracket $[\,.\,,.\,]$ on $\mathfrak{h}\wedge\mathfrak{h}$ by imposing that it matches the Lie algebra bracket of $\mathfrak{h}$:
\begin{equation}\label{bracket3}
[a,b]\equiv[a,b]_\mathfrak{h}
\end{equation}
so that item 6. of Definition \ref{def:tensoralgebra} is satisfied. Now we define the graded Lie bracket on $\mathfrak{h}\wedge T'$. Let $a\in\mathfrak{h}$ and $x\in T_{-k}$ (for $k\geq1$), then we set:
\begin{equation}\label{bracket2}
[a,x]\equiv \eta_{-k,a}(x)
\end{equation}
and impose that $[x,a]$ is $-[a,x]=-\eta_{-k,a}(x)$, where $\eta_{-k}:\mathfrak{g}\to \mathrm{End}(T_{-k})$ has been defined in Equation \eqref{representationshifted}. Thus, item 7. of Definition \ref{def:tensoralgebra} is satisfied. 

The bracket $[\,.\,,.\,]$ that we have defined should satisfy the Jacobi identity. First, by Proposition \ref{gradedlemma}, we know that the restriction of the bracket to $T'\wedge T'$ (where $T'=\bigoplus_{k\geq1} T_{-k}$) is a graded Lie bracket. Second, the restriction of the bracket to $\mathfrak{h}\wedge\mathfrak{h}$ satisfies the Jacobi identity because it coincides with the Lie bracket on $\mathfrak{h}$. Now, we have to show that the Jacobiator of the bracket $[\,.\,,.\,]$ vanishes on $\mathfrak{h}\wedge\mathfrak{h}\wedge T'$ and on $\mathfrak{h}\wedge T'\wedge T'$.
Let $a,b\in \mathfrak{h}$ and let $x\in  T_{-k}$, for some $k\geq1$, then the Jacobiator $\mathrm{Jac}(a,b,x)$ turns out to be zero because the Jacobi identity corresponds to the condition that the vector space $T_{-k}$ is a family of Lie algebra representations:
\begin{equation}
\big[a,[b,x]\big]+\big[b,[x,a]\big]+\big[x,[a,b]\big]=\eta_{-k,a}\circ\eta_{-k,b}(x)-\eta_{-k,b}\circ\eta_{-k,a}(x)-\eta_{-k,[a,b]}(x)=0
\end{equation}
 
 In order to show the last Jacobi identity, one just have to recall Equation \eqref{jacobiaaa} and to notice that when $a\in\mathfrak{h}$, it is equivalent to the fact that the Jacobiator $\mathrm{Jac}(a,x,y)$ is vanishing, since it can be rewritten as:
\begin{equation}
\big[a,[x,y]\big]=\big[[a,x],y\big]+\big[x,[a,y]\big]
\end{equation}
To conclude, the extended bracket $[\,.\,,.\,]$ satisfies the graded Jacobi identity on the whole of $T=\mathfrak{h}\oplus T'$, it is then a graded Lie bracket.

Let us now define the differential on $T$. First, 
the differential $\delta$ on the $\infty$-stem $\mathcal{U}$ induces a differential $\delta'$ on $s^2(U)$ as:
\begin{equation}\label{bonjouuurd}
\delta'_k\equiv s^2\circ\delta_k\circ s^{-2}
\end{equation}
for every $k\geq1$.
Then, let $\partial_{-k}:T_{-k-1}\to T_{-k}$ be the degree $+1$ map defined as in Equation \eqref{bracketbis} by:
\begin{equation}\label{bracket}
\iota_{s^{-1}(\partial_{-k}(x))}=-[\delta'_k,\iota_{s^{-1}(x)}]
\end{equation}
for every $x\in T_{-k-1}$, and $k\geq1$.
By duality, the maps $\partial_{-k}$ satisfy the homological condition $\partial_{-k}\circ\partial_{-k-1}=0$, so that we obtain a chain complex:
\begin{center}
\begin{tikzcd}[column sep=0.7cm,row sep=0.4cm]
0&\ar[l]T_{-1}&\ar[l,"\partial_{-1}"]T_{-2}&\ar[l,"\partial_{-2}"]T_{-3}&\ar[l]\ldots
\end{tikzcd}
\end{center}
Since $T_{-1}=s^{-1}V$ and $T_{-2}=s^{-2}W$, we deduce from Equation \eqref{bracket} that, for every $\alpha\in T_{-2}=s^{-2}W$ and $u\in s^2(U_0)=s^2(V^*)$, we have:
\begin{align}
\iota_{s^{-1}(\partial_{-1}(\alpha))}(u)&=-\iota_{s^{-1}(\alpha)}\circ\delta'_1(u)\\
\footnotesize\text{by Equation \eqref{party}}\hspace{0.5cm}\normalsize&=-\big\langle s^2\circ\delta_1\circ s^{-2}(u),s^{-1}(\alpha)\big\rangle_{s^{-3}W}\\
\footnotesize\text{by Equation \eqref{suspensiondecal}}\hspace{0.5cm}\normalsize&=-\big\langle s^{-1}\circ\delta_1\circ s^{-2}(u),s^2(\alpha)\big\rangle_{W}\\
\footnotesize\text{by definition of $\delta_1$}\hspace{0.633cm}\normalsize&=-\big\langle \mathrm{d}^*\circ s^{-2}(u),s^2(\alpha)\big\rangle_{W}\\
\footnotesize\text{by Eq. \eqref{dualitysmooth} and \eqref{suspensiondecal}}\hspace{0.2cm}\normalsize&=-\big\langle u,s^{-2}\circ\mathrm{d}\circ s^2(\alpha)\big\rangle_{s^{-2}V}\\
\footnotesize\text{by Equation \eqref{party}}\hspace{0.5cm}&=-\iota_{s^{-1}(s^{-1}\circ\mathrm{d}\circ s^2(\alpha))}(u)
\end{align}
Thus, we have at the lowest order:
\begin{equation}\label{differentielledebase}
\partial_{-1}=-s^{-1}\circ\dd \circ s^2
\end{equation}
where $\dd:W\to V$ is the collar of $\mathcal{V}$. This is consistant with item 9. of Definition \ref{def:tensoralgebra}.
Now, taking into account $T_0=\mathfrak{h}$, we define a linear map $\partial_0:T_{-1}\to T_0$ as: 
\begin{equation}\label{diff0}
\partial_{0}=-\Theta\circ s
\end{equation}
This map satisfies item 8. of Definition \ref{def:tensoralgebra}, as well as the homological condition $\partial_0\circ\partial_1=0$, by Proposition \ref{thetainclusion}. Thus we can extend the above chain complex to:
\begin{center}
\begin{tikzcd}[column sep=0.7cm,row sep=0.4cm]
0&\ar[l]T_{0}&\ar[l,"\partial_{0}"]T_{-1}&\ar[l,"\partial_{-1}"]T_{-2}&\ar[l,"\partial_{-2}"]T_{-3}&\ar[l]\ldots
\end{tikzcd}
\end{center}


\noindent In the following we will set $\partial\equiv(\partial_{-k})_{0\leq k}$; this family of maps defines a differential on $T$.

Let us summarize what we have obtained so far:
\begin{enumerate}
\item a (possibly infinite) graded vector space $ T=(T_{-i})_{i\geq 0}$ that satisfies items 1., 2. and 3. of Definition \ref{def:tensoralgebra};
\item a graded Lie algebra bracket $[\,.\,,.\,]$ on $T$ that satisfies items 4., 5., 6. and  7. of Definition~\ref{def:tensoralgebra};
\item a differential $\partial$ on $T$ that satisfies items 8. and 9. of Definition \ref{def:tensoralgebra}.
\end{enumerate}
Thus, the only thing that we have to show is that $[\,.\,,.\,]$ and $\partial$ are compatible in the sense that they induce a differential graded Lie algebra structure on $T$. Since the proof of this part, though conceptually very deep, is technical, we postpone it to Appendix \ref{appendicite2}. This concludes the proof that any robust $i$-stem induces a tensor hierarchy algebra of depth $i+1$.

The proof of the converse  consists essentially to taking the above proof in reverse direction, and construct $\delta,\pi$ and $\mu$ from the data contained in $(T,\partial,[\,.\,,.\,])$. This construction defines uniquely the corresponding stem $\mathcal{U}$. The fact that the sequence $(T_{-k})_{1\leq k<i+2}$ does not converge to the zero vector space ensures that $\mathcal{U}$ is caulescent. The fact that $(T',[\,.\,,.\,])$ is a robust graded Lie algebra ensures that $\mathcal{U}$ is robust. The depth of $T$, minus one, will be the height of $\mathcal{U}$.
\end{proof}

  The correspondence between robust stems and tensor hierarchy algebras is also valid at the morphism level:
  
  \begin{proposition}\label{isomequivmor}
  Let $\mathcal{U}$ (resp. $\overline{\mathcal{U}}$) be a robust stem associated to some Lie-Leibniz triple $\mathcal{V}$ (resp. $\overline{\mathcal{V}}$). Let $T$ (resp. $\overline{T}$) be the unique tensor hierarchy algebra induced by $\mathcal{U}$ (resp. $\overline{\mathcal{U}}$). Then there is a one-to-one correspondence between morphisms of stems from $\mathcal{U}$ to $\overline{\mathcal{U}}$, and tensor hierarchy algebra morphisms from $\overline{T}$ to $T$.
  \end{proposition}

\begin{proof}
Let $(\varphi,\Phi)$ be a morphism of stems from $\mathcal{U}$ to $\overline{\mathcal{U}}$. Then by definition, $\varphi: \overline{\mathfrak{g}}\to\mathfrak{g}$ is a Lie algebra morphism, and $\Phi=(\Phi_k:U_k\to \overline{U}_k)_{0\leq k <i+1}$ is a family of degree 0 linear maps satisfying all items of Definition \ref{defi0}, where $i\in\mathbb{N}\cup \{\infty\}$ is the height of $\mathcal{U}$. These data canonically induce a family of morphisms:
\begin{equation*}
\phi_0\equiv\varphi\big|_{\overline{\mathfrak{h}}}\hspace{1cm}\text{and}
\hspace{1cm}\phi_{-k}\equiv s^{-1}\circ\Phi^*_{k-1}\circ s : \overline{T}_{-k}\longrightarrow T_{-k}
\end{equation*}
for every $1\leq k<i+2$. The equation on the left is the first condition for $\phi$ to be a tensor hierarchy algebra morphism. Moreover, Equation \eqref{raoul}, together with Equation \eqref{representationshifted}, imply that for every $1\leq k<i+2$, the map $\phi_k$ satisfies Equation \eqref{insup}, as required.

We now have to show that the map $\phi$ is compatible with the respective differentials and brackets of $T$ and $\overline{T}$. We have to show that it is a (graded) Lie algebra morphism, and that it intertwines $\partial$ and $\overline{\partial}$. Since $\Phi$ intertwines $\pi$ and $\overline{\pi}$ (see item 2. of Definition \ref{defi0}), one can use the same strategy as in Equations \eqref{varphi}--\eqref{finish1} and \eqref{varphi2}--\eqref{varphi2bis} to prove that $\phi$ commutes with the graded Lie bracket on $\overline{T}'\wedge \overline{T}'$ and on $T'\wedge T'$, respectively. Since $\phi$ satisfies Equation \eqref{insup} for every $1\leq k<i+2$, it intertwines the brackets on $\mathfrak{h}\wedge \overline{T}'$ and $\mathfrak{h}\wedge T'$. On $\overline{T}_0$, $\phi_0=\varphi\big|_{\overline{\mathfrak{h}}}$ is a Lie algebra morphism, so it intertwines the Lie bracket of $\overline{\mathfrak{h}}$ and $\mathfrak{h}$. Thus, $\phi: \overline{T}\to T$ is a morphism of graded Lie algebras. Now, since $\Phi$ also intertwines $\delta$ and $\overline{\delta}$, one can use the same strategy as in Equations \eqref{varphi3bis}--\eqref{varphi3} to deduce that  $\phi$ intertwines the differentials $\partial$ and $\overline{\partial}$ on $\overline{T}'$. By item 1. of Definition \ref{defi0}, it obviously commutes with $\partial_0$. This proves that $\phi$ defines a morphism of differential graded Lie algebras between $\overline{T}$ and $T$ that moreover satisfies Equation \eqref{insup}. Hence, it is a tensor hierarchy algebra morphism. The proof of the converse statement consists of taking the proof in the reverse direction.
 \end{proof}

  Before concluding this section, let us turn to some unicity result. There is a natural notion of equivalence of tensor hierarchy algebras that are associated to the same Lie-Leibniz triple:
\begin{definition}\label{defi365}
Let  $T$ and $\overline{T}$ be two tensor hierarchy algebras of depth $i\in\mathbb{N}\cup\{\infty\}$, associated to the same Lie-Leibniz triple $\mathcal{V}=(\mathfrak{g},V,\Theta)$. Then $T$ and $\overline{T}$ are said \emph{equivalent} if there exists an isomorphism of tensor hierarchy algebras $(\varphi,\phi):T\to\overline{T}$ such that: 
\begin{enumerate}
\item $\varphi=\mathrm{id}_\mathfrak{g}$,
\item $\phi_{-1}=\mathrm{id}_{s^{-1}V}$, and
\item $\phi_{-2}=\mathrm{id}_{s^{-2}W}$, where $W$ is the bud of $\mathcal{V}$.
\end{enumerate}
\end{definition}
 \noindent This is an equivalence relation.   
 
 This definition allows us to deduce an important unicity result, 
 by using the one-to-one correspondence between robust stems and tensor hierarchy algebras:
\begin{corollaire}\label{ultimate}
A Lie-Leibniz triple induces -- up to equivalence -- a unique tensor hierarchy algebra.
\end{corollaire}

\begin{proof}
Let $T$ and $\overline{T}$ be two tensor hierarchy algebras, of respective depth $i$ and $\overline{i}$, associated to $\mathcal{V}$. Let $\mathcal{U}$ and $\overline{\mathcal{U}}$ be the corresponding robust stems, as given by Theorem \ref{prop:tensorhierarchy}. We know, by Corollary \ref{inftystem},  that $\mathcal{U}$ and $\overline{\mathcal{U}}$ are equivalent as stems. In particular, they have the same height, which implies that $i=\overline{i}$. Then, by Proposition \ref{isomequivmor}, the equivalence between $\mathcal{U}$ and $\overline{\mathcal{U}}$ induces a unique equivalence of tensor hierarchy algebras between $T$ and $\overline{T}$.
\end{proof}

We conclude this section by the following interesting result:

\begin{proposition}
Let $T$ (resp. $\overline{T}$) be a tensor hierarchy algebra associated to the Lie-Leibniz triple $\mathcal{V}=(\mathfrak{g},V,\Theta)$ (resp. $(\mathfrak{h}_V,V,\Theta_V)$). Then there exists a morphism of differential graded Lie algebras between $T$ and $\overline{T}$.
\end{proposition}

\begin{proof} Let $\mathfrak{h}=\mathrm{Im}(\Theta)$ and $W$ be the bud of $\mathcal{V}$. Let $\mathcal{U}$ and $\overline{\mathcal{U}}$ be the two robust stems corresponding to $T$ and $\overline{T}$, respectively.
Lemma \ref{propositionreve} gives us a Lie algebra morphism $\varphi:\mathfrak{h}\to\mathfrak{h}_V$, and Proposition \ref{propositionreve3} proves the existence of a map $\tau:W\to\bigslant{S^2(V)}{\mathrm{Ker}\big(\{\,.\,,.\,\}\big)}$ that is compatible with $\varphi$, see Equation \eqref{lalol}. By following the steps in the proofs of Lemmas \ref{isomequiv1} and \ref{isomequiv2}, one can construct a couple $(\varphi,\Phi)$ where $\Phi$ is a map from $\overline{\mathcal{U}}$ to $\mathcal{U}$ that satisfies all criteria of Definition \ref{defi0}, except that $\varphi$ is a map from $\mathfrak{h}$ to $\mathfrak{h}_V$, and not from the whole of $\mathfrak{g}$. Then, by slightly adapting the proof of Proposition \ref{isomequivmor}, we deduce that the data $(\varphi,\Phi)$ define a morphism of differential graded Lie algebras between $T$ and $\overline{T}$.
\end{proof}

\begin{remarque}
Interestingly, this results shows that, given a Leibniz algebra $V$, every tensor hierarchy algebras involving $V$ (i.e. associated to any Lie-Leibniz triple involving $V$) admits a differential graded Lie algebra morphism toward the unique -- up to equivalence -- tensor hierarchy algebra associated to the `standard' Lie-Leibniz triple $(\mathfrak{h}_V,V,\Theta_V)$. However it may not induce a morphism of tensor hierarchy algebras !
\end{remarque}

Thus we have shown in this paper that every Lie-Leibniz triple induces a unique tensor hierarchy algebra. This algebra coincides with the one that Jakob Palmkvist builds from Borcherds algebras \cite{palmkvist, palmkvist1}. Given that in supergravity models, the Bianchi identities induce a $L_\infty$-algebra structure on the (shifted) tensor hierarchy \cite{monpapier}, it would seem natural to understand how one passes from the tensor hierarchy algebra structure on $T=\mathfrak{h}\oplus T'$ to a $L_\infty$-algebra on $T'[-1]$. This is all the more important since $L_\infty$ algebras have recently drawn much interests in supergravity theories \cite{Henning2018, Hohm, Cagnacci, Cederwall, Henning2019}. This topic is indeed important because these $L_\infty$ structures encode the field strengths of the theory and their corresponding Bianchi identities. Hence, to deduce a $L_\infty$-algebra structure from a tensor hierarchy algebra structure would be very interesting because it would show that some physical information captured by the tensor hierarchy could be deduced by straightforward mathematical considerations. Also, on the mathematical side, this would be very interesting because it might provide a lifting of the skew-symmetric part of the Leibniz product to a $L_\infty$-algebra structure. This might be possible by applying a result by Fiorenza and Manetti \cite{Fiorenza} (that was found again later by Getzler \cite{getzler}) that states that a differential graded Lie algebra structure on $T=\mathfrak{h}\oplus T'$ induces a $L_\infty$-algebra structure on $T'[-1]$. This topic is still under investigation and may be the object of another paper.

\section{Examples}\label{sectionexamples}

\subsection{Differential crossed modules}

A differential crossed module is the data of two Lie algebras $\big(V,[\,.\,,.\,]_V\big)$ and $\big(\mathfrak{g},[\,.\,,.\,]_\mathfrak{g}\big)$, and two Lie algebra morphisms $\Theta:V\to \mathfrak{g}$ and $\rho:\mathfrak{g}\to\mathrm{End}(V)$ satisfying the following equations:
\begin{align}
[x,y]_V&=\rho_{\Theta(x)}(y)\\
\Theta\big(\rho_a(x)\big)&=\big[a,\Theta(x)\big]_\mathfrak{g}
\end{align}
for every $x,y\in V$ and $a\in\mathfrak{g}$. These data form a Lie Leibniz triple $(\mathfrak{g},V,\Theta)$ for which  the embedding tensor is $\mathfrak{g}$-equivariant (and not only $\mathfrak{h}=\mathrm{Im}(\Theta)$-equivariant). The tensor hierarchy algebra associated to it actually minimally depends  on $\mathfrak{g}$ and on $\Theta$ because $V$ is a Lie algebra. Hence, the symmetric bracket on $V$ is inexistent, and then it means that $\mathrm{Ker}\big(\{\,.\,,.\,\}\big)=S^2V$, which turns out to be a $\mathfrak{g}$-module. Then, the bud of $(V,\mathfrak{g},\Theta)$ is $\{0\}$. By induction, every other space of higher degree appearing in the construction of the stem associated to this Lie-Leibniz triple is zero. Hence the tensor hierarchy reduces to the following data:
\begin{align}
T&=\mathfrak{h}\oplus s^{-1}V\\
\partial_{0}&=-\Theta\circ s\\
[s^{-1}V,s^{-1}V]&=0\\
[a,\widetilde{y}]&=s^{-1}\big([x,s(\widetilde{y})]_V\big)\\
[a,b]&=[a,b]_\mathfrak{g}
\end{align}
for any $a=\Theta(x), b\in\mathfrak{h}$ and $\widetilde{y}\in T_{-1}=s^{-1}V$. The brackets on the left hand sides are the brackets on $T$. In the third line, the bracket of two elements of $T_{-1}=s^{-1}V$  is zero because there is no space concentrated in degrees lower than $-1$ since the bud of $V$ is zero. The last line does not depend on the pre-image $x$ of $a$ because $\mathrm{Ker}(\Theta)\subset\mathcal{Z}$. This is the typical tensor hierarchy that one obtains with differential crossed modules. It was one of the first examples of \emph{strict Lie $2$-algebras}, i.e. Lie $2$-algebras with vanishing 3-bracket \cite{Baez}. When one has only a Lie algebra $V$, then one can work with the differential crossed module $\big(\mathfrak{h}_V, V\Theta_V\big)$.

\subsection{A nilpotent Leibniz algebra}

Let us define a very simple example that is however illuminating. One can equip $\mathbb{R}^2$ with a Leibniz algebra structure. Let $a=(1,0)$ and $b=(0,1)$, and define a product $\bullet$ on $\mathbb{R}^2$ by first setting:
\begin{equation}
a\bullet a = b
\end{equation}
Requiring that $\bullet$ satisfies the Leibniz identity \eqref{leibnizidentity}, we deduce the following two other equations:
\begin{align}
a\bullet(a\bullet a)&=(a\bullet a)\bullet a + a\bullet (a\bullet a) \hspace{1cm}&&\Longrightarrow\hspace{1cm}b\bullet a=0\\
a\bullet(a\bullet b)&=(a\bullet a)\bullet b + a\bullet (a\bullet b) \hspace{1cm}&&\Longrightarrow\hspace{1cm}b\bullet b=0
\end{align}
Now, there is only one product left: $a$ acting on $b$. We set it to be zero:
\begin{equation}
a\bullet b=0
\end{equation}
Then, $\mathbb{R}^2$ equipped with this product $\bullet$ becomes a Leibniz algebra, that we call $V$. It has the particularity that the only non vanishing product is $a\bullet a =b$. This implies as well that the Leibniz product is symmetric. We call it a \emph{nilpotent Leibniz algebra}, because any combination of products of elements vanish after at most two  successive iterations.

Now let us find a Lie algebra $\mathfrak{g}$ and an embedding tensor $\Theta:V\to \mathfrak{g}$ such that: 1. $V$ is a $\mathfrak{g}$-module, and 2. $\Theta$ satisfies the linear constraint \eqref{eq:compat} and the quadratic constraint \eqref{eq:equiv}. We can assume that both $\mathfrak{g}$ and $\mathfrak{h}=\mathrm{Im}(\Theta)$ are Lie subalgebras of $\mathfrak{gl}_2(\mathbb{R})$, i.e. that the action $\rho:\mathfrak{g}\to\mathfrak{gl}_2(\mathbb{R})$ is an inclusion (hence the representation is faithful). 
Since $\mathfrak{h}$ should be a Lie algebra, but since the Leibniz product is symmetric, we know that $0=\Theta(a\bullet a)=\Theta(b)$. Then, the linear constraint implies that $\rho_{\Theta(a)}(a)=a\bullet a=b$, which implies in turn that:
\begin{equation}
\rho_{\Theta(a)}=\begin{pmatrix}
0 & 0 \\
1 & 0
\end{pmatrix}
\end{equation}
 and we note this matrix $A$. Since $\rho$ is injective we can assume that $\Theta(a)=A$, and we deduce that $\mathfrak{h}$ is the 1-dimensional Lie algebra generated by $A$. For now, we decide to choose $\mathfrak{g}$ to be the lower triangular $2\times 2$ matrices, so that $\mathfrak{h}$ is indeed a Lie subalgebra of $\mathfrak{g}$. The action of $\mathfrak{g}$ on $a$ is surjective on $\mathbb{R}^2$, whereas the image of the action of $\mathfrak{g}$ on $b$ is the sub-vector space of $\mathbb{R}^2$ spanned by $b$.
 
Let us now turn to defining the bud of $V$. Using the above notations, we can write $S^2(V)=\mathrm{Span}\big(a\odot a, a\odot b, b\odot b\big)$. Since the product $\bullet$ is symmetric, the kernel of the symmetric bracket is the subspace of $S^2(V)$ generated by $a\odot b$ and $b\odot b$. This is a $\mathfrak{g}$-module, so that the bud $W$ of $V$ is the 1-dimensional quotient $\bigslant{S^2(V)}{\mathrm{Ker}\big(\{\,.\,,.\,\}\big)}$. We denote by  $[a\odot a]$ its generator so that the collar $\mathrm{d}:W\to V$ sends $[a\odot a]$ to $b$. Then we have the usual factorization $\{\,.\,,.\,\} = \mathrm{d}\circ \Pi_W$, where $\Pi_W:S^2(V)\to W$ is the quotient map:
 \begin{center}
\begin{tikzpicture}
\matrix(a)[matrix of math nodes, 
row sep=5em, column sep=6em, 
text height=1.5ex, text depth=0.25ex] 
{&W\\ 
S^2(V)&V\\}; 
\path[->>](a-2-1) edge node[above left]{$\Pi_W$} (a-1-2); 
\path[->](a-1-2) edge node[right]{$\dd$} (a-2-2);
\path[->](a-2-1) edge node[above]{$\{\,.\,,.\,\}$} (a-2-2);
\end{tikzpicture}
\end{center}

Let us now set $U_0=V^*$, $U_1=s(W^*)$, and  define $\delta_1=s\circ\mathrm{d}^*$, $\pi_0=-\Pi_W^*\circ s^{-1}$. By noting $a^*$ and $b^*$ the respective dual elements of $a$ and $b$, and by $[a^*\odot a^*]$ the generator of $U_1$, this means in particular that:
\begin{equation}
\delta_1(a^*)=0,\hspace{1cm}\delta_1(b^*)=[a^*\odot a^*], \hspace{1cm}\text{and}\hspace{1cm}\pi_0\big([a^*\odot a^*]\big)=-a^*\odot a^*
\end{equation}
The proof for computing $\pi_0$ goes as follows: set $[a\odot a]^*$ be the dual element of $[a\odot a]$, and hence a generator of $W^*$. Then let us set $\Pi_W^*\big([a\odot a]^*\big)=\varpi\, a^*\odot a^* \in S^2(V^*)$. By formula \eqref{works2}, we have on the one side:
\begin{equation}
\big\langle\Pi_W^*\big([a\odot a]^*\big), a\odot a\big\rangle_{S^2(V)}=\varpi\,\big\langle a^*\odot a^*, a\odot a\big\rangle_{S^2(V)}=\varpi
\end{equation}
and on the other side:
\begin{equation}
\big\langle\Pi_W^*\big([a\odot a]^*\big), a\odot a\big\rangle_{S^2(V)}=\big\langle [a\odot a]^*, \Pi_W(a\odot a)\big\rangle_{W}=1
\end{equation}
Since both left hand sides are equal by duality, we deduce that $\varpi=1$. And then, applying item 3. of Definition \ref{wooo}, and noticing that $[a\odot a]^*=s^{-1}\big([a^*\odot a^*]\big)$, we deduce the correct formula for $\pi_0=-\Pi_W^*\circ s^{-1}:[a\odot a]^*\mapsto -a^\ast\odot a^\ast$.

Then,  extending $\pi_0$ as a derivation on $S^2(U_0\oplus U_1)$, we then compute $U_2=\mathrm{Ker}\big(\pi_0\big|_{S^2(U_0\oplus U_1)|_1}\big)$. We have an isomorphism $S^2(U_0\oplus U_1)|_1\simeq U_0\otimes U_1$ so that this space admits generators $a^*\otimes[a^*\odot a^*]$ and $b^*\otimes [a^*\odot a^*]$. The action of $\pi_0$ on these generators is:
\begin{align}
\pi_0\big(a^*\otimes[a^*\odot a^*]\big)&=-a^*\odot a^*\odot a^*\\
\pi_0\big(b^*\otimes[a^*\odot a^*]\big)&=-b^*\odot a^*\odot a^*
\end{align}
where both terms on the right hand side belong to $S^3(U_0)$ as expected. Hence $\pi_0$ is injective on $S^2(U_0\oplus U_1)|_1$, so $U_2=0$.

Now let us turn to find $U_3$. Since $U_2=0$, the only term that contributes in $S^2(U_0\oplus U_1)|_2$ is the 1-dimensional space $U_1\odot U_1$, with generator $[a^*\odot a^*]\odot[a^*\odot a^*]$. Obviously the action of $\pi_0$ on this element is not trivial, hence it is injective on $U_1\odot U_1$, and since $U_3=\mathrm{Ker}\big(\pi_0\big|_{U_1\odot U_1}\big)$, we deduce that $U_3=0$. Now notice that the fact that $U_2=0$ and $U_3=0$ imply that $S^2(U_0\oplus U_1)\big|_4=0$ which automatically implies that $U_4=0$, and so on for $U_5, U_6, etc.$ so that one deduces that $U_k=0$ for every $k\geq2$. 
Hence, the robust stem associated to the Lie-Leibniz triple $(\mathfrak{g},V,\Theta)$ defined in this sub-section is a 1-stem:
\begin{center}
\begin{tikzpicture}
\matrix(a)[matrix of math nodes, 
row sep=5em, column sep=5em, 
text height=1.5ex, text depth=0.25ex] 
{U_{0}&U_1\\ 
S^2(U_0)&\\}; 
\path[->](a-1-1) edge node[above]{$\delta_1$}  (a-1-2); 
\path[->](a-1-1) edge node[left]{$-\{\,.\,,.\,\}^*$} (a-2-1);
\path[->](a-1-2) edge node[below right]{$-\Pi_W^*\circ s^{-1}$} (a-2-1);
\end{tikzpicture}
\end{center}
where the map $\mu_0=-\{\,.\,,.\,\}^*$ indeed satisfies $\mu_0(a^*)=0$ and $\mu_0(b^*)=-a^*\odot a^*$, as expected.

Now let us determine the tensor hierarchy algebra structure associated to this stem. First, set $T_{-1}=s^{-1}V$ with generators $\widetilde{a}=s^{-1}a,\widetilde{b}=s^{-1}b$, $T_{-2}=s^{-2}W$ with generator $u=s^{-2}[a\odot a]$, and $T'= T_{-1}\oplus T_{-2}$.  Then the map $\pi_0$ defines a degree $+1$ homological vector field on $s^{-1}T'$ by Equation \eqref{bonjouuurq}. By Theorem \ref{correspondence}, this induces a graded Lie algebra structure on $T'$, with only one bracket, obtained from Equation \eqref{bracketdebase} and defined by:
\begin{equation}
[\widetilde{a},\widetilde{a}]=2u
\end{equation}
This is consistent with Equation \eqref{conditioncrochet}, and all other brackets vanish. This graded Lie algebra structure can be completed with a linear application $\partial_{-1}:T_{-2}\to T_{-1}$ whose action is obtained by Equation \eqref{differentielledebase}:
\begin{equation}
\partial_{-1}(u)=-\widetilde{b}
\end{equation}

This map would have played the role of a differential if it had satisfied the compatibility condition with the bracket.
It is not the case since $\partial_{-1} \big([\widetilde{a},\widetilde{a}]\big)\neq0$ on the one hand, whereas $\big[\partial_{-1} (\widetilde{a}),\widetilde{a}\big]=0$ on the other hand. Thus, to satisfy the Leibniz identity, we add the Lie algebra $\mathfrak{h}$ to $T'$ as a degree 0 vector space, and we set $T=\mathfrak{h}\oplus T'$. We then define $\partial_0$ as in Equation \eqref{diff0}, and it is compatible with $\partial_{-1}$ in the sense that $\partial_0\circ\partial_{-1}(u)=0$. The bracket between $\mathfrak{h}$ and $T'$ id defined as in Equation \eqref{bracket3}, whereas on $\mathfrak{h}$ we take the usual Lie bracket, which is zero since $\mathfrak{h}$ is nilpotent and 2-dimensional. Then, using these new definitions, the (unique) Leibniz identity is satisfied:
\begin{equation}
\partial_{-1} \big([\widetilde{a},\widetilde{a}]\big)=-2\widetilde{b}=-2\eta_{-1,\Theta\circ s(\widetilde{a})} (\widetilde{a})=\big[\partial_{0} (\widetilde{a}),\widetilde{a}\big]+(-1)^{|\widetilde{a}|}\big[\widetilde{a},\partial_{0} (\widetilde{a})\big]
\end{equation}
where $\eta_{-1}:\mathfrak{g}\to \mathrm{End}(s^{-1}V)$ is the representation of $\mathfrak{g}$ on $T_{-1}=s{-1}V$, defined by Equation \eqref{representationshifted}, and where $(-1)^{|\widetilde{a}|}=-1$ since the degree of $\widetilde{a}$ is $-1$.
Thus we have obtained the tensor hierarchy algebra associated to the Lie-Leibniz triple $(\mathfrak{g},V,\Theta)$. It consists of the following data:
\begin{align}
T&=\mathfrak{h}\oplus s^{-1}V\oplus s^{-2}W\\
\partial_{0}(\widetilde{a})&=A\\
\partial_{0}(\widetilde{b})&=0\\
\partial_{-1}(u)&=-\widetilde{b}\\
[\widetilde{a},\widetilde{a}]&=2u\\
[A,\widetilde{a}]&=\widetilde{b}
\end{align}
All other brackets being zero. One can check that these data satisfy all conditions of Definition \ref{def:tensoralgebra}.

\subsection{The $(1,0)$ superconformal model}

An example of a 2-stem arises from the six-dimensional $(1,0)$ superconformal model in six dimensions presented in full generality \cite{Henning2011}. Its mathematical aspects were investigated in \cite{Palmer, monpapier}. The symmetry algebra of this model is $\mathfrak{g}\equiv\mathfrak{e}_{5(5)}=\mathfrak{so}(5,5)$ \cite{Trigiante}. 
The model involves a set of $p$-forms (for $p=1,2,3,\ldots,6$) taking values, respectively, in the following $\mathfrak{g}$-modules: $V=\textbf{16}$, $W=\textbf{10}$, $X=\overline{\textbf{16}}$, $Y=\textbf{45}$, $Z=\overline{\textbf{144}}$ and $A=\textbf{10}\oplus\overline{\textbf{126}}\oplus\textbf{320}$ \cite{Trigiante}. These modules are defined from the representation constraint that sets $W$ and that is induced by supersymmetric considerations.  From this, all other spaces are uniquely defined. 

Notice that in supergravity, since supersymmetry provides a supplementary set of informations, the choice of gauge subalgebra need not be made at the beginning but at the very end of the construction. These physical considerations imply also that the choice of gauge algebra has no consequence on the modules $W, X, Y, Z$, and moreover that the dimension of the possible candidates for gauge algebras is constant. 
 This has two consequences: first, the gauge algebra $\mathfrak{h}\equiv\mathrm{Im}(\Theta)$ does not explicitly appear, see \cite{Bergshoeff, Henning2011}, and we rather work with a formal subalgebra $\mathfrak{h}_V$ (see below). Second, we construct a tensor hierarchy with abstract tensors, and then, at the very end, one fixes $\Theta$ and deduces the explicit form of the maps, as is done in \cite{Henning2011}.  Fixing the embedding tensor automatically fixes the gauge algebra : since it is done at the very end, the embedding tensor is considered as a \emph{spurionic object}. We will provide here the formal machinery and will hence do not bother on fixing $\Theta$, in the same spirit of the original paper \cite{Henning2011}, from which most notations are taken.

%
%

A priori the hierarchy is not constrained and goes to infinity, but since the space-time dimension is bounded, physicists are not interested tensor hierarchies of depth striclty higher than $6$. However, the computation are so complicated that usually Physicists stop the calculations at depth 3 or 4, and we will follow them on this point. Moreover, the particularity of this model is that the 3-form fields $C_t$ are dual to the 1-forms $A^a$. The  top (resp. bottom) indices are taken from the beginning (resp. the end) of the alphabet, to emphasize this duality. The reader who is not familiar with the $(1,0)$ superconformal model in six dimensions is advised to refer herself to \cite{Henning2011}, where this is discussed in full generality.
Due to the heavy calculations induced by the model, we will not present the whole hierarchy and restrain ourselves to the first orders. See \cite{Ortin} for an exposition of higher orders and \cite{Bergshoeff} for a more general discussion of supergravity models in $D=6$ dimensions.

The beginning of the hierarchy  is governed by a set of constants $h^a_I$, $g^{It}$, 
$f_{ab}^c \equiv f_{[ab]}^c$,  $d^I_{ab} \equiv d^I_{(ab)}$, $b^{}_{Ita}$, 
subject to the following relations:
\begin{align}
	2\big(d^{J}_{c(a}d^{I}_{b)s}-d^{I}_{cs}d^{J}_{ab}\big)h^{s}_{J}			&=2f_{c(a}{}^{s}d^{I}_{b)s}-b^{}_{Jsc}d^{J}_{ab}g^{Is}\label{lol1}\\
	\big(d^{J}_{rs}b^{}_{Iut}+d^{J}_{rt}b^{}_{Isu}+2d^{K}_{ru}b^{}_{Kst}\delta^{J}_{I}\big)h^{u}_{J}&=f_{rs}{}^{u}b^{}_{Iut}+f_{rt}{}^{u}b^{}_{Isu}+g^{Ju}b^{}_{Iur}b^{}_{Jst}\label{lol2}\\
	f_{[ab}{}^{r}f_{c]r}{}^{s}-\frac{1}{3}h_{I}^{s}d^{I}_{r[a}f_{bc]}{}^{r}&=0\label{lol3}\\
	h^{a}_{I}g^{It}													&=0\label{lol4}\\
	f_{rb}{}^{a}h_{I}^{r}-d^{J}_{rb}h^{a}_{J}h^{r}_{I}				&=0\label{lol5}\\
	g^{Js}h^{r}_{I}b^{}_{Ksr}-2h_{K}^{s}h_{I}^{r}d^{J}_{rs}				&=0\label{lol6}\\
	-f_{rt}{}^{s}g^{It}+d^{J}_{rt}h^{s}_{J}g^{It}-g^{It}g^{Js}b^{}_{Jtr}	&=0\label{lol7}\\
	b^{}_{Jt(a}d^{J}_{bc)}												&=0\label{lol8}
\end{align}	

The 1-forms $A^a$ take values in the $\mathfrak{g}$-module $V=\textbf{16}$. 
This $\mathfrak{g}$-module $V$ can be equipped with a Leibniz algebra structure whose generators are noted $X_a$. The Leibniz product is defined, for any $X_a,X_b\in V$ by:
\begin{equation}\label{conformaleibniz}
X_{a}\bullet X_b\equiv-X_{ab}{}^c X_c
\end{equation}
where $X_{ab}{}^c=-f_{ab}^c+d^I_{ab}h^c_I$ are the \emph{structure constants} of the Leibniz algebra. For consistency with Definition \ref{def:lieleibniz}, this action should coincide with the action of $\mathfrak{h}\equiv\mathrm{Im}(\Theta)$ on $V$:
\begin{equation}\label{return1}
\eta_{\Theta(X_a)} (X_b)\equiv-X_{ab}{}^c X_c
\end{equation}
The (skew)-symmetric brackets are then defined by:
\begin{equation}
[X_a,X_b]_V=f_{ab}{}^cX_c\hspace{1cm}\text{and}\hspace{1cm}\{X_a,X_b\}_V=-d^I_{ab}h_I^c X_c
\end{equation}
where $h_I^c$ is a tensor that corresponds to the collar $\mathrm{d}$ of the Lie-Leibniz triple $(\mathfrak{g},V,\Theta)$.

The 2-forms $B^I$ take values in $W=\textbf{10}$, which is a sub-representation of $S^2(V)$ and which is the bud of $V$. These fields are labelled by capital letters of the middle of the alphabet, and a set of generators of $W$ is noted $\{X_I\}$. The quotient map $\Pi_W:S^2(V)\to W$ is defined by:
\begin{equation}\label{convention2}
\Pi_W(X_a\odot X_b)=-d^I_{ab}X_I
\end{equation}
and this is consistant ith the fact that $\{\,.\,,.\,\}=\mathrm{d}\circ \Pi_W$. 
The action of $\mathfrak{g}$ on $W$ is defined by:
\begin{equation}\label{return2}
\eta_{\Theta(X_a)}(X_I)\equiv-X_{aI}{}^J X_J
\end{equation}
where $X_{aI}{}^J=2h_{I}^cd^J_{ac}-g^{Js}b^{}_{Isa}$, and where $X_a\in V$.
Going further up, we reach the 3-form fields $C_t$, taking values in $X=\overline{\textbf{16}}$. In the $(1,0)$ superconformal model, the 3-forms $C_t$ are dual to the 1-forms $A^a$, that is why we use latin letters of the end of the alphabet as labels. 
By duality, the action of $\mathfrak{g}$ on a generator $X^s$ of $X$ is defined by:
\begin{equation}\label{structure}
\eta_{\Theta(X_a)}(X^s)\equiv X_{at}{}^sX^t
\end{equation}
where $X_{at}{}^s=-f_{at}{}^s+d^I_{at}h^s_I$.

By setting $U_0=V^*$, $U_1=s(W^*)$ and $U_2=s^2(X^*)$, the maps of interest are written on the following diagram (the signs and the symbols can directly be read on the Bianchi identities of the field strengths in \cite{Henning2011}):
\begin{center}
\begin{tikzpicture}
\matrix(a)[matrix of math nodes, 
row sep=5em, column sep=4em] 
{U_0&U_1&U_2\\
U_0\odot U_0&U_0\odot U_1\\};
\path[left hook->](a-1-2) edge node[above left]{$d^{I}_{ab}$}  (a-2-1);
\path[->](a-1-1) edge node[above]{$h_I^a$}  (a-1-2);
\path[->](a-1-2) edge node[above]{$g^{It}$}  (a-1-3);
\path[left hook->](a-1-3) edge node[above left]{$-b^{}_{Ita}$} (a-2-2);
\end{tikzpicture}
\end{center}
In the following, we will define $X^a, X^I$ and $X_t$ as the respective shifted dual elements of $X_a, X_I$ and $X^t$, i.e. in the sense that $X^a$ has degree 0, $X^I$ has degree $+1$, $X_t$ has degree $+2$, etc. This means that $\{s^{-1}X^I\}$ is a basis of $W^*=s^{-1}U_1$ for example. Then we have:
\begin{align*}
\delta_1(X^a)=h^a_I\, X^I\hspace{1cm}&\text{and}\hspace{1cm}\delta_{2}(X^I)=g^{It}\,X_t,\\
\pi_0(X^I)=d^I_{ab}\,X^a\odot X^b\hspace{1cm}&\text{and}\hspace{1cm}\pi_1(X_t)=-b^{}_{Ita}\,X^I\odot X^a,\\
\mu_0(X^a)=d^{I}_{bc}h_{I}^a\,X^b\odot X^c\hspace{1cm}&\text{and}\hspace{1cm}\mu_1(X^I)=X_{aJ}{}^I\,X^a\odot X^J
\end{align*}
The expression for $\delta_1$ is the mere dual expression of the collar $\mathrm{d}$, whereas the expression for $\pi_0$ is a bit more intricate to find. Recall that $\pi_0$ is defined from the map $\Pi_W:W^*\to S^2(V^*)$ by Item 3. in Definition \ref{wooo}. So let us compute $\Pi_W^*$ in coordinates, given the expression of $\Pi_W$ in Equation \eqref{convention2}. We  set $\Pi_W^*(s^{-1}X^I)=M_{ef}^I\,X^e\odot X^f$ for some tensor $M_{ef}^c$, where the lower indices are symmetric. Then we have, by Equation \eqref{works2}:
\begin{equation}
\big\langle\Pi_W^*(s^{-1}X^I),X_a\odot X_b\big\rangle=\big\langle M_{ef}^I\,X^e\odot X^f,X_a\odot X_b\big\rangle=M_{ab}^I
\end{equation}
But, by Equation \eqref{convention2} we obtain on the other hand:
\begin{equation}
\big\langle\Pi_W^*(s^{-1}X^I),X_a\odot X_b\big\rangle=\big\langle s^{-1}X^I,\Pi_W(X_a\odot X_b)\big\rangle=-d_{ab}^I
\end{equation}
Thus, we have $M_{ab}^I=-d^I_{ab}$, which implies that $\Pi_W^*(s^{-1}X^I)=-d_{ab}^I\,X^a\odot X^b$. By comparing with the formula of $\pi_0$, this proves indeed that $\pi_0(X^I)=-\Pi_W^*(s^{-1}X^I)=d^I_{ab}X^a\odot X^b$.


Now let us show that Equations \eqref{lol3}$-$\eqref{lol8} encode all items of Definition \ref{wooo}, except the $\mathfrak{e}_{5(5)}$-equivariance of $\pi_0$ and $\pi_1$ which is implicit in the definition of $W$ and $X$, see for example the construction of tensor hierarchies in \cite{Cederwall2}.
 This $\mathfrak{e}_{5(5)}$-equivariance of $\pi_0$ and $\pi_1$ implies a $\mathfrak{h}$-equivariance (where $\mathfrak{h}$ would be the gauge subalgebra), this is the content of Equations \eqref{lol1} and \eqref{lol2}. 
Equation \eqref{lol3} corresponds to the Jacobi identity for the skew-symmetric bracket $[\,.\,,.\,]_V$, when one uses the tensors corresponding to $[\,.\,,.\,]_V$ and $\{\,.\,,.\,\}_V$ in Equations \eqref{jacobiator0} and \eqref{jacobiator}. Equation \eqref{lol4} corresponds to the condition $\delta_{2}\circ\delta_1=0$. Equations \eqref{lol5} and \eqref{lol6}  are implied by the fact that $\eta_{\Theta\circ\mathrm{d}(X_I)}=0$ on $V$ and on $W$ since, for example, multiplying the left hand side of Equation \eqref{lol6} by $X^K$ gives:
\begin{equation}
h_I^rX_{rK}{}^J\,X^K=h_I^r\eta_{\Theta(X_r)}(X^J)=\eta_{\Theta\circ\mathrm{d}(X_I)}(X^J)
\end{equation}
whose vanishing is induced by the homological condition $\Theta\circ \mathrm{d}$. 
Equation \eqref{lol5} can also be seen as the $\mathfrak{h}$-equivariance of $\delta_1$:
\begin{align}
\eta_{\Theta(X_b)}\big(\delta_1(X^a)\big)-\delta_1\big(\eta_{\Theta(X_b)}(X^a)\big)&=h^a_J\eta_{\Theta(X_b)}\big(X^J\big)-\delta_1\big(X_{br}{}^aX^r\big)\\
&=h^a_JX_{bI}{}^JX^I-h^r_IX_{br}{}^aX^I\\
&=\Big(h^a_J\big(2h^r_Id^J_{br}-g^{Jt}b_{Itb}\big)\nonumber\\
&\hspace{2cm}-h^r_I\big(-f_{br}{}^a+d^J_{br}h_J^a\big)\Big)X^I\\
&=\Big(h^a_Jh^r_Id^J_{br}-f_{rb}{}^ah^r_I\Big)X^I
\end{align}
where we used Equation \eqref{lol4} between the second and the third line, and the skew-symmetry of lower indices of $f_{br}{}^a$ between the third and the fourth line.
By the same line of arguments, Equation \eqref{lol7} symbolizes the $\mathfrak{h}$-equivariance of $\delta_{2}$, and Equation \eqref{lol8} is the condition $\pi^2|_{U_2}=0$:
\begin{equation}
\pi^2(X_t)=\pi\big(-b^{}_{Ita}X^I\odot X^a\big)=-b_{Ita}^{}d^I_{bc}\,X^b\odot X^c\odot X^a=-b_{It(a}^{}d^I_{bc)}\,X^a\odot X^b\odot X^c
\end{equation}
The fact that $\pi_1$ is injective, and that $\mathrm{Im}(\pi_1)=\mathrm{Ker}\big(\pi|_{U_0\odot U_1}\big)$ is guaranteed from physical considerations, see \cite{Henning2011, Cederwall2}. Now let us check that the condition that $\mu$ is a null-homotopic map at levels 0 and 1 is satisfied. First let us compute $\{\,.\,,.\,\}_V^*:U_0\to S^2(U^0)$ in coordinates. We set $\{\,.\,,.\,\}_V^*(X^c)=N_{ef}^c\,X^e\odot X^f$ for some tensor $N_{ef}^c$, where the lower indices are symmetric. Then we have, by Equation \eqref{works2}:
\begin{equation}
\big\langle\{\,.\,,.\,\}_V^*(X^c),X_a\odot X_b\big\rangle=\big\langle N_{ef}^c\,X^e\odot X^f,X_a\odot X_b\big\rangle=N_{ab}^c
\end{equation}
But, by Equation \eqref{worksz} and the definition of $\{\,.\,,.\,\}_V$, we have $N_{ab}^c=-d^I_{ab}h^c_I$. This implies that $\{\,.\,,.\,\}_V^*(X^c)=-h^c_Id_{ab}^I\,X^a\odot X^b$. By comparing with the formula of $\mu_0$, this proves indeed that $\mu_0=-\{\,.\,,.\,\}_V^*$.
Finally, to show that $\mu_1=\delta_1\circ\pi_0+\pi_1\circ\delta_2$, we compute straighforwardly:
\begin{equation}
\Big(\delta_1\circ\pi_0+\pi_1\circ\delta_2\Big)(X^I)=\delta_1\big(d^I_{ab}X^a\odot X^b\big)+\pi_1\big(g^{It}X_t\big)=\Big(2h^a_Jd^I_{ab}-g^{It}b_{Jtb}\Big)X^J\odot X^b
\end{equation}
and the parenthesis on the right hand side is indeed equal to $X_{bJ}{}^I$, as required.  Hence, all this set of maps and spaces form a 2-stem as defined in Definition \ref{wooo}.

As explained in \cite{Henning2011}, the hierarchy of differential forms $A^a, B^I, C_t$ can be extended one step further by adding a set of 4-forms $D_\alpha$ that take values in the $\mathfrak{g}$-module $Y=\textbf{45}$. 
 Three new tensors $k_t^\alpha, c^{}_{\alpha IJ}$ and $c^t_{\alpha a}$ have to be introduced  so that this extension is consistent. They obey a set of additional conditions:
\begin{align}
	g^{Kt}k_{t}^{\alpha}												&=0\label{lol9}\\
	4d^J_{ab} c^{}_{\alpha IJ} - b^{}_{Ita} c^t_{\alpha b} - b^{}_{Itb} c^t_{\alpha a}
	&=0\label{lol10}\\
	k_t^\alpha c^{}_{\alpha IJ} -h^a_{[I}b^{}_{J]ta} 
	&=0\label{lol11}\\
	k_t^\alpha c^s_{\alpha a} - f_{ta}{}^s + b^{}_{Jta}g^{Js} - d^J_{ta}h^s_J
	&=0\label{lol12}
\end{align}

\noindent By setting $U_3=s^3(Y^*)$, the corresponding 3-stem is as follows (the signs are obtained from the Bianchi identities given in \cite{Henning2011}):

\begin{center}
\begin{tikzpicture}
\matrix(a)[matrix of math nodes, 
row sep=5em, column sep=4em] 
{U_0&U_1&U_2&U_3\\
U_0\odot U_0&U_0\odot U_1&\begin{array}{ll}U_1\odot U_1\\\oplus \ U_0\odot U_2\end{array}\\};
\path[left hook->](a-1-2) edge node[above left]{$d^{I}_{ab}$}  (a-2-1);
\path[->](a-1-1) edge node[above]{$h_I^a$}  (a-1-2);
\path[->](a-1-2) edge node[above]{$g^{It}$}  (a-1-3);
\path[left hook->](a-1-3) edge node[above left]{$-b^{}_{Ita}$} (a-2-2);
\path[->](a-1-3) edge node[above]{$k_{t}^\alpha$} (a-1-4);
\path[left hook->](a-1-4) edge node{$-c^s_{\alpha a}+c^{}_{\alpha IJ}$}  (a-2-3);
\end{tikzpicture}
\end{center}

\noindent The new maps $\delta_2, \pi_2$ and $\mu_2$ are:
\begin{align*}
\pi_2(X_\alpha)&=-c_{\alpha a}^t\,X_t\odot X^a +c^{}_{\alpha IJ}\, X^ I\odot X^J,\\
\delta_{3}(X_t)&=k^\alpha_t\, X_\alpha\hspace{1cm}\text{and}\hspace{1cm}\mu_2(X_t)=- X_{at}{}^s\,X^a\odot X_s
\end{align*}
where $X_\alpha$ is the dual of $X^\alpha$.
The presence of a minus sign in the definition of $\mu_2$ was expected because the index labelling the 3-forms is at the bottom. The space $U_3=s^3(Y^*)$ can be seen as a sub-module of $(V^*\otimes X^*)\oplus (W^*\odot W^*)$, when identified  with the kernel of the map $\pi|_{U_1\odot U_1\oplus  U_0\odot U_2}$.

Equation \eqref{lol9} corresponds to the homological condition $\delta_{2}\circ\delta_{1}=0$, and Equation \eqref{lol10} corresponds to the condition $\pi^2|_{U_3}=0$:
\begin{align}
\pi^2(X_\alpha)&=-c_{\alpha a}^t\,\pi\big(X_t\odot X^a\big) +c^{}_{\alpha IJ}\,\pi\big( X^ I\odot X^J\big)\\
&=-c_{\alpha a}^t\,\big(-b^{}_{Itb}\,X^I\odot X^b\odot X^a-0\big)+2c^{}_{\alpha IJ}d^I_{ab}\,X^a\odot X^b\odot X^J\\
&=\Big(b^{}_{It(a|}c^t_{\alpha |b)}-2c^{}_{\alpha IJ}d^J_{ab}\Big)\,X^a\odot X^b\odot X^J
\end{align}
And the term in parenthesis is indeed the left hand side of Equation \eqref{lol10}.
 Equation \eqref{lol12} can be written as $ f_{at}{}^s- d^J_{at}h^s_J=-k_t^\alpha c^s_{\alpha a} - b^{}_{Jta}g^{Js}$. The left hand side is the structure constant $-X_{at}^s$ of the contragredient action of $\mathfrak{h}$ on $X^*$ and it can be seen as the map $\mu_2:U_3\to U_0\odot U_2$, whereas the right hand side corresponds to applying $\pi_2\circ\delta_{3}+(\delta_1\otimes \mathrm{id}_{U_1}+\mathrm{id}_{U_0}\otimes \delta_2)\circ\pi_{1}$ and taking the corestriction to $U_0\odot U_2$. Equation \eqref{lol11} corresponds to the fact that the corestriction of $\pi_2\circ\delta_{3}+(\delta_1\otimes \mathrm{id}_{U_1}+\mathrm{id}_{U_0}\otimes \delta_2)\circ\pi_{1}$ to $U_1\odot U_1$ is always zero, by construction. Hence, Equations \eqref{lol11} and \eqref{lol12} correspond to the null-homotopic condition $\mu=[\delta,\pi]$ at level~2. Finally, as a side remark we notice that Equation \eqref{lol7} is obtained by contracting Equation \eqref{lol12} with $g^{It}$.
Hence, by setting $\delta=(\delta_k)_{1\leq k\leq 3}$, $\pi=(\pi_k)_{0\leq k \leq 2}$ $\mu=(\mu_k)_{0\leq k \leq 2}$, and $U=(U_{k})_{0\leq k\leq 3}$, we observe that $(U,\delta,\pi,\mu)$ is a 3-stem over the Lie-Leibniz triple $\big(\mathfrak{e}_{5(5)},V,\Theta\big)$, where $\Theta: V\to \mathfrak{e}_{5(5)}$ is to be fixed later. It is not a proper robust stem as such because we should push the computations to higher levels, but physicists did not go further so we shall stop here, having in mind that theoretically the process does not meet any obstacle to build a robust stem.

 That is why we will use this 3-stem to  build the beginning of the tensor hierarchy algebra that is associated to the $(1,0)$ superconformal model in six dimensions. 
We define $T_{-1}\equiv s^{-1}V=s^{-1}(U_0^*)$, $T_{-2}\equiv s^{-2}W=s^{-1}(U_1^*)$, $T_{-3}\equiv s^{-3}X=s^{-1}(U_2^*)$ and $T_{-4}\equiv s^{-4}Y=s^{-1}(U_3^*)$, so that $T_{-k}$ can be considered as a space of degree $-k$, as desired. We finally set $T'\equiv(T_{-k})_{1\leq k\leq 4}$.
Let us now define basis for $T'$: a basis of $T_{-1}$ is given by the elements $e_a\equiv s^{-1}(X_a)$, a basis of $T_{-2}$ is given by the elements $e_I\equiv s^{-2}(X_I)$, a basis of $T_{-3}$ is given by the elements $e^t\equiv s^{-3}(X^t)$ and a basis of $T_{-4}$ is given by the elements $e^\alpha\equiv s^{-4}(X^\alpha)$.
Let us now turn to the application of Lemma \ref{gradedlemma}. We have to show that the graded vector space $T'$ can be equipped with a bracket that satisfies the Jacobi identity (at least for the Jacobiators taking values in $T'$). The idea is to show that $s^{-1}T'=(s^{2}U)^*$ is a $Q$-manifold. 
We set $u^a\in (s^{-1}T_{-1})^*=s^2U_0$ the dual coordinate of $s^{-1}e_a$, $u^I\in (s^{-1}T_{-2})^*=s^2U_1$ the dual coordinate of $s^{-1}e_I$, $u_t\in (s^{-1}T_{-3})^*=s^2U_2$ the dual coordinate of $s^{-1}e^t$ and $u_\alpha\in (s^{-1}T_{-4})^*=s^2U_3$ the dual coordinate of $s^{-1}e^\alpha$.
That is to say, we have the following duality relations:
\begin{align}
\iota_{s^{-1}e_b}(u^a)&=\delta_b^a\\
\iota_{s^{-1}e_J}(u^I)&=\delta_J^I\\
\iota_{s^{-1}e^s}(u_t)&=\delta_t^s\\
\iota_{s^{-1}e^\beta}(u_\alpha)&=\delta^\beta_\alpha
\end{align}
where the $\delta$'s are Kronecker's symbols. In particular we can make the following formal identifications:
\begin{align}
\frac{\partial}{\partial u^a}&\longleftrightarrow\iota_{s^{-1}e_a}\\
\frac{\partial}{\partial u^I}&\longleftrightarrow\iota_{s^{-1}e_I}\\
\frac{\partial}{\partial u_t}&\longleftrightarrow\iota_{s^{-1}e^t}\\
\frac{\partial}{\partial u_\alpha}&\longleftrightarrow\iota_{s^{-1}e^\alpha}
\end{align}

Then, since $\pi: U\to S^2(U)$ is a degree $-1$ map, it canonically induces a degree $+1$ map $Q_\pi:s^{2}U\to S^2\big(s^{2}U\big)$, that we can extend as a derivation to $S\big(s^{2}U\big)=S\big((s^{-1}T')^*\big)$ as follows:
\begin{align}
	Q_\pi (u^{a})		&=0\\
	Q_\pi (u^{I})		&=d^{I}_{bc}u^{b}\odot u^{c}\\
	Q_\pi (u_{t})	\hspace{0.05cm}	&=-b^{}_{Ita}u^{I}\odot u^{a}\\
	Q_\pi (u_{\alpha})	&=c^{}_{\alpha IJ}u^{I}\odot u^{J}-c_{\alpha a}^{t}u_{t}\odot u^{a}
\end{align}
By Theorem \ref{correspondence}, this induces the following bilinear bracket on $T'$:
\begin{align}
	[e_{a},e_{b}]	&=	-2\,d^{I}_{ab}\,e_{I}\label{brak1}\\
	[e_{a},e_{I}]	&=	-b^{}_{Ita}\,e^{t}\label{brak2}\\
	[e_{I},e_{J}]	&=	-2\,c^{}_{\alpha IJ}\,e^{\alpha}\label{brak3}\\
	[e_{a},e^{t}]	&=	c_{\alpha a}^{t}\,e^{\alpha}\label{brak4}
\end{align}
In particular, by Equation \eqref{convention2}, the bracket $[e_a,e_b]$ satisfies Equation \eqref{conditioncrochet}.
Since the degree of $e_a, e_b$ is $-1$, their bracket is symmetric, whereas the bracket of $e_I, e_J$ is skew-symmetric, for they have degree $-2$. To be more precise,  one obtains for example the bracket $[e_{I},e_{J}]$ by the following calculation:
\begin{align}
\iota_{s^{-1}[e_{I},e_{J}]}&=(-1)^{-2}\big[[Q_\pi,\iota_{s^{-1}e_I}],\iota_{s^{-1}e_J}\big]\\
&=\big[\iota_{s^{-1}e_I}Q_\pi,\iota_{s^{-1}e_J}\big]\\
&=2c^{}_{\alpha IK}\,\big[u^K\frac{\partial}{\partial u_\alpha},\iota_{s^{-1}e_J}\big]\\
&=-2c^{}_{\alpha IK}\delta^K_J\, \iota_{s^{-1}e^\alpha}\\
&=\iota_{s^{-1}(-2c^{}_{\alpha IJ}e^\alpha)}
\end{align}
An other example is:
\begin{align}
\iota_{s^{-1}[e_{a},e_{I}]}&=(-1)^{-1}\big[[Q_\pi,\iota_{s^{-1}e_a}],\iota_{s^{-1}e_I}\big]\\
&=-\big[-\iota_{s^{-1}e_a}Q,\iota_{s^{-1}e_I}\big]\\
&=-b_{Kta}\,\big[u^K\frac{\partial}{\partial u_t},\iota_{s^{-1}e_I}\big]\\
&=-b_{Kta}\delta_I^K\iota_{s^{-1}e^t}\\
&=\iota_{s^{-1}(-b_{Ita}e^t)}
\end{align}

This bracket satisfies the graded Jacobi identity on $T_{-1}\otimes T_{-1}\otimes T_{-1}$ (resp. $T_{-1}\otimes T_{-1}\otimes T_{-2}$) because the corresponding Jacobiator takes values in $T_{-3}$ (resp. $T_{-4}$) and identically vanishes since it is equivalent to the homological conditions $\pi^2\big|_{U_2=0}$ (resp. $\pi^2\big|_{U_3=0}$).
However, every other Jacobiator has a degree strictly lower than $-4$, and thus cannot be defined since  $T'$ has be defined only up to degree $-4$. For example to compute the Jacobiator on $T_{-2}\otimes T_{-2}\otimes T_{-1}$, one need the bracket on $T_{-1}\otimes T_{-4}$ to be defined, which has not been done because one needs to define the 4-stem associated to $\big(\mathfrak{e}_{5(5)},V,\Theta\big)$ before. However in the case that we had extended the study to degree $-5$ and $-6$, Lemma \ref{gradedlemma} ensures that the Jacobi identities would be satisfied at these levels. The robustness condition on $T'$ is satisfied by construction of the modules $W,X,Y$ in supergravity theories, see the discussion of the construction of the tensor hierarchy in \cite{Cederwall2}.

Recall that up to now the Leibniz algebra structure defined on $V$ by Equation \eqref{conformaleibniz} is formal, and so is the center $\mathcal{Z}$ of $V$ and thus the quotient $\mathfrak{h}_V\equiv\bigslant{V}{\mathcal{Z}}$. These data are uniquely fixed as soon as one chooses a specific embedding tensor $\Theta:V\to \mathfrak{e}_{5(5)}$, that defines a gauge algebra $\mathfrak{h}\subset \mathfrak{e}_{5(5)}$. In supergravity theories, this is usually done at the end of the calculations.
In the $(1,0)$ superconformal model in six dimensions, the $\mathfrak{e}_{5(5)}$-module $V$ is the Majorana-Weyl spinor representation of $\mathfrak{e}_{5(5)}$, hence it is faithful, see\cite{Bergshoeff}. This implies by Lemma \ref{propositionreve} that any choice of gauge algebra $\mathfrak{h}$ is isomorphic to the algebra $\mathfrak{h}_V$.  
Hence, even if we do not have made a choice for a gauge algebra yet, we can formally continue the construction of the tensor hierarchy algebra by setting $T_0\equiv \mathfrak{h}_V$, and $T\equiv(T_{-k})_{0\leq k\leq3}$. We will not go to lower degrees because the fields taking values in these spaces have not been defined in \cite{Henning2011}. Since the embedding tensor $\Theta_V:V\to \mathfrak{h}_V$ is by definition surjective, a  set of generators of $\mathfrak{h}_V$ is $\Big\{\Theta_V\big(s(e_a)\big)\Big\}$ where the $\{e_a\}$ form a basis of $T_{-1}=s^{-1}V$.

We extend the bracket (see Equations \eqref{brak1}-\eqref{brak4}) to $T_0$ by Equation  \eqref{bracket2} and by setting that the bracket between two generators $\Theta_V\big(s(e_a)\big)$ and $\Theta_V\big(s(e_b)\big)$ of $\mathfrak{h}_V$ satisfies:
\begin{equation}
\Big[\Theta_V(s(e_a)),\Theta_V(s(e_b))\Big]=f_{ab}{}^c\,\Theta_V\big(s(e_c)\big)
\end{equation}
One can now define a differential $\partial$ on $T$ by Equations \eqref{bracket} and \eqref{diff0}. That is to say, $\partial_0=-\Theta_V\circ s$, and:
\begin{equation}
\partial_{-1}(e_I)=-h_I^a\,e_a\hspace{0.1cm},\hspace{0.9cm}\partial_{-2}(e^t)=g^{It}\,e_I\hspace{0.5cm}\text{and}\hspace{0.5cm}\partial_{-3}(e^\alpha)=-k_t^\alpha\,e^t
\end{equation}
Let us take a concrete example to explain how things work precisely: since $\delta_3(X_t)=k^\alpha_t \, X_\alpha$, we deduce that we can set $\delta'_3=k^\alpha_tu_\alpha\frac{\partial}{\partial u_t}$. In that case, since $s^{-1}e^\alpha$ has degree $-5$, we obtain:
\begin{align}
\iota_{s^{-1}(\partial_{-3}(e^\alpha))}&=-\big[\delta'_3,\iota_{s^{-1}(e^\alpha)}\big]\\
&=-\iota_{s^{-1}(e^\alpha)}\circ\delta'_3\\
&=-k^\alpha_t\iota_{s^{-1}e^t}
\end{align}
Moreover, as required by Equation \eqref{differentielledebase}, we indeed have $\partial_{-1}=-s^{-1}\circ \mathrm{d}\circ s^2$ because $\delta_1(X^a)=s\circ \mathrm{d^*}(X^a)=h_I^a X^I$.


    These operator satisfy the Jacobi and Leibniz identities that we can compute, i.e. those that take values in $T$, since they are in one to one correspondence with Equations \eqref{lol1}-\eqref{lol8}, and Equations \eqref{lol9}-\eqref{lol12}, and no more.    
    For example, we have:
    \begin{align}
    	\partial_{-3}\big([e_{I},e_{J}]\big)	-\big[\partial_{-1}(e_I),e_J\big]-\big[e_I,\partial_{-1}(e_J)\big]&=	2\,c^{}_{\alpha IJ}k^\alpha_t e^t+2h^a_{[I}[e_a,e_{J]}]\\
	&=2\Big(\,c^{}_{\alpha IJ}k^\alpha_t -h^a_{[I}b^{}_{J]ta}\Big)\, e^t
    \end{align}
  which vanishes by Equation \eqref{lol11}. Another example using $\partial_0$ is:
  \begin{align}
  \partial_{-3}\big([e_a,e^t]\big)-\big[\partial_{0}(e_a),e^t\big]+\big[e_a,\partial_{-2}(e^t)\big]&=-c^t_{\alpha a}k_s^\alpha\,e^s+\big[\Theta_V(s(e_a)),e^t\big]+g^{It}\big[e_a,e_I\big]\\
  &=\Big(-c^t_{\alpha a}k_s^\alpha+X_{as}{}^t-g^{It}b^{}_{Isa}\Big)\,e^s
  \end{align}
    which vanishes by the definition of $X_{as}{}^t=-f_{as}{}^t+d_{as}^Kh^t_K$ and  Equation \eqref{lol12}.

Hence, this turns $\big( T,\partial,[\,.\,,.\,]\big)$ into a truncation at level 4 of a tensor hierarchy algebra.
    In other words, the data of the tensor hierarchy that we have defined so far from the 4-stem $(U,\delta,\pi,\mu)$ is completely contained in this truncation , and carries all the physical information that is needed. This justifies why the tensor hierarchy algebra is the correct object to look at when considering the $(1,0)$ superconformal model in six dimensions. We expect that there is a way of deducing the Lie $3$-algebra structure on $sT_{-1}\oplus sT_{-2}\oplus sT_{-3}$ given in \cite{monpapier} from the tensor hierarchy algebra structure on $T$. This Lie $3$-algebra was obtained by looking at the Bianchi identities satisfied by the field strengths, but in the present case one has to think the other way around : do not assume that the field strength are not given, and deduce them from the data of the tensor hierarchy algebra. This would show that much of the physical information captured in the tensor hierarchy is actually a mere mathematical artifact that can be deduced from straightforward computations.

\appendix

\section{Proof of Equations \eqref{donkey2} and \eqref{donkey1}}\label{appendicite}

The goal of this appendix is to give explicit proofs of Equations \eqref{donkey2} and \eqref{donkey1}. Let us start with the following Lemma:
\begin{lemme}
Let $i\in\mathbb{N}^*\cup\{\infty\}$ and let $\mathcal{V}=(\mathfrak{g},V,\Theta)$ be a Lie-Leibniz triple admitting a  $i$-stem $\mathcal{U}=(U,\delta,\pi,\mu)$. Then:
\begin{align}
\delta_1^*&=\mathrm{d}\circ s\label{oukey4}\\
\pi_0^*&=- s^{-1}\circ\Pi_W \label{oukey5}
\end{align}
\end{lemme}
\begin{proof}  We start by computing the dual map $\delta_1^*$. Let $u\in U_0=V^*$ and $\alpha\in U_1^*=s^{-1}W$, so that $\delta_1^*(\alpha)\in U_0^*$. Then, by Equation \eqref{dualitysmooth}, we have:
\begin{equation}\label{oukey1}
\big\langle\delta_1^*(\alpha),u\big\rangle_{U_0}=-\big\langle\alpha,\delta_1(u)\big\rangle_{U_1}
\end{equation}
By item 3. of Definition \ref{wooo}, we know that $\delta_1=s\circ \mathrm{d}^*$. Then, using Equation \eqref{suspensiondecal} on the right hand side of Equation \eqref{oukey1}, we have:
\begin{equation}\label{oukey2}
\big\langle\delta_1^*(\alpha),u\big\rangle_{V^*}=\big\langle s(\alpha),\mathrm{d}^*(u)\big\rangle_{W^*}
\end{equation}
Then, since the map $\mathrm{d}$ does not carry any degree, see e.g. Equation \eqref{pouuf}, we obtain by Equation \eqref{dualitysmooth} the following identity:
\begin{equation}\label{oukey3}
\big\langle\delta_1^*(\alpha),u\big\rangle_{V^*}=\big\langle \mathrm{d}\circ s(\alpha),u\big\rangle_{V^*}
\end{equation}
from which  we deduce Equation \eqref{oukey4}.

 Let us now compute the dual of the map $\pi_0$.  Given the definition of $\pi_0$ in item 3. of Definition \ref{wooo}, we can apply Equation \eqref{dualpi} to $u=s^{-1}(v)$, for some $v\in U_1$, to obtain:
\begin{equation}
\big\langle\pi_0(v),x\odot y\big\rangle_{S^2(V)}=-\big\langle s^{-1}(v),\Pi_W(x\odot y)\big\rangle_{W}
\end{equation}
Using Equation \eqref{suspensiondecal}, and recalling that $s^{-1}W=U_1^*$ and that $V=U_0^*$, we have:
\begin{equation}
\big\langle\pi_0(v),x\odot y\big\rangle_{S^2(U_0^*)}=-\big\langle v,s^{-1}\circ\Pi_W(x\odot y)\big\rangle_{U_1^*}
\end{equation}
so that we obtain that the dual of $\pi_0$ is the map $\pi_0^*:S^2(U_0^*)\to U_1^*$ satisfying Equation \eqref{oukey5}.
\end{proof}

Let us now turn to the core statement of this appendix:

\begin{proposition}\label{technicalemma}
Let $i\in\mathbb{N}^*\cup\{\infty\}$ and let $\mathcal{V}=(\mathfrak{g},V,\Theta)$ be a Lie-Leibniz triple admitting a  $i$-stem $\mathcal{U}=(U,\delta,\pi,\mu)$.
 Assume that $U_{i+1}$ has been defined through a map $\pi_i:U_{i+1}\to S^2(U)_i$ satisfying item 4. of Definition \ref{wooo}, and assume that there exists a map $\mu_i:U_i\to U_0\otimes U_i$ defined as in item 5. of Definition \ref{wooo}. Then:
\begin{align}
\delta\circ\mu_{i-1}&=\mu_i\circ\delta_{i}\label{diddy1}\\
\pi\circ\mu_{i}&=\mu\circ\pi_{i-1}\label{diddy2}
\end{align}
\end{proposition}

\begin{proof}
Let us first show Equation \eqref{diddy1} for $i\geq2$.
We know that $\mu_{i-1}$ takes values in $U_0\odot U_{i-1}$, thus $\delta\circ \mu_i$ takes values in $U_1\odot U_{i-1}\oplus U_0\odot U_i$. Let us show that it actually takes values only in $U_0\odot U_i$. Let $\alpha\in U_1^*$ , $\beta\in U_{i-1}^*$ and $u\in U_{i-1}$. First assume that $i\geq3$ so that $i-1\geq2$. Then the image of $\delta\circ\mu_{i-1}$ in $U_1\odot U_{i-1}$ satisfies:
\begin{align}
\Big\langle\alpha\odot \beta\,,\delta\circ\mu_{i-1}(u)\Big\rangle_{U_1\odot U_{i-1}}&=(-1)^{i}\Big\langle\delta_1^*(\alpha)\odot\beta\,,\mu_{i-1}(u)\Big\rangle_{U_0\odot U_{i-1}}\\
&\hspace{3cm}+(-1)^{i+1}\Big\langle\alpha\odot\delta_{i-1}^*(\beta),\mu_{i-1}(u)\Big\rangle_{U_0\odot U_{i-1}}\nonumber\\
&=(-1)^{i}\Big\langle\mathrm{d}\big(s(\alpha)\big)\odot\beta\,,\mu_{i-1}(u)\Big\rangle_{U_0\odot U_{i-1}}\\
&=\frac{(-1)^i}{2}\Big\langle\beta\,,\rho_{i-1,\Theta(\mathrm{d}(s(\alpha)))}(u)\Big\rangle_{U_{i-1}}\\
&=0
\end{align}
In the first line, we passed from the left hand side of the equal sign  to the right hand side by taking the dual of the map $\delta$, and by using Equation \eqref{dualitysmooth} and the fact that $\mu_{i-1}$ takes values in $U_0\odot U_{i-1}$. We passed from the first line to the second line by noticing that $\alpha\odot\delta_{i-1}^*(\beta)$ is not taking values in $U_0^*\odot U_{i-1}^*$, and by applying Equation \eqref{oukey4}. Then we passed from the second to the third line by applying the very definition of $\mu_{i-1}$ as given in item 5. of Definition \ref{wooo}. The result is zero because $\Theta\circ\mathrm{d}=0$, as proven in Proposition \ref{thetainclusion}.
In the case where $i=2$, we could not get rid of the term $\big\langle\alpha\odot\delta_{i-1}^*(\beta),\mu_{i-1}(u)\big\rangle_{U_1\odot U_0}$ since in that case $\delta_{i-1}^*(\beta)=\mathrm{d}\circ s(\beta)\in U_0^*$. But, still, the action of $\Theta$ on $\mathrm{d}$ in the last line would make this contribution vanish. Hence, we conclude that $\delta\circ\mu_{i-1}$ takes values in $U_0\odot U_{i}$.

Let us now compute this contribution. Assume that $i\geq2$, and let $x\in U_0^*=V$, $\alpha\in U_{i}^*$ and $u\in U_{i-1}$. Then we have:
\begin{align}
\Big\langle x\odot \alpha,\delta\circ\mu_{i-1}(u)\Big\rangle_{U_0\odot U_i}&=(-1)^i\Big\langle x\odot \delta_i^*(\alpha),\mu_{i-1}(u)\Big\rangle_{U_0\odot U_{i-1}}\\
\footnotesize\text{by definition of $\mu_{i-1}$}\hspace{0.5cm}\normalsize&=\frac{(-1)^i}{2}\,\Big\langle \delta_i^*(\alpha),\rho_{\Theta(x)}(u)\Big\rangle_{U_{i-1}}\\
\footnotesize\text{by Equation \eqref{dualitysmooth}}\hspace{0.72cm}\normalsize&=\frac{1}{2}\,\Big\langle \alpha,\delta_i\big(\rho_{\Theta(x)}(u)\big)\Big\rangle_{U_{i}}\\
\footnotesize\text{by $\mathfrak{h}$-equivariance of $\delta_i$}\hspace{0.45cm}\normalsize&=\frac{1}{2}\,\Big\langle \alpha,\rho_{\Theta(x)}\big(\delta_i(u)\big)\Big\rangle_{U_{i}}\\
\footnotesize\text{by definition of $\mu_i$}\hspace{0.83cm}\normalsize&=\Big\langle x\odot\alpha,\mu_i\circ\delta_i(u)\Big\rangle_{U_0\odot U_{i}}
\end{align}
Hence we have:
\begin{equation}
\Big\langle x\otimes  \alpha,\delta\circ\mu_{i-1}(u)-\mu_i\circ\delta_i(u)\Big\rangle_{U_0\odot U_i}=0
\end{equation}
for every $x\in U_0^*,\alpha\in U_i^*$ and $u\in U_{i-1}$. We conclude that Equation \eqref{diddy1} is true whenever~$i\geq2$.

In the case where $i=1$, we have to show that:
\begin{equation}\label{saddle2}
\delta\circ\mu_0=\mu_1\circ\delta_1
\end{equation}
Both sides take values in $U_0\odot U_1$. 
Let $x\in U_0^*=V$, $\alpha\in U_1^*=s^{-1}W$, then for any $ u\in U_0= V^*$, we have:
\begin{align}
\Big\langle x\odot \alpha\,,\delta\circ\mu_{0}(u)\Big\rangle_{U_0\odot U_1}&=-\Big\langle x\odot \delta_1^*(\alpha),\mu_{0}(u)\Big\rangle_{U_0\odot U_{0}}\\
\footnotesize\text{by Equation \eqref{oukey4}}\hspace{0.803cm}\normalsize&=-\Big\langle x\odot \mathrm{d}\big(s(\alpha)\big),\mu_0(u)\Big\rangle_{U_{0}\odot U_0}\\
\footnotesize\text{by definition of $\mu_0$}\hspace{0.801cm}\normalsize&=\Big\langle \big\{x,\mathrm{d}\big(s(\alpha)\big)\big\},u\Big\rangle_{U_0}\\
\footnotesize\text{by definition of $\{\,.\,,.\,\}$}\hspace{0.5cm}\normalsize&=\frac{1}{2}\,\Big\langle x\bullet\mathrm{d}\big(s(\alpha)\big),u\Big\rangle_{U_0}\\
\footnotesize\text{by $\mathfrak{h}$-equivariance of $\mathrm{d}$}\hspace{0.5cm}\normalsize&=\frac{1}{2}\,\Big\langle \mathrm{d}\circ\eta_{W,\Theta(x)}\big(s(\alpha)\big),u\Big\rangle_{U_0}\\
\footnotesize\text{by Equation \eqref{representation2}}\hspace{0.782cm}\normalsize&=\frac{1}{2}\,\Big\langle \mathrm{d}\circ s\big(\rho^\vee_{1,\Theta(x)}(\alpha)\big),u\Big\rangle_{U_0}\\
\footnotesize\text{by Equations \eqref{oukey4} and \eqref{dualitysmooth}}\hspace{0.2cm}\normalsize&=-\frac{1}{2}\,\Big\langle \rho^\vee_{1,\Theta(x)}(\alpha),\delta_1(u)\Big\rangle_{U_1}\\
\footnotesize\text{by Equation \eqref{equationmu}}\hspace{0.782cm}\normalsize&=\,\Big\langle x\odot \alpha\,,\mu_1\circ\delta_1(u)\Big\rangle_{U_0\odot U_1}
\end{align}
Hence we have:
\begin{equation}
\Big\langle x\otimes  \alpha\,,\delta\circ\mu_0(u)-\mu_1\circ\delta_1(u)\Big\rangle_{U_0\odot U_1}=0
\end{equation}
for every $x\in U_0^*, \alpha\in U_1^*$ and $u\in U_0$, which implies Equation \eqref{saddle2}. Thus, Equation \eqref{diddy1} is true for every $i\geq1$, as can be illustrated in the following commutative diagram:
\begin{center}
\begin{tikzpicture} 
\matrix(a)[matrix of math nodes, 
row sep=4em, column sep=3em, 
text height=1.5ex, text depth=0.25ex] 
{U_{i-1}&U_{i}\\ 
S^{2}(U)_{i-1}&S^{2}(U)_{i}\\}; 
\path[->](a-1-1) edge node[right]{$\mu_{i-1}$}  (a-2-1); 
\path[->](a-1-2) edge node[right]{$\mu_i$}  (a-2-2); 
\path[->](a-2-1) edge node[above]{$\delta$} (a-2-2);
\path[->](a-1-1) edge node[above]{$\delta_{i}$} (a-1-2);
\end{tikzpicture}
\end{center}

Now let us prove Equation \eqref{diddy2}. First, the derivation properties of $\mu$ and $\pi$ imply the following two identities:
\begin{align}
\Big\langle\alpha\odot\beta\odot\gamma\,,\pi(u\odot v)\Big\rangle_{S^2(U)_{|u|+|v|-1}}&=(-1)^{|\alpha|+|\beta|+|\gamma|}\Big\langle\pi^*(\alpha,\beta)\odot\gamma\,,u\odot v\Big\rangle_{U_{|u|}\odot U_{|v|}}+\circlearrowleft\label{SJW1}\\
2\,\Big\langle\alpha\odot\beta\odot\gamma\,,\mu(u\odot v)\Big\rangle_{U_{|u|}\odot U_{|v|}}&=-\Big\langle\rho^\vee_{\Theta(\alpha)}(\beta)\odot\gamma+\beta\odot\rho^\vee_{\Theta(\alpha)}(\gamma),u\odot v\Big\rangle_{U_{|u|}\odot U_{|v|}}+\circlearrowleft\label{SJW2}
\end{align}
for any homogeneous elements $\alpha,\beta,\gamma\in U^*$ and $u,v\in U$, and where $\circlearrowleft$ indicates that we perform a (graded) circular permutation of $\alpha,\beta,\gamma$. Here as well, $\Theta$ is considered as the zero function on $U_k$, as soon as $k\geq1$.

Now let assume that $i\geq2$. The map $\pi_{i-1}$ takes values in: \begin{equation*}
S^2(U)_{i-1}\,= \, \underset{k+l\,=\,i-1}{\bigoplus_{k,l\,\geq\, 0}}U_k\odot U_l
\end{equation*}
For any $k,l\geq0$ such that $k+l=i-1$, we define $\pi^{(k,l)}_{i-1}$ to be the co-restriction of $\pi_{i-1}$ to the subspace $U_k\odot U_{l}$. We extend it to all $S(U)$ by derivation. In particular, it is symmetric in the $k,l$ indices. Since the map $\mu_n$ takes values in $U_0\odot U_n$ for every $1\leq n\leq i$, both the map $\mu\circ \pi^{(k,l)}_{i-1}$ and the map $\pi^{(k,l)}_{i-1}\circ\mu_i$ take values in $U_0\odot U_k\odot U_l$. Now, assume that $k,l\geq1$ and let $x\in U_0^*=V$, $\alpha\in U_k^*$, $\beta\in U_l^*$ and $u\in U_{i}$. Thus, using Equation \eqref{SJW2}, we have:
\begin{align}
2\,\Big\langle x\odot \alpha\odot \beta\,, \mu\circ\pi^{(k,l)}_{i-1}(u)\Big\rangle_{U_0\odot U_k\odot U_l}&=-\Big\langle\rho^\vee_{k,\Theta(x)}(\alpha)\odot\beta+\rho^\vee_{l,\Theta(x)}(\beta)\odot\alpha\,,\pi_{i-1}^{(k,l)}(u)\Big\rangle_{U_k\odot U_l}\\
\footnotesize\text{by the derivation property of $\rho^\vee$}\hspace{0.2cm}\normalsize&=-\Big\langle\rho^\vee_{\Theta(x)}(\alpha\odot\beta)\,,\pi_{i-1}^{(k,l)}(u)\Big\rangle_{U_k\odot U_l}\\
\footnotesize\text{by definition of $\rho^\vee$}\hspace{0.803cm}\normalsize&=\Big\langle\alpha\odot\beta\,,\rho_{\Theta(x)}\big((\pi_{i-1}^{(k,l)}(u)\big)\Big\rangle_{U_k\odot U_l}\\
\footnotesize\text{by $\mathfrak{g}$-equivariance of $\pi$}\hspace{0.6cm}\normalsize&=\Big\langle\alpha\odot\beta\,,\pi_{i-1}^{(k,l)}\big(\rho_{i,\Theta(x)}(u)\big)\Big\rangle_{U_k\odot U_l}\\
\footnotesize\text{by Equation \eqref{dualitysmooth}}\hspace{0.75cm}\normalsize&=(-1)^{k+l}\Big\langle\big(\pi_{i-1}^{(k,l)}\big)^*(\alpha,\beta)\,,\rho_{i,\Theta(x)}(u)\Big\rangle_{U_i}\\
\footnotesize\text{by definition of $\mu_i$}\hspace{0.815cm}\normalsize&=(-1)^{k+l}2\,\Big\langle x\odot \big(\pi_{i-1}^{(k,l)}\big)^*(\alpha,\beta)\,,\mu_i(u)\Big\rangle_{U_0\odot U_i}\\
\footnotesize\text{because $x\in U_0^*$}\hspace{1cm}\normalsize&=2\,\Big\langle x\odot \alpha\odot\beta\,,\pi_{i-1}^{(k,l)}\circ\mu_i(u)\Big\rangle_{U_0\odot U_k\odot U_l}
\end{align}
Thus we have proven that
\begin{equation}\label{saddle}
\Big\langle x\odot \alpha\odot \beta\,, \mu\circ\pi^{(k,l)}_{i-1}(u)-\pi_{i-1}^{(k,l)}\circ\mu_i(u)\Big\rangle_{U_0\odot U_k\odot U_l}=0
\end{equation} for every $k,l\geq1$ such that $k+l=1$, and for every $x\in U_0^*=V$, $\alpha\in U_k^*$, $\beta\in U_l^*$ and $u\in U_{i}$.

Now assume that either $k=0$ or $l=0$. They cannot be equal to zero at the same time, for we have chosen $i\geq2$. Assume for example that $k=0$, then necessarily $l=i-1$. Let $x,y\in U_0^*=V$, $\alpha\in U_{i-1}^*$ and $u\in U_i$, and for some clarity set $\widehat{\pi}_{i-1}\equiv \pi^{(0,i-1)}_{i-1}$. Then, using Equation \eqref{SJW2}, we have:
\begin{align}
2\Big\langle x\odot y\odot\alpha\,,\mu\circ\widehat{\pi}_{i-1}(u)\Big\rangle_{U_0\odot U_0\odot U_{i-1}}
&=-\Big\langle y\odot\rho^\vee_{i-1,\Theta(x)}(\alpha)\,,\widehat{\pi}_{i-1}(u)\Big\rangle_{U_0\odot U_{i-1}}-\, x\leftrightarrow y\nonumber\\
&\hspace{1.5cm}-2\,\Big\langle \{x,y\}\odot\alpha\,,\widehat{\pi}_{i-1}(u)\Big\rangle_{U_0\odot U_{i-1}}\\
\footnotesize\text{by the derivation property of $\rho^\vee$}\hspace{0.3cm}\normalsize&=-\Big\langle \rho^\vee_{\Theta(x)}(y\odot\alpha)\,,\widehat{\pi}_{i-1}(u)\Big\rangle_{U_0\odot U_{i-1}}-\, x\leftrightarrow y\\
\footnotesize\text{by definition of $\rho^\vee$}\hspace{0.565cm}\normalsize&=\Big\langle y\odot\alpha\,,\rho_{\Theta(x)}\big(\widehat{\pi}_{i-1}(u)\big)\Big\rangle_{U_0\odot U_{i-1}}+\, x\leftrightarrow y\\
\footnotesize\text{by $\mathfrak{g}$-equivariance of $\pi$}\hspace{0.4cm}\normalsize&=\Big\langle y\odot\alpha\,,\widehat{\pi}_{i-1}\big(\rho_{i,\Theta(x)}(u)\big)\Big\rangle_{U_0\odot U_{i-1}}+\, x\leftrightarrow y\\
\footnotesize\text{by Equation \eqref{dualitysmooth}}\hspace{0.52cm}\normalsize&=(-1)^{(i-1)}\Big\langle \big(\widehat{\pi}_{i-1}\big)^*(y,\alpha)\,,\rho_{i,\Theta(x)}(u)\Big\rangle_{U_0\odot U_{i-1}}\nonumber\\
&\hspace{3cm}+(-1)^{(i-1)}\, x\leftrightarrow y\\
\footnotesize\text{by definition of $\mu_i$}\hspace{0.577cm}\normalsize&=(-1)^{(i-1)}2\,\Big\langle x\odot\big(\widehat{\pi}_{i-1}\big)^*(y,\alpha)\,,\mu_i(u)\Big\rangle_{U_0\odot U_{i-1}}\nonumber\\
&\hspace{3cm}+(-1)^{(i-1)}\, x\leftrightarrow y\\
\footnotesize\text{by Equation \eqref{dualitysmooth}}\hspace{0.52cm}\normalsize&=2\,\Big\langle x\odot y\odot\alpha\,,\widehat{\pi}_{i-1}\circ\mu_i(u)\Big\rangle_{U_0\odot U_0 \odot U_{i-1}}
\end{align}
Thus we have proven that: 
\begin{equation}
\Big\langle x\odot y\odot\alpha\,,\mu\circ\widehat{\pi}_{i-1}(u)-\widehat{\pi}_{i-1}\circ\mu_i(u)\Big\rangle_{U_0\odot U_0 \odot U_{i-1}}=0
\end{equation}
$x,y\in U_0^*=V$, $\alpha\in U_{i-1}^*$ and $u\in U_i$ (recall that $i\geq2$). Thus, merging this result with the one of Equation \eqref{saddle}, we finally obtain Equation \eqref{diddy2} for $i\geq2$.

When $i=1$, we have to show the following identity:
\begin{equation}\label{saddle3}
\pi\circ\mu_1=\mu\circ\pi_0
\end{equation}
Both sides of the equality take values in $S^3(U_0)$.
However, since $\mu_1$ takes values in $U_0\otimes U_1$, and that $\pi|_{U_0}=0$, the map $\pi$ only on the $U_1$ components of $\mathrm{Im}(\mu_1)$.
Let $x,y,z\in U_0^*=V$ and $u\in U_1=s(W^*)$, then we have:
\begin{align}
2\,\Big\langle x\odot y\odot z\,,\mu\circ\pi_0(u)\Big\rangle_{S^3(U_0)}&=-2\,\Big\langle\{x,y\}\odot z\,,\pi_0(u)\Big\rangle_{S^2(U_0)}\,+\circlearrowleft\\
\footnotesize\text{by Equation \eqref{oukey5}}\hspace{0.728cm}\normalsize&=2\,\Big\langle\Pi_W\big(\{x,y\},z\big),s^{-1}(u)\Big\rangle_{W^*}\,+\circlearrowleft\\
\footnotesize\text{by definition of $\{\,.\,,.\,\}$}\hspace{0.5cm}&=\Big\langle\Pi_W\big(x\bullet y,z\big)+\Pi_W\big(y\bullet x,z\big),s^{-1}(u)\Big\rangle_{W^*}\,+\circlearrowleft\\
\footnotesize\text{by using the permutation}\hspace{0.224cm}&=\Big\langle\Pi_W\big(x\bullet y,z\big)+\Pi_W\big(y, x\bullet z\big),s^{-1}(u)\Big\rangle_{W^*}\,+\circlearrowleft\\
\footnotesize\text{by $\mathfrak{g}$-equivariance of $\Pi_W$}\hspace{0.35cm}\normalsize &=\Big\langle \eta_{W,\Theta(x)}\big( \Pi_W(y,z)\big),s^{-1}(u)\Big\rangle_{W^*}\, +\circlearrowleft\\
\footnotesize\text{by definition of $\eta_W^\vee$}\hspace{0.609cm}\normalsize&=-\Big\langle \Pi_W(y,z),\eta_{W,\Theta(x)}^\vee\big(s^{-1}(u)\big)\Big\rangle_{W^*}\, +\circlearrowleft\\
\footnotesize\text{by Equation \eqref{representationshifted}}\hspace{0.647cm}\normalsize&=-\Big\langle \Pi_W(y,z),s^{-1}\big(\rho_{1,\Theta(x)}(u)\big)\Big\rangle_{W^*}\, +\circlearrowleft\\
\footnotesize\text{by Equation \eqref{suspensiondecal}}\hspace{0.647cm}\normalsize&=-\Big\langle s^{-1}\circ\Pi_W(y,z),\rho_{1,\Theta(x)}(u)\Big\rangle_{U_1}\, +\circlearrowleft\\
\footnotesize\text{by Equation \eqref{oukey5}}\hspace{0.728cm}\normalsize&=\Big\langle \pi_0^*(y,z),\rho_{1,\Theta(x)}(u)\Big\rangle_{U_1}\, +\circlearrowleft\\
\footnotesize\text{by definition of $\mu_1$}\hspace{0.726cm}\normalsize&=2\,\Big\langle x\odot\pi_0^*(y,z),\mu_1(u)\Big\rangle_{U_0\odot U_1}\, +\circlearrowleft\\
\footnotesize\text{by Equation \eqref{SJW1}}\hspace{0.565cm}\normalsize&=2\,\Big\langle x\odot y\odot z\,,\pi\circ\mu_1(u)\Big\rangle_{S^3(U_0)}
\end{align}
where $\circlearrowleft$ symbolizes a circular graded permutation of the elements $x,y,z$.
Hence we have:
\begin{equation}
\Big\langle x\odot y\odot z\,,\mu\circ\pi_0(u)-\pi\circ\mu_1(u)\Big\rangle_{S^3(U_0)}=0
\end{equation}
for every $u\in U_1$ and $ x,y,z\in U_0^*$, which implies Equation \eqref{saddle3}. Thus, Equation \eqref{diddy2} is true for every $i\geq1$, as can be illustrated in the following commutative diagram:
\begin{center}
\begin{tikzpicture} 
\matrix(a)[matrix of math nodes, 
row sep=4em, column sep=3em, 
text height=1.5ex, text depth=0.25ex] 
{&U_{i}\\
S^{2}(U)_{i-1}&S^2(U)_{i}\\ 
S^{3}(U)_{i-1}&\\}; 
\path[left hook->](a-1-2) edge node[above left]{$\pi_{i-1}$}  (a-2-1); 
\path[->](a-2-2) edge node[above left]{$\pi$}  (a-3-1); 
\path[->](a-1-2) edge node[right]{$\mu_i$} (a-2-2);
\path[->](a-2-1) edge node[left]{$\mu$} (a-3-1);
\end{tikzpicture}
\end{center}
\end{proof}

\section{Compatibility between $[\,.\,,.\,]$ and $\partial$}\label{appendicite2}

In this appendix, we show that the bracket and the differential defined in the proof of Theorem \ref{prop:tensorhierarchy} are compatible, in the sense that they define a differential graded Lie algebra structure on $T=\mathfrak{h}\oplus T'$, where $T'=s^{-1}(U^*)$.
We split this proof in four steps: first, we prove that $\partial$ is compatible with the restriction of the bracket to $T_{-k}\wedge T_{-l}$, for $k,l\geq2$. Second, we prove the compatibility when either $k$ or $l$ is equal to $1$, and then when $k=l=1$. Eventually, we prove that the differential and the bracket are compatible on $\mathfrak{h}\wedge T$ (the case $\mathfrak{h}\wedge\mathfrak{h}$ being trivial since $\partial|_{T_0}=0$).

First, let us set $\widetilde{U}\equiv\bigoplus_{1\leq k<\infty}U_{k}$, and let its shifted dual be $\widetilde{T}\equiv s^{-1}(\widetilde{U}^*)=\bigoplus_{2\leq k<\infty}T_{-k}$. The co-restriction of the map $\mu$ to $S^2(\widetilde{U})$ is identically zero, because $\mu$ takes values in $U_0\odot U_{k}$. Since $\mu$ is null-homotopic (see item 6. of Definition \ref{wooo}), we deduce that the co-restriction of the map $\pi\circ\delta+\delta\circ\pi: U_k\to S^2(U)$ to $S^2(\widetilde{U})$ is zero:
\begin{equation}\label{nullhomotopy1}
\big(\pi\circ\delta+\delta\circ\pi\big)\big|^{S^2(\widetilde{U})}=0
\end{equation}
From this identity, together with Equations \eqref{bonjouuurq} and \eqref{bonjouuurd}, we deduce that:
\begin{equation}\label{nullhomotopy2}
\big(Q_\pi\circ\delta'+\delta'\circ Q_\pi\big)\big|^{S^2((s^{-1}\widetilde{T})^*)}=0
\end{equation}
Let us check how this identity translates on $\widetilde{T}$. Let $k,l\geq2$ and let $x\in T_{-k}, y\in T_{-l}$, then we will show that Equation \eqref{nullhomotopy2} implies the vanishing of the following quantity:
\begin{equation}\label{chourave}
\Delta_{k,l}(x,y)\equiv\partial_{-k-l+1}\big([x,y]\big)-\big[\partial_{-k+1}(x),y\big]-(-1)^k\big[x,\partial_{-l+1}(y)\big]
\end{equation} 
We will compute each term on the right hand side, one after the other,  using alternatively Equations \eqref{voronov2} and \eqref{bracket}:
\begin{align}
\iota_{s^{-1}[\partial_{-k+1}(x),y]}&=(-1)^{k-1}\big[[Q_\pi,\iota_{s^{-1}(\partial_{-k+1}(x))}],\iota_{s^{-1}(y)}\big]\\
&=(-1)^{(k+1)(l+1)}\iota_{s^{-1}(y)}\circ\big[\delta'_{k-1},\iota_{s^{-1}(x)}\big]\circ Q_\pi\\
&=(-1)^{(k+1)(l+1)}\iota_{s^{-1}(y)}\circ\Big(\delta'\circ\iota_{s^{-1}(x)}-(-1)^{k-1}\iota_{s^{-1}(x)}\circ\delta'\Big)\circ Q_\pi\label{liberal1}
\end{align}
We also have:
\begin{align}
\iota_{s^{-1}[x,\partial_{-l+1}(y)]}&=(-1)^{k}\big[[Q_\pi,\iota_{s^{-1}(x)}],\iota_{s^{-1}(\partial_{-l+1}(y))}\big]\label{liberal2bis}\\
&=-(-1)^{kl}\iota_{s^{-1}(\partial_{-l+1}(y))}\iota_{s^{-1}(x)}\circ Q_\pi\\
&=(-1)^{kl}\big[\delta'_{l-1},\iota_{s^{-1}(y)}\big]\circ\iota_{s^{-1}(x)}\circ Q_\pi\\
&=(-1)^{(k+1)l}\iota_{s^{-1}(y)}\circ\delta'\circ\iota_{s^{-1}(x)}\circ Q_\pi\label{liberal2}
\end{align}
And finally:
\begin{align}
\iota_{s^{-1}(\partial_{-k-l+1}[x,y])}&=-[\delta'_{k+l-1},\iota_{s^{-1}[x,y]}]\label{liberal3bis}\\
&=(-1)^{k+1}\Big[\delta'_{k+l-1},\big[[Q_\pi,\iota_{s^{-1}(x)}],\iota_{s^{-1}(y)}\big]\Big]\\
&=(-1)^{k(l+1)}\big[\delta',\iota_{s^{-1}(y)}\iota_{s^{-1}(x)}\circ Q_\pi\big]\\
&=(-1)^{(k+1)l}\iota_{s^{-1}(y)}\iota_{s^{-1}(x)}\circ Q_\pi\circ\delta'\label{liberal3}
\end{align}
Thus, by substracting Line \eqref{liberal1} and $(-1)^k$ times Line \eqref{liberal2} to Line \eqref{liberal3}, we obtain:
\begin{equation}\label{peterpeter}
\iota_{s^{-1}(\Delta_{k,l}(x,y))}=(-1)^{(k+1)l}\iota_{s^{-1}(y)}\iota_{s^{-1}(x)}\circ\big( Q_\pi\circ\delta'+\delta'\circ Q_\pi\big)
\end{equation}
Since $s^{-1}(x)$ and $s^{-1}(y)$ belongs to $s^{-1}\widetilde{T}$, the bracket on the right hand side is implicitely co-restricted to $S^2\big((s^{-1}\widetilde{T})^*\big)$. This bracket vanishes by Equation \eqref{nullhomotopy2}, so that we deduce that the left hand side of Equation \eqref{peterpeter} vanishes for every $x,y\in \widetilde{T}$. Thus, $\Delta_{k,l}(x,y)=0$ for every $k,l\geq2$ and for every $x\in T_{-k}, y\in T_{-l}$, which implies, by Equation \eqref{chourave}, that the differential $\partial$ is compatible with the Lie bracket $[\,.\,,.\,]$ on $\widetilde{T}\wedge\widetilde{T}$:
\begin{equation}\label{lignebleue1}
\partial_{-k-l+1}\big([x,y]\big)=\big[\partial_{-k+1}(x),y\big]+(-1)^k\big[x,\partial_{-l+1}(y)\big]
\end{equation}

Let us now turn ourselves to the case where either $x\in T_{-1}$ or $y\in T_{-1}$ but not both at the same time. From item 5. of Definition \ref{wooo}, we know that for every $1\leq k<i$ the map $\mu_k$ takes values in $U_0\odot U_k$. From item 6. we deduce that the co-restriction of $\delta\circ\pi+\pi\circ\delta:U_k\to S^2(U)$ to $U_0\odot U_k$ is equal to $\mu_k$:
\begin{equation}\label{nullhomotopy3}
\big(\pi\circ\delta+\delta\circ\pi\big)\big|^{U_0\odot U_k}=\mu_k
\end{equation}
Let us check how this identity translates on $T$. Let $k\geq2$ and let $x\in T_{-1}=s^{-1}V$ and $y\in T_{-k}=s^{-1}(U_{k-1}^*)$, then we will show that Equation \eqref{nullhomotopy3} implies the vanishing of the following quantity:
\begin{equation}\label{chourave2}
\Xi_{k}(x,y)\equiv\partial_{-k}\big([x,y]\big)+\eta_{-k,\Theta(s(x))}(y)+\big[x,\partial_{-k+1}(y)\big]
\end{equation} 
We will compute each term on the right hand side, one after another, and sum them up afterwards. Let $\alpha\in T_{-k}^*=sU_{k-1}$, then by Equation \eqref{representationshifted}, we have:
\begin{align}
\Big\langle \alpha\,,\eta_{-k,\Theta(s(x))}(y)\Big\rangle_{T_{-k}}&=\Big\langle \alpha\,,s^{-1}\circ\rho^*_{k-1,\Theta(s(x))}\circ s(y)\Big\rangle_{s^{-1}(U_{k-1}^\vee)}\label{liberal6bis}\\
\footnotesize\text{by Equation \eqref{suspensiondecal}}\hspace{0.647cm}\normalsize&=\Big\langle s^{-1}(\alpha),\rho^\vee_{k-1,\Theta(s(x))}\circ s(y)\Big\rangle_{U_{k-1}^*}\\
\footnotesize\text{by Equation \eqref{equationmu}}\hspace{0.647cm}\normalsize&=-2\,\Big\langle  \mu_{k-1}\big(s^{-1}(\alpha)\big),s(x)\odot s(y)\Big\rangle_{U_0^*\odot U_{k-1}^*}\label{liberal6}
\end{align}
We also have: 
\begin{align}
\Big\langle\alpha\,,\big[x,\partial_{-k+1}(y)\big]\Big\rangle_{T_{-k}}&=\Big\langle s(\alpha),s^{-1}\Big(\big[x,\partial_{-k+1}(y)\big]\Big)\Big\rangle_{s^{-1}T_{-k}}\label{liberal7bis}\\
\footnotesize\text{by Equation \eqref{party}}\hspace{0.647cm}\normalsize&=\iota_{s^{-1}([x,\partial_{-k+1}(y)])}\big(s(\alpha)\big)\\
\footnotesize\text{by Eq. \eqref{liberal2bis}--\eqref{liberal2}}\hspace{0.57cm}\normalsize&=\iota_{s^{-1}(y)}\circ\delta'\circ\iota_{s^{-1}(x)}\circ Q_\pi \big(s(\alpha)\big)\\
\footnotesize\text{because $x\in T_{-1}$}\hspace{0.85cm}\normalsize&=\iota_{s^{-1}(y)}\iota_{s^{-1}(x)}\circ\delta'\circ Q_\pi \big(s(\alpha)\big)\\
\footnotesize\text{by Equation \eqref{party2}}\hspace{0.647cm}\normalsize&=2\,\Big\langle\delta'\circ Q_\pi \big(s(\alpha)\big),s^{-1}(x)\odot s^{-1}(y)\Big\rangle_{s^{-1}T_{-1}\odot s^{-1}T_{-k}}\\
\footnotesize\text{by Eq. \eqref{bonjouuurq} and \eqref{bonjouuurd}}\hspace{0.2cm}\normalsize&=2\,\Big\langle \big(s^2\odot s^2\big)\circ \delta\circ\pi\big(s^{-1}(\alpha)\big),s^{-1}(x)\odot s^{-1}(y)\Big\rangle_{s^{-2}(U_{0}^*)\odot s^{-2}(U_{k-1}^*)}\\
\footnotesize\text{by Equation \eqref{suspensiondecal2}}\hspace{0.647cm}\normalsize&=2\,\Big\langle \delta\circ\pi\big(s^{-1}(\alpha)\big),s(x)\odot s(y)\Big\rangle_{U_{0}^*\odot U_{k-1}^*}\label{liberal7} 
\end{align}
And finally:
\begin{align}
\Big\langle\alpha\,,\partial_{-k}\big([x,y]\big)\Big\rangle_{T_{-k}}&=\Big\langle s(\alpha),s^{-1}\circ\partial_{-k}\big([x,y]\big)\Big\rangle_{s^{-1}T_{-k}}\label{liberal8bis}\\
\footnotesize\text{by Equation \eqref{party}}\hspace{0.647cm}\normalsize&=\iota_{s^{-1}(\partial_{-k}[x,y])}\big(s(\alpha)\big)\\
\footnotesize\text{by Eq. \eqref{liberal3bis}--\eqref{liberal3}}\hspace{0.47cm}\normalsize&=\iota_{s^{-1}(y)}\iota_{s^{-1}(x)}\circ Q_\pi\circ\delta'\big(s(\alpha)\big)\\
\footnotesize\text{by Equation \eqref{party2}}\hspace{0.647cm}\normalsize&=2\,\Big\langle Q_\pi\circ\delta'\big(s(\alpha)\big),s^{-1}(x)\odot s^{-1}(y)\Big\rangle_{s^{-1}T_{-1}\odot s^{-1}T_{-k}}\\
\footnotesize\text{by Eq. \eqref{bonjouuurq} and \eqref{bonjouuurd}}\hspace{0.2cm}\normalsize&=2\,\Big\langle\big(s^2\odot s^2\big)\circ \pi\circ\delta\big(s^{-1}(\alpha)\big),s^{-1}(x)\odot s^{-1}(y)\Big\rangle_{s^{-2}(U_{0}^*)\odot s^{-2}(U_{k-1}^*)}\\
\footnotesize\text{by Equation \eqref{suspensiondecal2}}\hspace{0.647cm}\normalsize&=2\,\Big\langle \pi\circ\delta\big(s^{-1}(\alpha)\big),s(x)\odot s(y)\Big\rangle_{U_{0}^*\odot U_{k-1}^*}\label{liberal8} 
\end{align}

Summing Lines \eqref{liberal6}, \eqref{liberal7} and \eqref{liberal8}, and dividing by 2, one obtains the following quantity:
\begin{equation}
\Big\langle  \big(\delta\circ\pi+\pi\circ\delta-\mu_{k-1}\big)\big(s^{-1}(\alpha)\big),s(x)\odot s(y)\Big\rangle_{U_{0}^*\odot U_{k-1}^*}
\end{equation}
that vanishes by using Equation \eqref{nullhomotopy3} at level $k-1$.
From this, by adding the left hand sides of Lines \eqref{liberal6bis}, \eqref{liberal7bis} and \eqref{liberal8bis}, we deduce the following identity: 
\begin{equation}
\big\langle \alpha,\Xi_k(x,y)\big\rangle=0
\end{equation}
Since it holds for every $2\leq k<i$ and for every $\alpha\in T_{-k}^*$, $x\in T_{-1}$ and $y\in T_{-k}$, it implies that the quantity $\Xi_k(x,y)$ is identically zero. Then, using Equations \eqref{representationshifted}, \eqref{bracket2} and \eqref{diff0} in Equation \eqref{chourave2}, it implies that the following identity holds:
\begin{equation}\label{lignebleue2}
\partial_{-k}\big([x,y]\big)=\big[\partial_0(x),y\big]-\big[x,\partial_{-k+1}(y)\big]
\end{equation}
for every $x\in T_{-1}$ and $y\in T_{-k}$, where $2\leq k<i$. We have thus proven that the differential and the bracket are compatible on $T_{-1}\wedge \widetilde{T}$.

Now, let us turn to the case where both $x$ and $y$ are elements of $T_{-1}$. Item 6. of Definition \ref{wooo} induces the following identity:
\begin{equation}\label{nullhomotopy4}
\mu_0=\pi_0\circ\delta_1
\end{equation}
where every map is defined in item 5. In particular, $\mu_0$ takes values in $S^2(U_0)$. Let us check how this identity translates to $T_{-1}\wedge T_{-1}$. We will show that Equation \eqref{nullhomotopy4} implies the vanishing of the following quantity:
\begin{equation}\label{chourave3}
\Omega(x,y)\equiv\partial_{-1}\big([x,y]\big)+\eta_{-1,\Theta(s(x))}(y)+\eta_{-1,\Theta(s(y))}(x)
\end{equation}
where $x,y\in T_{-1}$. We will compute each term on the right hand side one by one, and sum them up afterwards. Let $x,y\in T_{-1}=s^{-1}V$, and $\alpha\in T_{-1}^*=s(V^*)$, then by Equation \eqref{representationshifted}, we have (aussi par equationmu):
\begin{align}
\Big\langle\alpha\,,\eta_{-1,\Theta(s(x))}(y)+\eta_{-1,\Theta(s(y))}(x)\Big\rangle_{T_{-1}}&=\Big\langle\alpha\,,s^{-1}\circ\rho^\vee_{0,\Theta(s(x))}\big(s(y)\big)+x\leftrightarrow y\Big\rangle_{s^{-1}(U_0^*)}\label{liberal9bis}\\
\footnotesize\text{by Equation \eqref{suspensiondecal}}\hspace{0.95cm}\normalsize&=\Big\langle s^{-1}(\alpha),\rho^\vee_{0,\Theta(s(x))}\big(s(y)\big)+x\leftrightarrow y\Big\rangle_{U_0^*}\\
\footnotesize\text{by definition of $\rho_0$}\hspace{1.055cm}\normalsize&=\Big\langle s^{-1}(\alpha),\eta_{V,\Theta(s(x))}\big(s(y)\big)+x\leftrightarrow y\Big\rangle_{U_0^*}\\
\footnotesize\text{by definition of $\mathcal{V}$}\hspace{1.143cm}\normalsize&=\Big\langle s^{-1}(\alpha),s(x)\bullet s(y)+ s(y)\bullet s(x)\Big\rangle_{V}\\
\footnotesize\text{by definition of $\{\,.\,,.\,\}$}\hspace{0.84cm}\normalsize&=\Big\langle s^{-1}(\alpha),2\,\big\{s(x),s(y)\big\}\Big\rangle_{V}\\
\footnotesize\text{by definition of $\mu_0$}\hspace{1.027cm}\normalsize&=-2\,\Big\langle \mu_0\big(s^{-1}(\alpha)\big),s(x)\odot s(y)\Big\rangle_{S^2(U_0^*)}\label{liberal9}
\end{align}
On the other hand, we have by Equations \eqref{liberal8bis}--\eqref{liberal8}, for $k=1$:
\begin{equation}
\Big\langle\alpha\,,\partial_{-1}\big([x,y]\big)\Big\rangle_{T_{-1}}=2\,\Big\langle \pi_0\circ\delta_1\big(s^{-1}(\alpha)\big),s(x)\odot s(y)\Big\rangle_{S^2(U_0^*)}\label{liberal10}
\end{equation}
This result could also have been obtained by using the explicit definitions of the bracket and of the differential given in Equations \eqref{bracketdebase} and \eqref{differentielledebase}, and their relationship to $\pi_0$ and $\delta_1$.

Summing Line  \eqref{liberal9} and the right hand side of \eqref{liberal10}, and dividing by 2, one obtain the following quantity:
\begin{equation}
\Big\langle  \big(\pi_0\circ\delta_1-\mu_{0}\big)\big(s^{-1}(\alpha)\big),s(x)\odot s(y)\Big\rangle_{S^2(U_{0}^*)}
\end{equation}
that vanishes by using Equation \eqref{nullhomotopy4}.
From this, since Line \eqref{liberal9} is equal to the left hand side of Line \eqref{liberal9bis}, then by summing the left hand sides of Line \eqref{liberal9bis}, and of Equation \eqref{liberal10}, we deduce the following identity: 
\begin{equation}
\big\langle \alpha\,,\Omega(x,y)\big\rangle_{T_{-1}}=0
\end{equation}
Since it holds for every $\alpha\in T_{-1}^*$ and every $x,y\in T_{-1}$, it implies that the quantity $\Omega(x,y)$ is identically zero. Then, using Equations \eqref{representationshifted}, \eqref{bracket2} and \eqref{diff0} in Equation \eqref{chourave3}, it implies that the following identity holds:
\begin{equation}\label{lignebleue3}
\partial_{-1}\big([x,y]\big)=\big[\partial_0(x),y\big]-\big[x,\partial_{0}(y)\big]
\end{equation}
for every $x,y\in T_{-1}$. We have thus proven that the differential and the bracket are compatible on $T_{-1}\wedge T_{-1}$.

Let us now turn to the last case, i.e. the compatibility of the bracket and the differential on $\mathfrak{h}\wedge T'$, since on $\mathfrak{h}\wedge \mathfrak{h}$ it is trivial. We have to show the following identity:
\begin{equation}\label{lignebleue4}
\partial_{-k+1}\big([a,x]\big)=\big[a,\partial_{-k+1}(x)\big]
\end{equation}
for every $a\in\mathfrak{h}$ and $x\in T_{-k}=s^{-1}(U_{k-1}^*)$, where $1\leq k<i$. By Equation \eqref{bracket2}, this is equivalent to showing that $\partial_{-k+1}$ is $\mathfrak{h}$-equivariant. Let us first assume that $k\geq2$ and let $a\in \mathfrak{h}$, $x\in T_{-k}$ and $\alpha\in T_{-k+1}^*=s(U_{k-2})$.
 Then, we have:
\begin{align}
\Big\langle\alpha\,,\partial_{-k+1}\big(\eta_{-k,a}(x)\big)\Big\rangle_{T_{-k+1}}&=\Big\langle s(\alpha)\,,s^{-1}\circ\partial_{-k+1}\big(\eta_{-k,a}(x)\big)\Big\rangle_{s^{-1}T_{-k+1}}\label{varphi3bis}\\
\footnotesize\text{by Equation \eqref{party}}\hspace{0.649cm}\normalsize&=\iota_{s^{-1}(\partial_{-k+1}(\eta_{-k,a}(x)))}\big(s(\alpha)\big)\\
\footnotesize\text{by Equation \eqref{bracket}}\hspace{0.649cm}\normalsize&=-\Big[\delta'_{k-1},\iota_{s^{-1}(\eta_{-k,a}(x))}\Big]\big(s(\alpha)\big)\\
\footnotesize\text{by Equation \eqref{party}}\hspace{0.649cm}\normalsize&=(-1)^k\Big\langle\delta'_{k-1}\big(s(\alpha)\big)\,,s^{-1}\big(\eta_{-k,a}(x)\big)\Big\rangle_{s^{-1}T_{-k}}\\
\footnotesize\text{by Eq. \eqref{bonjouuurd} and \eqref{representationshifted}}\hspace{0.3cm}\normalsize&=(-1)^k\Big\langle s^2\delta_{k-1}\big(s^{-1}(\alpha)\big)\,,s^{-2}\circ\rho^\vee_{k-1,a}\big(s(x)\big)\Big\rangle_{s^{-2}(U_{k-1})^*}\\
\footnotesize\text{by Equation \eqref{suspensiondecal}}\hspace{0.649cm}\normalsize&=(-1)^k\Big\langle \delta_{k-1}\big(s^{-1}(\alpha)\big)\,,\rho^\vee_{k-1,a}\big(s(x)\big)\Big\rangle_{U_{k-1}^*}\\
\footnotesize\text{by definition of $\rho^\vee_{k-1}$}\hspace{0.6cm}\normalsize&=-(-1)^k\Big\langle \rho_{k-1,a}\circ\delta_{k-1}\big(s^{-1}(\alpha)\big)\,,s(x)\Big\rangle_{U_{k-1}^*}\\
\footnotesize\text{by $\mathfrak{h}$-equivariance of $\delta$}\hspace{0.55cm}\normalsize&=-(-1)^k\Big\langle \delta_{k-1}\circ\rho_{k-2,a}\big(s^{-1}(\alpha)\big)\,,s(x)\Big\rangle_{U_{k-1}^*}\\
\footnotesize\text{by Equation \eqref{representationshifted}}\hspace{0.649cm}\normalsize&=-(-1)^k\Big\langle \delta_{k-1}\circ s^{-2}\circ s\big(\eta^\vee_{-k+1,a}(\alpha)\big)\,,s(x)\Big\rangle_{U_{k-1}^*}
\end{align}
\begin{align}
\footnotesize\text{by Equation \eqref{bonjouuurd}}\hspace{0.649cm}\normalsize&=-(-1)^k\Big\langle \delta'_{k-1}\big(s\circ\eta^\vee_{-k+1,a}(\alpha)\big)\,,s^{-1}(x)\Big\rangle_{s^{-2}(U_{k-1}^*)}\\
\footnotesize\text{by Eq. \eqref{party} and \eqref{bracket}}\hspace{0.3cm}\normalsize 
&=-\iota_{s^{-1}(\partial_{-k+1}((x))}\big(s\circ\eta^\vee_{-k+1,a}(\alpha)\big)\\
\footnotesize\text{by Equation \eqref{party}}\hspace{0.649cm}\normalsize&=-\Big\langle s\big(\eta^\vee_{-k+1,a}(\alpha)\big),s^{-1}\circ\partial_{-k+1}(x)\Big\rangle_{s^{-1}T_{-k+1}}\\
\footnotesize\text{by Equation \eqref{suspensiondecal}}\hspace{0.649cm}\normalsize&=-\Big\langle \eta^\vee_{-k+1,a}(\alpha),\partial_{-k+1}(x)\Big\rangle_{T_{-k+1}}\\
\footnotesize\text{by definition of $\eta^\vee_{-k+1}$}\hspace{0.55cm}\normalsize&=\Big\langle \alpha\,,\eta_{-k+1,a}\big(\partial_{-k+1}(x)\big)\Big\rangle_{T_{-k+1}}\label{varphi3}
\end{align}
Hence we conclude that $\partial_{-k+1}$ is $\mathfrak{h}$-equivariant. This result holds for every $k\geq2$. In the case where $k=1$, we know by Equation \eqref{diff0} that $\partial_{0}=-\Theta\circ s$. But this map is obviously $\mathfrak{h}$-equivariant, because $\Theta$ is by construction. This discussion hence proves that Equation \eqref{lignebleue4} is satisfied for every $1\leq k<i$. Thus, from Equation \eqref{lignebleue1}, \eqref{lignebleue2}, \eqref{lignebleue3} and \eqref{lignebleue4}, we deduce that  the bracket $[\,.\,,.\,]$ and the differential $\partial$ are compatible on the whole of $T$, hence concluding the desired part of the proof of Theorem \ref{prop:tensorhierarchy}.

%
%
%
\bibliographystyle{naturemag}
\bibliography{mabibliographie}

\end{document}